%% file: parallelCC.tex
\documentclass[10pt]{article}
\usepackage{setspace}
\usepackage{amsfonts,bm,xspace,amsthm,fullpage,amssymb,dsfont}
\usepackage{url,cite}
\usepackage{enumerate}
\usepackage{soul}
\usepackage{standalone}
\usepackage{verbatim}
\usepackage{algorithm}
\usepackage[numbers]{natbib}
\usepackage{paralist}
\usepackage{mdwmath}
\usepackage{amssymb}
\usepackage{amsmath}
\usepackage{amsthm}
\usepackage{xfrac}
\usepackage{mathtools}  
\usepackage{chngcntr}
\usepackage{tabulary}
\usepackage{booktabs}
\usepackage{epsfig}
\usepackage{wrapfig}
\usepackage{lipsum}
\usepackage{caption}
\usepackage{subcaption}
\usepackage{multirow}
\usepackage{algcompatible}
\usepackage{url}
\usepackage{pgf}
\usepackage{tikz}
\usepackage[utf8]{inputenc}
\usetikzlibrary{positioning}
\usetikzlibrary{shapes, arrows, decorations.markings}
\usetikzlibrary{calc}
\usetikzlibrary{fit}
\usepackage{thm-restate}

\algblockdefx{PARFOR}{ENDPARFOR}[1]
  {\textbf{for all }#1 \textbf{do in parallel}}
  {\textbf{end for}}
\algblockdefx{PWHILE}{ENDPWHILE}[1]
  {\textbf{while }#1 \textbf{do in parallel}}
  {\textbf{end while}}

\newtheorem{theo}{Theorem}

\newtheorem{prop}[theo]{Proposition}
\newtheorem{defin}{Definition}
\newtheorem{lem}[theo]{Lemma}
\newtheorem{cor}[theo]{Corollary}

\newtheorem{asm}{Assumption}

\theoremstyle{plain}
\newtheorem*{corclustering*}{Correlation Clustering}

\newcommand{\Cv}{\mathcal{C}_v}

\newcommand{\OPT}{\mathsf{OPT}}
\newcommand{\Tb}{\mathcal{T}_b}
\newcommand{\setA}{\mathcal{A}}

\newcommand{\CC}{{\it C4}}
\newcommand{\CW}{\textsc{ClusterWild!}}
\newcommand{\KC}{{\it KwikCluster}}

\renewcommand{\Pr}{\mathbb{P}}

\begin{document}

\title{\huge Parallel Correlation Clustering  on Big Graphs}

\author{
Xinghao Pan$^{\alpha,\epsilon}$, Dimitris Papailiopoulos$^{\alpha,\epsilon}$, Samet Oymak$^{\alpha,\epsilon}$ \\
 Benjamin Recht$^{\alpha,\epsilon, \sigma}$, Kannan Ramchandran$^{\epsilon}$, Michael I. Jordan$^{\alpha,\epsilon, \sigma}$\\
$^{\alpha}$AMPLab, $^{\epsilon}$EECS at UC Berkeley, $^{\sigma}$Statistics at UC Berkeley
} 

\maketitle
\begin{abstract}
Given a similarity graph between items, correlation clustering (CC) groups similar items together and dissimilar ones apart.
One of the most popular CC algorithms is \KC{}:  an algorithm that serially clusters neighborhoods of vertices, and obtains 
a $3$-approximation ratio.
Unfortunately, \KC{} in practice requires a large number of clustering rounds, a potential bottleneck for large graphs.

We present \CC{} and \CW{}, two algorithms for parallel correlation clustering that run in a polylogarithmic number of rounds and achieve nearly linear speedups, provably.
\CC{} uses concurrency control to enforce serializability of a parallel clustering process, and guarantees a $3$-approximation ratio.
\CW{} is a coordination free algorithm that abandons consistency for the benefit of better scaling; this leads to a provably small loss in the $3$-approximation ratio.

We provide extensive experimental results for both algorithms,  where we outperform the state of the art, both in terms of clustering accuracy and running time.
We show that our algorithms can cluster billion-edge graphs in under 5 seconds on 32 cores, while achieving a $15\times$ speedup.
\end{abstract}

\input{intro}

\input{algorithms}

\input{theories}

\input{bsp_discussion}

\input{prior}

\input{experiments}

\section{Conclusions and Future Directions}
We presented two parallel algorithms for correlation clustering that admit provable nearly linear speedups and approximation ratios.
Our algorithms can cluster billion-edge graphs in under 5 seconds on 32 cores, while achieving a $15\times$ speedup.
The two approaches complement each other: when \CC{} is  fast relative to \CW{}, we may prefer it for its guarantees of accuracy; and when \CW{} is accurate relative to \CC{}, we may prefer it for its speed.

Both \CC{} and \CW{} are well-suited for a distributed setup since they run for at most a polylogarithmic number of rounds.
In the future, we intend to implement our algorithms in a distributed environment, where synchronization and communication often account for the highest cost.

\section*{Acknowledgments}
XP is generously supported by a DSO National Laboratories Postgraduate Scholarship.
DP is generously supported by NSF awards CCF-1217058 and CCF-1116404 and MURI AFOSR grant 556016. 
SO is generously supported by the Simons Institute for the Theory of Computing and NSF award CCF-1217058. 
BR is generously supported by ONR awards N00014-11-1-0723 and N00014-13-1-0129, NSF awards CCF-1148243 and CCF-1217058, AFOSR award FA9550-13-1-0138, and a Sloan Research Fellowship. 
KR is generously supported by NSF award CCF-1116404 and MURI AFOSR grant 556016.
This research is supported in part by NSF CISE Expeditions Award CCF-1139158, LBNL Award 7076018, and DARPA XData Award FA8750-12-2-0331, and gifts from Amazon Web Services, Google, SAP, The Thomas and Stacey Siebel Foundation, Adatao, Adobe, Apple, Inc., Blue Goji, Bosch, C3Energy, Cisco, Cray, Cloudera, EMC2, Ericsson, Facebook, Guavus, HP, Huawei, Informatica, Intel, Microsoft, NetApp, Pivotal, Samsung, Schlumberger, Splunk, Virdata and VMware.

\bibliographystyle{unsrtnat}
{\small
\bibliography{parallelCC}
}
\appendix
\counterwithin{theo}{section}
\counterwithin{lem}{section}
\counterwithin{cor}{section}
\counterwithin{prop}{section}
\appendix
\input{theory}

\input{implementation}

\input{exptresults}

\end{document}

%% file: intro.tex
\section{Introduction}

Clustering items according to some notion of similarity is a major primitive in machine learning.
{\it Correlation clustering} serves as a basic means to achieve this goal: 
given a similarity measure between items, the goal is to group similar items together and dissimilar items apart.
In contrast to other clustering approaches, the number of clusters is not determined a priori, and good solutions aim to balance the tension between grouping all items together versus isolating them.

\begin{wrapfigure}{R}{0.5\columnwidth}
\begin{minipage}{0.5\columnwidth}
\centerline{ \includegraphics[width=.99\columnwidth]{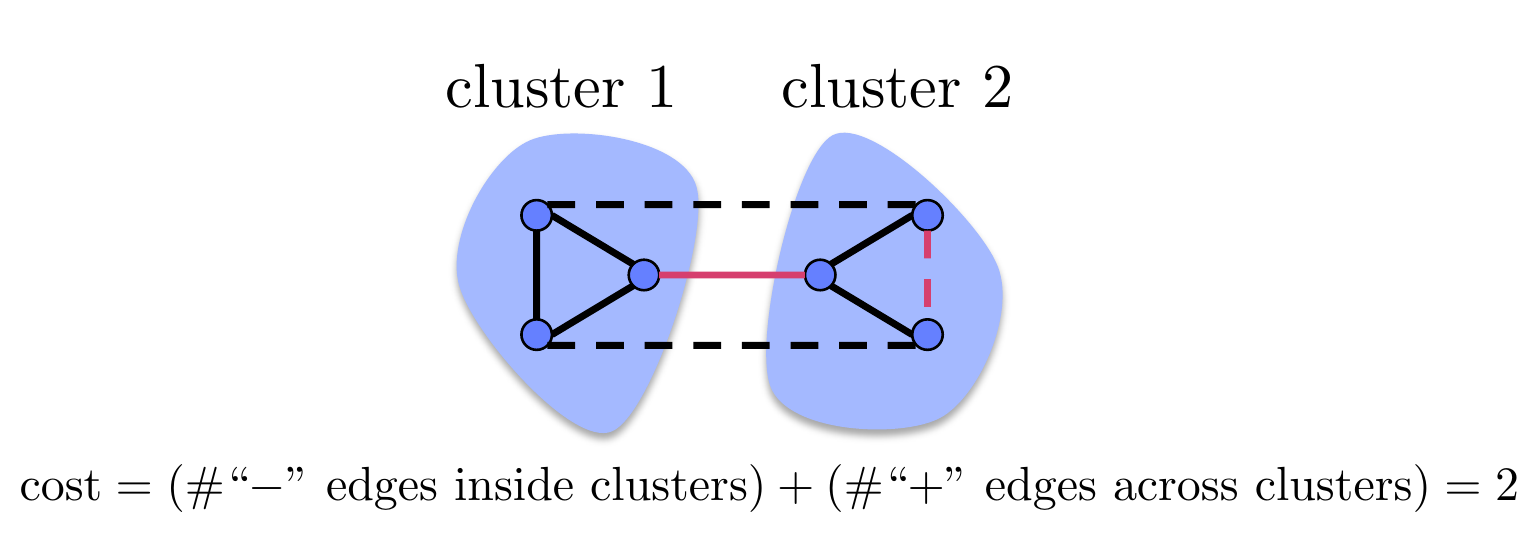}}
\caption{\footnotesize In the above graph, solid edges denote similarity and dashed dissimilarity.
The number of disagreeing edges in the above clustering is $2$; we color these edges with red.}
\label{fig:graph_example}
\end{minipage}
\end{wrapfigure}

The simplest CC variant can be described on a complete signed graph.
Our input is a graph $G$ on $n$ vertices, with $+1$ weights on edges between similar items, and $-1$ edges between dissimilar ones. 
Our goal is to generate a partition of vertices into disjoint sets that minimizes the number of {\it disagreeing edges}: this equals the number of ``$+$" edges cut by the clusters plus the number of ``$-$" edges inside the clusters.
This metric is commonly called the {\it number of disagreements}.
In Figure 1, we give a toy example of a CC instance.

Entity deduplication is the archetypal motivating example for correlation clustering, with applications in chat disentanglement, co-reference resolution, and spam detection~\cite{
elmagarmid2007duplicate,
arasu2009large,
elsner2009bounding,
hussain2013evaluation,
bonchi2014correlation,chierichetti2014correlation}.
The input is a set of entities (say, results of a keyword search), and a pairwise classifier that indicates---with some error---similarities between entities.
Two results of a keyword search might refer to the same item, but might look different if they come from different sources.
By building a similarity graph between entities and then applying CC, the hope is to cluster duplicate entities in the same group; in the context of keyword search, this  implies a more meaningful and compact list of results.
CC has been further applied to finding communities in signed networks, classifying missing edges in opinion or trust networks~\cite{yang2007community,cesa2012correlation},
gene clustering~\cite{ben1999clustering}, and consensus clustering~\cite{elsner2009bounding}.

\KC{} is the simplest CC algorithm that achieves a provable 3-approximation ratio \cite{ailon2008aggregating}, and works in the following way: pick a vertex $v$ at random (a {\it cluster center}), create a cluster for $v$ and its positive neighborhood $N(v)$ (i.e., vertices connected to $v$ with positive edges), peel these vertices and their associated edges from the graph, and repeat until all vertices are clustered. 
Beyond its theoretical guarantees, experimentally \KC{} performs well when combined with local heuristics \cite{elsner2009bounding}.

 \KC{}  seems like an inherently sequential algorithm, and
in most cases of interest it requires many peeling rounds.
This happens because a small number of vertices are clustered per round.
 This can be a bottleneck for large graphs.
Recently, there have been efforts to develop scalable variants of \KC{} \cite{bonchi2014correlation,
chierichetti2014correlation}.
In \cite{chierichetti2014correlation} a distributed peeling algorithm was presented in the context of MapReduce.
Using an elegant analysis, the authors establish a $(3+\epsilon)$-approximation in 
a polylogarithmic number of rounds.
The algorithm employs a simple step that rejects vertices that are executed in parallel but are ``conflicting"; however, we see in our experiments, this seemingly minor coordination step hinders scale-ups in a parallel core setting.
In the tutorial of \cite{bonchi2014correlation}, a sketch of a distributed algorithm was presented.
This algorithm achieves the same approximation as \KC{}, in a logarithmic number of rounds, in expectation.
However, it performs significant redundant work, per iteration, in its effort to detect in parallel which vertices should become cluster centers.

\paragraph{Our contributions} 
We present \CC{} and \CW{}, two parallel CC algorithms that in practice outperform the state of the art, both in terms of running time and clustering accuracy.
\CC{} is a parallel version of \KC{} that uses concurrency control to establish a $3$-approximation ratio.
\CW{} is a simple to implement, coordination-free algorithm that abandons consistency for the benefit of better scaling, while having a provably small loss in the 3 approximation ratio.

\CC{} achieves a $3$ approximation ratio, in a poly-logarithmic number of rounds, by enforcing consistency between concurrently running peeling threads.
Consistency is enforced using {\it concurrency control}, a notion  extensively studied for databases transactions, that was recently used to parallelize inherently sequential machine learning algorithms~\cite{pan2013optimistic}.

\CW{} is a coordination-free parallel CC algorithm that waives consistency in favor of speed. The cost that we pay is an arbitrarily small loss in \CW{}'s accuracy. 
We show that \CW{} achieves a $(3+\epsilon)\OPT+O(\epsilon\cdot n\cdot \log^2 n)$ approximation, in a poly-logarithmic number of rounds, with provable nearly linear speedups.
Our main theoretical innovation for \CW{} is analyzing the coordination-free algorithm as a serial variant of \KC{} that runs on a ``noisy" graph.

In our extensive experimental evaluation, we demonstrate that our algorithms gracefully scale up to graphs with billions of edges.
In these large graphs, our algorithms output a valid clustering in less than 5 seconds, on 32 threads, up to an order of magnitude faster than \KC{}.
We observe how, not unexpectedly, \CW{} is faster than \CC{}, and quite surprisingly, abandoning coordination in this parallel setting, only amounts to a $1$\% of relative loss in the clustering accuracy.
Furthermore, we compare against state of the art parallel CC algorithms, showing that we consistently outperform these algorithms in terms of both running time and clustering accuracy.

\paragraph{Notation} $G$ denotes a graph with $n$ vertices and $m$ edges.
$G$ is complete and only has $\pm1$ edges.
We denote by $d_v$ the positive degree of a vertex, i.e., the number of vertices connected to $v$ with positive edges.
$\Delta$ denotes the positive maximum degree of $G$, and $N(v)$ denotes the positive neighborhood of $v$; moreover, let $C_v = \{v,N(v)\}$.
Two vertices $u$, $v$ are neighbors in $G$ if $u \in N(v)$ and vice versa.
We denote by $\pi$ a permutation of $\{1,\ldots, n\}$.

%% file: algorithms.tex
\section{Two Parallel Algorithms for Correlation Clustering}
The formal definition of correlation clustering is given below.
\begin{corclustering*}
Given a graph $G$ on $n$ vertices, partition the vertices into an arbitrary number $k$ of disjoint subsets $\mathcal{C}_1,\ldots,\mathcal{C}_k$ such that the sum of negative edges within the subsets plus the sum of positive edges across the subsets is minimized:
\begin{equation}
\OPT = 
\min_{1\le k\le n}
\min_{
{\tiny
\begin{smallmatrix}
\mathcal{C}_i \cap \mathcal{C}_j = 0, \forall i\ne j\\
\cup_{i=1}^k \mathcal{C}_i = \{1,\ldots, n\}
\end{smallmatrix}
}
} \sum_{i=1}^k E^-(\mathcal{C}_i,\mathcal{C}_i) +\sum_{i=1}^k\sum_{j=i+1}^k E^+(\mathcal{C}_i,\mathcal{C}_j)  \nonumber
\end{equation}
where $E^+$ and $E^-$ are the sets of positive and negative edges in $G$.
\end{corclustering*}

\KC{} is a remarkably simple algorithm that approximately solves the above combinatorial problem, and operates as follows. 
A random vertex $v$ is picked, a cluster $C_v$ is created with $v$ and its positive neighborhood, then the vertices in $C_v$ are peeled from the graph, and this process is repeated until all vertices are clustered.
\begin{wrapfigure}{R}{0.5\columnwidth}
\vspace{-0.8cm}
\begin{minipage}{0.45\columnwidth}
\begin{algorithm}[H]
   \caption{\KC{} with $\pi$}
          {
\begin{algorithmic}[1]
\STATE $\pi$ =  a random permutation of $\{1,\ldots, n\}$
\WHILE{$V\ne \emptyset $} 
\STATE select the vertex $v$ indexed by $\pi(1)$
\STATE $C_v = \{v, N(v)\}$
\STATE Remove clustered vertices from $G$ and $\pi$
\ENDWHILE
\end{algorithmic}
}
   \label{alg:KCpi}
 \end{algorithm}
\end{minipage}
\end{wrapfigure}
\KC{} can be equivalently executed, as noted by \cite{bonchi2014correlation}, if we substitute the random choice of a vertex per peeling round, with a random order $\pi$ preassigned to vertices, (see Alg.~\ref{alg:KCpi}).
That is, select a random permutation on vertices, then peel the vertex indexed by $\pi(1)$, and its neighbors. Remove from $\pi$ the vertices in $C_v$ and repeat this process.
Having an order among vertices makes the discussion of parallel algorithms more convenient.

\subsection{\CC{}: Parallel CC using Concurrency Control}

Suppose we now wish to run a parallel version of \KC{}, say on two threads: one thread picks vertex $v$ indexed by $\pi(1)$ and the other thread picks $u$ indexed by $\pi(2)$, concurrently.
Can both vertices be cluster centers? They can, if and only if they are not neighbors in $G$.
If $v$ and $u$ are connected with a positive edge, then the vertex with the smallest order wins.
This is our  {\it concurency rule no. 1}.
Now, assume that $v$ and $u$ are not neighbors in $G$, and  both $v$ and $u$ become cluster centers.
Moreover, assume that $v$ and $u$ have a common, unclustered neighbor, say $w$: should $w$ be clustered with $v$, or $u$?
We need to follow what would happen with \KC{} in Alg.~\ref{alg:KCpi}: $w$ will go with the vertex that has the smallest permutation number, in this case $v$.
This is {\it concurency rule  no. 2}.
Following the above simple rules, we develop \CC{}, our serializable parallel CC algorithm.
Since, \CC{} constructs the same clusters as \KC{} (for a given ordering $\pi$), it inherits its $3$ approximation by design.
The above idea of identifying the cluster centers in rounds was first used in \cite{blelloch2012greedy} to obtain a parallel algorithm for maximal independent set (MIS).

\CC{}, shown as Alg.~\ref{alg:CC}, starts by assigning a random permutation $\pi$ to the vertices, it then samples an active set $\setA$ of $\frac{n}{\Delta}$ unclustered vertices; this sample is taken from the prefix of $\pi$.
After sampling $\setA$, each of the $P$ threads picks a vertex with the smallest order in $\setA$, then checks if that vertex can become a cluster center.
We  first enforce  {\it concurrency rule no. 1}: adjacent vertices cannot be cluster centers at the same time.
\CC{} enforces it by making each thread check the neighbors of the vertex, say $v$, that is picked from $\setA$.
A thread will check in \texttt{attemptCluster} whether its vertex $v$ has any preceding neighbors (according to $\pi$) that are cluster centers.
If there are none, it will go ahead and label $v$ as cluster center, and proceed with creating a cluster.
If a preceding neighbor of $v$ is a cluster center, then $v$ is labeled as not being a cluster center.
If a preceding neighbor of $v$, call it $u$, has not yet received a label (i.e., $u$ is currently being processed and is not yet labeled as cluster center or not), then the thread processing $v$, will wait on $u$ to receive a label.
The major technical detail is in showing that this wait time is bounded; we show that no more than $O(\log n)$ threads can be in conflict at the same time, using a new subgraph sampling lemma \cite{krivelevich2014phase}.
Since \CC{} is serializable, it has to respect  {\it concurrency rule no. 2}: if a vertex $u$ is adjacent to two cluster centers, then it gets assigned to the one with smaller permutation order.
This is accomplished in  \texttt{createCluster}.
After processing all vertices in $\setA$, all threads are synchronized in bulk, the clustered vertices are removed, a new active set is sampled, and the same process is repeated until everything has been clustered.
In the following section, we present the theoretical guarantees for \CC{}.

\begin{center}
\begin{minipage}{0.45\columnwidth}
\begin{algorithm}[H]
   \caption{\CC{}  \& \CW{}}
          {
\begin{algorithmic}[1]
\STATE {\bf Input}: $G, \epsilon$
\STATE $\text{clusterID}(1) = \ldots = \text{clusterID}(n) = \infty$
\STATE $\pi=$ a random permutation of $\{1,\ldots, n\}$
\WHILE{$V\ne \emptyset $} 
\STATE $\Delta=$ maximum vertex degree in $G(V)$
\STATE $\setA=$ the first $\epsilon\cdot \frac{n}{\Delta}$ vertices in $V[\pi]$.
\PWHILE{$\setA\neq \emptyset$}
\STATE$v = \text{first element in }\setA$
\STATE $\setA = \setA-\{v\}$
\IF{ {\color{blue}\CC{}}} 
\STATE \texttt{attemptCluster}($v$) 
\ELSIF{ {\color{red}\CW{}}} 
\STATE \texttt{createCluster}($v$)  
\ENDIF
\ENDPWHILE
\STATE Remove clustered vertices from $V$ and $\pi$

\ENDWHILE
 \STATE {\bf Output:} $\{\text{clusterID}(1),\ldots, \text{clusterID}(n)\}$.
\end{algorithmic}
}
   \label{alg:CC}
 \end{algorithm}
\end{minipage}
\begin{minipage}{0.5\columnwidth}
   \vspace{0.50cm}
   \hrule
\vspace{0.1cm}
          {
\begin{algorithmic}[0]
\STATE\hspace{-0.2cm}\texttt{createCluster($v$)}:
\STATE $\text{clusterID}(v) = \pi(v)$ 
\FOR{$u \in \Gamma(v)\setminus \setA$}
\STATE $\text{clusterID}(u) = \min(\text{clusterID}(u), \pi(v))$ 
\ENDFOR
\end{algorithmic}
\vspace{0.1cm}
\hrule
\vspace{0.1cm}
\begin{algorithmic}[0]
\STATE \hspace{-0.3cm}\texttt{attemptCluster}($v$):
\IF{ $\text{clusterID}(u)=\infty$ {\bf and}  $\texttt{isCenter}(v)$}
\STATE \texttt{createCluster($v$)}
\ENDIF
\end{algorithmic}
\vspace{0.1cm}
\hrule
\vspace{0.1cm}
\begin{algorithmic}[0]
\STATE \hspace{-0.3cm}\texttt{isCenter}($v$):
\FOR{$u \in \Gamma(v)$} {\color{gray}//check friends (in order of $\pi$)}
\IF{ $\pi(u) < \pi(v)$ } {\color{gray}//if they precede you, wait}
\STATE {\bf wait} until $\text{clusterID}(u)\ne \infty$ {\color{gray}//till clustered}
\IF{ \texttt{isCenter}($u$) }
\STATE {\bf return} $0$  {\color{gray}//a friend is center, so you can't be}
\ENDIF
\ENDIF
\ENDFOR
\STATE {\bf return} $1$ {\color{gray}//no earlier friend are centers, so you are}
\end{algorithmic}
\hrule
}
\end{minipage}
\end{center}

\subsection{\CW{}: Coordination-free Correlation Clustering}
\CW{} speeds up computation 
by ignoring the first concurrency rule. 
It uniformly samples unclustered vertices, and builds clusters around {\it all of them}, without respecting the rule that cluster centers cannot be neighbors in $G$.
In \CW{}, threads bypass the \texttt{attemptCluster} routine; this eliminates the ``waiting" part of \CC{}.
\CW{} samples a set $\setA$ of vertices from the prefix of $\pi$.
Each thread picks the first ordered vertex remaining in  $\setA$, and using that vertex as a cluster center, it creates a cluster around it.
It peels away the clustered vertices and repeats the same process, on the next remaining vertex in $\setA$.
At the end of processing all vertices in $\setA$, all threads are synchronized in bulk, the clustered vertices are removed, a new active set is sampled, and the parallel clustering is repeated.
A careful analysis along the lines of \cite{chierichetti2014correlation} shows that the number of rounds (i.e., bulk synchronization steps) is only poly-logarithmic.

Quite unsurprisingly, \CW{} is faster than \CC{}.
Interestingly, abandoning consistency does not incur much loss in the approximation ratio.
We show how the error introduced in the accuracy of the solution can be bounded.
We characterize this error theoretically, and show that in practice it only translates to only a relative $1$\% loss in the objective.
The main intuition of why \CW{} does not introduce too much error is that the chance of two randomly selected vertices being neighbors is small, hence the concurrency rules are infrequently broken.

%% file: theories.tex
\section{Theoretical Guarantees}

In this section, we bound the number of rounds required for each algorithms, and establish the theoretical speedup one can obtain with $P$ parallel threads.
We proceed to present our approximation guarantees.
We would like to remind the reader that---as in relevant literature---we consider graphs that are complete, signed, and unweighted.
The omitted proofs can be found in the Appendix.

\subsection{Number of rounds and running time}
Our analysis follows those of  \cite{blelloch2012greedy} and \cite{chierichetti2014correlation}.
The main idea is to track how fast the maximum degree decreases in the remaining graph at the end of each round.

\begin{restatable}{lem}{lemnumrounds} \label{lem:numrounds}
\CC{} and \CW{} terminate after $O\left(\frac{1}{\epsilon}\log n \cdot \log\Delta\right)$ rounds w.h.p.
\label{lem:rounds}
\end{restatable}
We now analyze the running time of both algorithms under a simplified BSP model.
The main idea is that the the running time of each ``super step" (i.e., round) is determined by the ``straggling" thread (i.e., the one that gets assigned the most amount of work), plus the time needed for synchronization at the end of each round.
\begin{asm}
We assume that threads operate asynchronously within a round, and synchronize at the end of a round.
A memory cell can be written/read concurrently by multiple threads.
The time spent per round of the algorithm is proportional to the time of the slowest thread.
The cost of thread synchronization at the end of each batch takes time $O(P)$, where $P$ is the number of threads.
The total computation cost is proportional to the sum of the time spent for all rounds, plus the time spent during the bulk synchronization step.
\end{asm}
Under this simplified model, we show that both algorithms obtain nearly linear speedup, with \CW{} being faster than \CC{}, precisely due to lack of coordination.
Our main tool for analyzing \CC{} is a recent graph-theoretic result (Theorem 1 in  \cite{krivelevich2014phase}), which guarantees that if one samples an $O(n/\Delta)$ subset of vertices in a graph, the sampled subgraph has a connected component of size at most $O(\log n)$.
Combining the above, in the appendix we show the following result.

\begin{restatable}{theo}{thmrunningtime}
The theoretical running time of \CC{}, on $P$ cores and $\epsilon=1/2$, is upper bounded by
$$O\left( \left(\frac{m+n\log^2 n}{P}+P\right)\log n\cdot\log\Delta\right)$$
as long as the number of cores $P$ is smaller than  $\min_i \frac{n_i}{2\Delta_i}$, where $\frac{n_i}{2\Delta_i}$ is the size of the batch in the $i$-th round of each algorithm.
The running time of \CW{} on $P$ cores is upper bounded by
$$O\left( \left(\frac{m+n}{P}+P\right)\frac{\log n\cdot\log\Delta}{\epsilon^2}\right)$$
for any constant $\epsilon>0$.
\end{restatable}

\subsection{Approximation ratio}
We now proceed with establishing the approximation ratios of \CC{} and \CW{}.
\subsubsection{\CC{} is serializable}
It is straightforward that \CC{} obtains precisely the same approximation ratio as \KC{}.
One has to simply show that for any permutation $\pi$, \KC{} and \CC{} will output the same clustering.
This is indeed true, as the two simple concurrency rules mentioned in the previous section are sufficient for \CC{} to be equivalent to \KC{}.
\begin{theo} 
\CC{} achieves a $3$ approximation ratio, in expectation.
\end{theo}
\subsubsection{\CW{} as a serial procedure on a noisy graph}
Analyzing \CW{} is a bit more involved.
Our guarantees are based on the fact that \CW{} can be treated {\it as if} one was running a peeling algorithm on a ``noisy" graph.
Since adjacent active vertices can still become cluster centers in \CW{}, one can view the edges between them as ``deleted," by a somewhat unconventional adversary.
We analyze this new, noisy graph and establish our theoretical result.

\begin{theo} 
\CW{} achieves a $\left(3+\epsilon\right)\cdot \OPT+O(\epsilon\cdot n\cdot \log^2n)$ approximation, in expectation.
\end{theo}

We provide a sketch of the proof, and delegate the details to the appendix.
Since \CW{} ignores the edges among active vertices, we treat these edges as deleted.
In our main result, we quantify the loss of clustering accuracy that is caused by ignoring these edges.
Before we proceed, we define {\it bad triangles}, a combinatorial structure that is used to measure the clustering quality of a peeling algorithm.

\begin{defin}
A {\it bad triangle} in $G$ is a set of three vertices, such that two pairs are joined with a positive edge, and one pair is joined with a negative edge. 
Let $\mathcal{T}_b$ denote the set of bad triangles in $G$.
\end{defin}

To quantify the cost of \CW{}, we make the below observation.
\begin{restatable}{lem}{lemalgocost}\label{lem:algocost}
The cost of any greedy algorithm that picks a vertex $v$  (irrespective of the sampling order), creates $\mathcal{C}_v$, peels it away and repeats, is equal to the number of bad triangles adjacent to each cluster center $v$.
\end{restatable}

\begin{restatable}{lem}{lemcwcost}\label{lem:cwcost}
Let $\hat G$ denote the random graph induced by deleting all edges between active vertices per round, for a given run of \CW{}, and let  $\tau_{\text{new}}$ denote the number of additional bad triangles that $\hat G$ has compared to $G$.
Then, the expected cost of \CW{}  can be upper bounded as
$$\mathbb{E}\left\{\sum_{t\in\Tb} {\bf 1}_{\mathcal{P}_t} +\tau_{\text{new}}\right\},$$
where $\mathcal{P}_t$ is the event that triangle $t$, with end points $i,j,k$, is bad, and at least one of its end points becomes active, while $t$ is still part of the original unclustered graph.
\end{restatable}

We provide the proof for the above two lemmas in the Appendix.
We continue with bounding the second term $\mathbb{E}\{\tau_{\text{new}}\}$ in the bound of Lemma~\ref{lem:cwcost}, by considering the number of new bad triangles $\tau_{\text{new},i}$ created at each round $i$ (in the following $\setA_i$, denotes the set of active vertices at round $i$):

$$\mathbb{E}\left\{\tau_{\text{new},i}\right\}
\le   \sum_{(u,v) \in E_i^+} \Pr(u,v\in \mathcal{A}_i) \cdot |N_i(u) \cup N_i(v)| 
\le \sum_{(u,v) \in E_i^+} \left(\frac{\epsilon}{\Delta_i}\right)^2 \cdot 2 \cdot \Delta_i
\quad\le\quad 2\cdot \epsilon^2\cdot \frac{E_i}{\Delta_i}
\quad\le\quad 2\cdot \epsilon^2 \cdot n$$
where $E_i^+$ is the set of remaining positive and $N_i(v)$ the neighborhood of vertex $v$ at round $i$, the second inequality is due to the fact that the size of the neighborhoods is upper bounded by $\Delta_i$, the maximum positive degree at round $i$, and the probability bound is true since we are sampling $\frac{n_i}{\Delta_i}$ vertices without replacement from a total of $n_i$, the number of unclustered vertices at round $i$; the final inequality is true since $E_i\le n\cdot \Delta_i$.
Using the result that \CW{} terminates after at most $O(\frac{1}{\epsilon}\log n \log \Delta)$ rounds, we get that\footnote{We skip the constants to simplify the presentation; however they are all smaller than 10.} 
$$\mathbb{E}\left\{\tau_{\text{new}}\right\}\le O(\epsilon\cdot n\cdot \log^2n ).$$
We are left to bound  
$$\mathbb{E}\left\{\sum_{t\in\Tb} {\bf 1}_{\mathcal{P}_t} \right\} =\sum_{t\in\Tb} p_t .$$
To do that we use the following lemma.
\begin{lem}
\label{lem:alphaapprox}
If $p_t$ satisfies 
$$\forall e,\;\sum_{t: e \subset t \in \mathcal{T}_b} \frac{p_t}{\alpha} \leq 1,$$
then,
$$\sum_{t \in \mathcal{T}_b} p_t \le \alpha \cdot OPT.$$
\end{lem}
\begin{proof}
Let $\mathcal{B}_*$ be one (of the possibly many) sets of edges that attribute a $+1$ in the cost of an optimal algorithm.
Then,
$$\OPT  = \sum_{e\in \mathcal{B}^*} 1 
\ge \sum_{e\in \mathcal{B}^*} \sum_{t: e \subset t \in \mathcal{T}_b} \frac{p_t}{\alpha}
= \sum_{t\in \mathcal{T}_b} \underbrace{|B_*\cap t|}_{\ge 1} \frac{p_t}{\alpha}
\ge \sum_{t\in \mathcal{T}_b} \frac{p_t}{\alpha}.
$$
\end{proof}
Now, as with \cite{chierichetti2014correlation}, we will simply have to bound the expectation of the bad triangles, adjacent to an edge $(u,v)$: 
$$\sum_{t:\{u,v\}\subset t \in \mathcal{T}_b} {\bf 1}_{\mathcal{P}_t}.$$
Let $\mathcal{S}_{u,v} = \bigcup_{\{u,v\}\subset t \in \mathcal{T}_b} t$ be the union of the sets of nodes of the bad triangles that contain both vertices $u$ and $v$.
Observe that if some $w \in S \backslash \{u,v\}$ becomes active before $u$ and $v$, then a cost of $1$ (i.e., the cost of the bad triangle $\{u,v,w\}$) is incurred.
On the other hand, if either $u$ or $v$, or both, are selected as pivots in some round, then $\mathcal{C}_{u,v} $ can be as high as $|S|-2$, i.e., at most equal to all bad triangles containing the edge $\{u,v\}$.
Let $A_{uv} = \{\text{$u$ or $v$ are activated before any other vertices in $S_{u,v}$}\}$.
Then,
\begin{align*}
\mathbb{E}\left[\mathcal{C}_{u,v}\right]
&=\mathbb{E}\left[\left.\mathcal{C}_{u,v} \right| A_{u,v}\right] \cdot \Pr(A_{u,v})
+ \mathbb{E}\left[\left.\mathcal{C}_{u,v} \right| A_{u,v}^C\right] \cdot \Pr(A_{u,v}^C)\\
&\leq 1 + (|S|-2) \cdot \Pr(\{u,v\} \cap \mathcal{A} \neq \emptyset | \mathcal{S} \cap \mathcal{A} \neq \emptyset)\\
&\leq 1 + 2|S| \cdot \Pr(v\cap \mathcal{A} \neq \emptyset | \mathcal{S} \cap \mathcal{A} \neq \emptyset)
\end{align*}
where the last inequality is obtained by a union bound over $u$ and $v$.
We now bound the following probability:
\begin{align*}
\Pr\left\{
\left.v \cap \mathcal{A} \neq\emptyset\right|  \mathcal{S}\cap\mathcal{A} \neq\emptyset
\right\} 
= \frac{\Pr\left\{v\in\mathcal{A}\right\}\cdot \Pr\left\{\mathcal{S}\cap\mathcal{A} \neq\emptyset \left|v\in\mathcal{A} \right.\right\}}{\Pr\left\{\mathcal{S}\cap\mathcal{A} \neq\emptyset \right\}}
= \frac{\Pr\left\{v\in\mathcal{A}\right\}}{\Pr\left\{\mathcal{S}\cap\mathcal{A} \neq\emptyset \right\}}
=\frac{\Pr\left\{v\in\mathcal{A}\right\}}{1-\Pr\left\{\mathcal{S}\cap\mathcal{A} =\emptyset\right\}}.
\end{align*}
Observe that 
$\Pr\left\{v\in\mathcal{A}\right\} = \frac{\epsilon}{\Delta}$, hence we need to upper bound $\Pr\left\{\mathcal{S}\cap\mathcal{A} =\emptyset\right\}$.
The probability, per round, that no positive neighbors in $\mathcal{S}$ become activated is upper bounded by
\begin{align*}
\frac{{n-|\mathcal{S}|\choose P}}{{n\choose P}}
&= \prod_{t=1}^{|\mathcal{S}|}\left(1- \frac{P}{n-|\mathcal{S}|+t}\right)\le \left(1- \frac{P}{n}\right)^{|\mathcal{S}|}\\
&= \left[\left(1- \frac{P}{n}\right)^{n/P}\right]^{|\mathcal{S}|n/P}
\le \left(\frac{1}{e}\right)^{|\mathcal{S}|n/P}.
\end{align*}
Hence, we obtain the following bound
\begin{align*}
|\mathcal{S}|\Pr\left\{
\left.v \cap \mathcal{A} \neq\emptyset\right|  \mathcal{S}\cap\mathcal{A} \neq\emptyset
\right\} 
&\le \frac{\epsilon\cdot \sfrac{|\mathcal{S}|}{\Delta}}{1-e^{-\epsilon\cdot \sfrac{|\mathcal{S}|}{\Delta}}}.
\end{align*}
We now know that $|\mathcal{S}|\le 2\cdot\Delta+2$ and also $\epsilon\le 1$.
Then, 
$$0\le \epsilon\cdot \frac{|\mathcal{S}|}{\Delta}\le\epsilon\cdot \left(2+\frac{2}{\Delta}\right)\le 4. $$
Hence, we have 
$$\mathbb{E}(\mathcal{C}_{u,v}) \leq 1 + 2\cdot \frac{4\epsilon}{1-\exp\{-4\epsilon\}}.$$
The overall expectation is then bounded by
$$
\mathbb{E}\left\{\sum_{t\in\Tb} {\bf 1}_{\mathcal{P}_t} +\tau_{\text{new}}\right\}
\le \left(1+2\cdot \frac{4\cdot \epsilon }{1-e^{-4\cdot \epsilon}}\right)\cdot \OPT+O(\epsilon\cdot n\cdot \log^2 n)
\le \left(3+\epsilon\right)\cdot \OPT+O(\epsilon\cdot n\cdot \log^2 n)
$$
which establishes our approximation ratio for \CW{}.

%% file: bsp_discussion.tex
\begin{wrapfigure}{RH}{0.5\columnwidth}
\begin{minipage}{0.5\columnwidth}
\begin{algorithm}[H]
   \caption{\CC{}  \& \CW{}\\ \centerline{(asynchronous execution)}}
          \scriptsize{
\begin{algorithmic}[1]
\STATE {\bf Input}: $G$
\STATE $\text{clusterID}(1) = \ldots = \text{clusterID}(n) = \infty$
\STATE $\pi=$ a random permutation of $\{1,\ldots, n\}$
\WHILE{$V\ne \emptyset $} 
\STATE$v = \text{first element in }V$
\STATE $V = V-\{v\}$
\IF{ {\color{blue}\CC{}}} {\color{blue}// concurrency control}
\STATE \texttt{attemptCluster}($v$) 
\ELSIF{ {\color{red}\CW{}}}  {\color{red}// coordination free}
\STATE \texttt{createCluster}($v$)  
\ENDIF
\STATE Remove clustered vertices from $V$ and $\pi$

\ENDWHILE
 \STATE {\bf Output:} $\{\text{clusterID}(1),\ldots, \text{clusterID}(n)\}$.
\end{algorithmic}
}
   \label{alg:CCasync}
 \end{algorithm}
 \end{minipage}
 \end{wrapfigure}

\subsection{BSP Algorithms as a Proxy for Asynchronous Algorithms}
\label{sec:bsp}
We would like to note that the analysis under the BSP model can be a useful proxy for the performance of completely asynchronous variants of our algorithms.
Specifically, see Alg.~\ref{alg:CCasync}, where we remove the synchronization barriers.

The only difference between the asynchronous execution in Alg.~\ref{alg:CCasync}, compared to Alg.~\ref{alg:CC}, is the complete lack of bulk synchronization, at the end of the processing of each active set $\setA$.
Although the analysis of the BSP variants of the algorithms is tractable, unfortunately analyzing precisely the speedup of the asynchronous \CC{} and the approximation guarantees for the asynchronous \CW{} is challenging.
However, in our experimental section we test the completely asynchronous algorithms against the BSP algorithms of the previous section, and observe that they perform quite similarly both in terms of accuracy of clustering, and running times.

%% file: prior.tex
\section{Related Work}

Correlation clustering was formally introduced by Bansal et al.~\cite{bansal2002correlation}.
In the general case, minimizing disagreements is NP-hard and hard to approximate within an arbitrarily small constant (APX-hard)~\cite{bansal2002correlation,charikar2003clustering}.
There are two variations of the problem: 
{\it i)} CC on complete graphs where all edges are present and all weights are $\pm1$, and 
{\it ii)} CC on general graphs with arbitrary edge weights.
Both problems are hard, however the general graph setup seems fundamentally harder.
The best known approximation ratio for the latter is $O(\log n)$, and a reduction to the minimum multicut problem indicates that any improvement to that requires fundamental breakthroughs in theoretical algorithms \cite{demaine2006correlation}.

In the case of complete unweighted graphs, a long series of results establishes 
a 2.5 approximation via a rounded linear program (LP)~\cite{ailon2008aggregating}.
A recent result establishes a 2.06 approximation using an elegant rounding to the same LP relaxation~\cite{chawla2015near}. 
By avoiding the expensive LP, and by just using the rounding procedure of \cite{ailon2008aggregating} as a basis for a greedy algorithm yields \KC{}:
a 3 approximation for CC on complete unweighted graphs.

Variations of the cost metric for CC change the algorithmic landscape: maximizing agreements (the dual measure of disagreements) \cite{bansal2002correlation,swamy2004correlation,giotis2006correlation}, or maximizing the difference between the number of  agreements and disagreements~\cite{charikar2004maximizing, alon2006quadratic}, come with different hardness and approximation results.
There are also several variants: chromatic CC \cite{bonchi2012chromatic},
overlapping CC \cite{bonchi2011overlapping},
or CC with small number of clusters and added constraints that are suitable for biology applications \cite{puleo2014correlation}.

The way \CC{} finds the cluster centers can be seen as a variation of the MIS algorithm of \cite{blelloch2012greedy}; 
the main difference is that in our case, we ``passively" detect the MIS, by locking on memory variables, and by waiting on preceding ordered threads.
This means, that a vertex only ``pushes" its cluster ID and status (cluster center/clustered/unclustered) to its neighbors, versus ``pulling" (or asking) for its neighbors' cluster status.
This saves a substantial amount of computational effort.
A sketch of the idea of using parallel MIS algorithms for CC was presented in \cite{bonchi2014correlation}, where the authors suggest using Luby's algorithm for finding an MIS, and then using the MIS vertices as cluster centers.
However, a closer look on this approach reveals that there is fundamentally more work need to be done to cluster the vertices.

%% file: experiments.tex
\section{Experiments}
Our parallel algorithms were all implemented in Scala---we defer a full discussion of the implementation details to Appendix \ref{app:exptresults}.
We ran all our experiments on Amazon EC2's r3.8xlarge (32 vCPUs, 244Gb memory) instances, using 1-32 threads.
\begin{table}[h]
\centering\scriptsize
\begin{tabular}{|c|r|r|c|}\hline
Graph & \# vertices & \# edges & Description \\\hline\hline
DBLP-2011    &     986,324 &      6,707,236 & 2011 DBLP co-authorship network              \cite{BoVWFI, BRSLLP, BCSU3}. \\\hline
ENWiki-2013  &   4,206,785 &    101,355,853 & 2013 link graph of English part of Wikipedia \cite{BoVWFI, BRSLLP, BCSU3}. \\\hline
UK-2005      &  39,459,925 &    921,345,078 & 2005 crawl of the .uk domain                 \cite{BoVWFI, BRSLLP, BCSU3}. \\\hline
IT-2004      &  41,291,594 &  1,135,718,909 & 2004 crawl of the .it domain                 \cite{BoVWFI, BRSLLP, BCSU3}. \\\hline
WebBase-2001 & 118,142,155 &  1,019,903,190 & 2001 crawl by WebBase crawler                \cite{BoVWFI, BRSLLP, BCSU3}. \\\hline
\end{tabular}
\caption{\footnotesize Graphs used in the evaluation of our parallel algorithms.}
\label{tab:graphstats}
\end{table}
The real graphs listed in Table \ref{tab:graphstats} were each tested with 100 different random $\pi$ orderings.
We measured the runtimes, speedups (ratio of runtime on 1 thread to runtime on $p$ threads), and objective values obtained by our parallel algorithms.
For comparison, we also implemented the algorithm presented in \cite{chierichetti2014correlation}, which we denote as CDK for short\footnote{
CDK was only tested on the smaller graphs of DBLP-2011 and ENWiki-2013,
because CDK was prohibitively slow, often 2-3 orders of magnitude slower than \CC{}, \CW{}, and even serial \KC{}.
}.
Values of $\epsilon = 0.1, 0.5, 0.9$ were used for \CC{} BSP, \CW{} BSP and CDK.
In the interest of space, we present only representative plots of our results;
full results are given in our appendix.

\begin{figure}[t!]
  \centering
  \begin{tabular}{cc}
    \begin{subfigure}[b]{0.49\textwidth}
      \includegraphics[width=185pt]{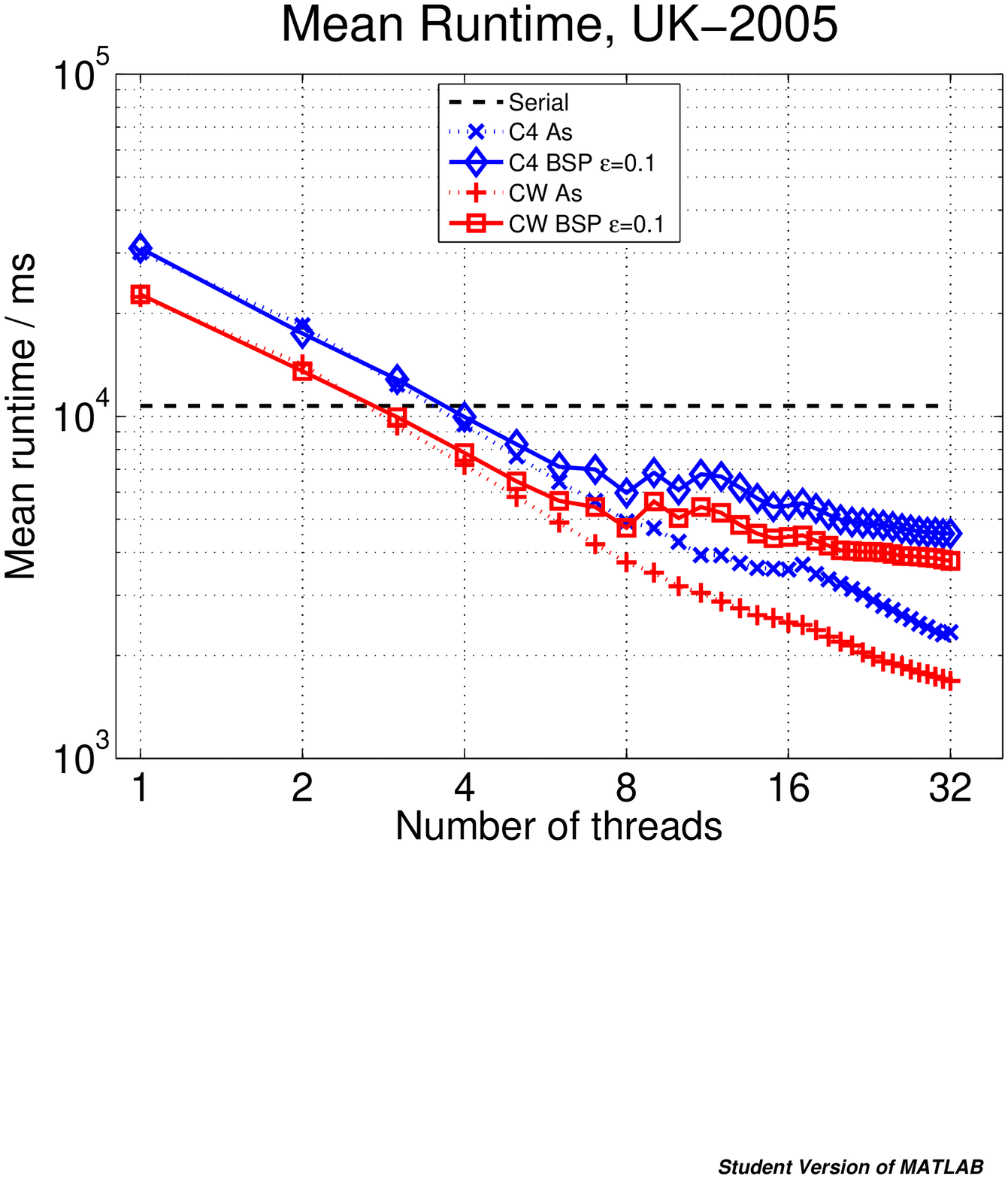}
      \caption{ Mean runtimes, UK-2005, $\epsilon = 0.1$}
      \label{fig:runtimes_uk05_ave_01}
    \end{subfigure} &
    \begin{subfigure}[b]{0.49\textwidth}
      \includegraphics[width=185pt]{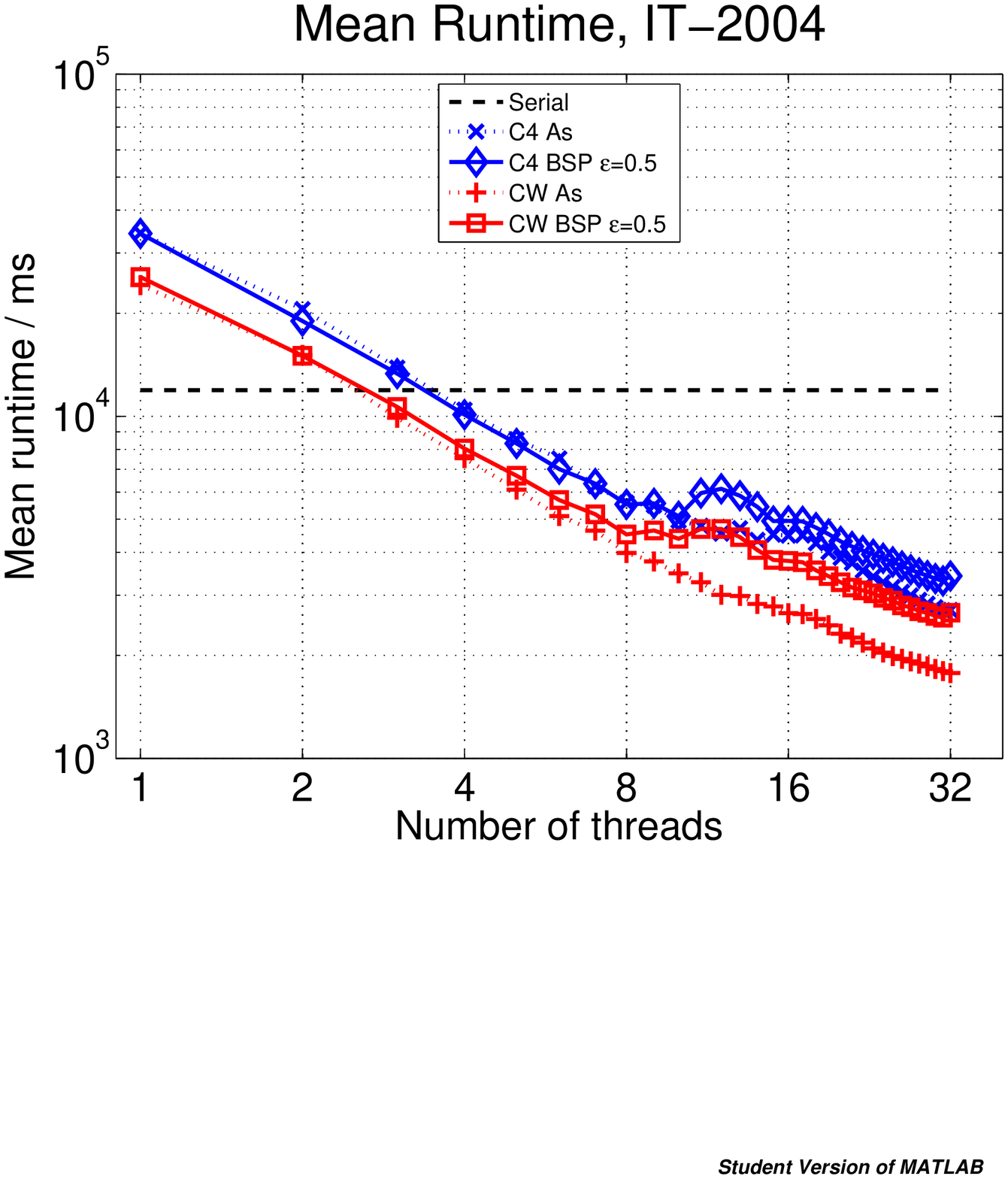}
      \caption{ Mean runtimes, IT-2004, $\epsilon = 0.5$}
      \label{fig:runtimes_it04_ave_05}
    \end{subfigure} \\
    \begin{subfigure}[b]{0.49\textwidth}
      \includegraphics[width=185pt]{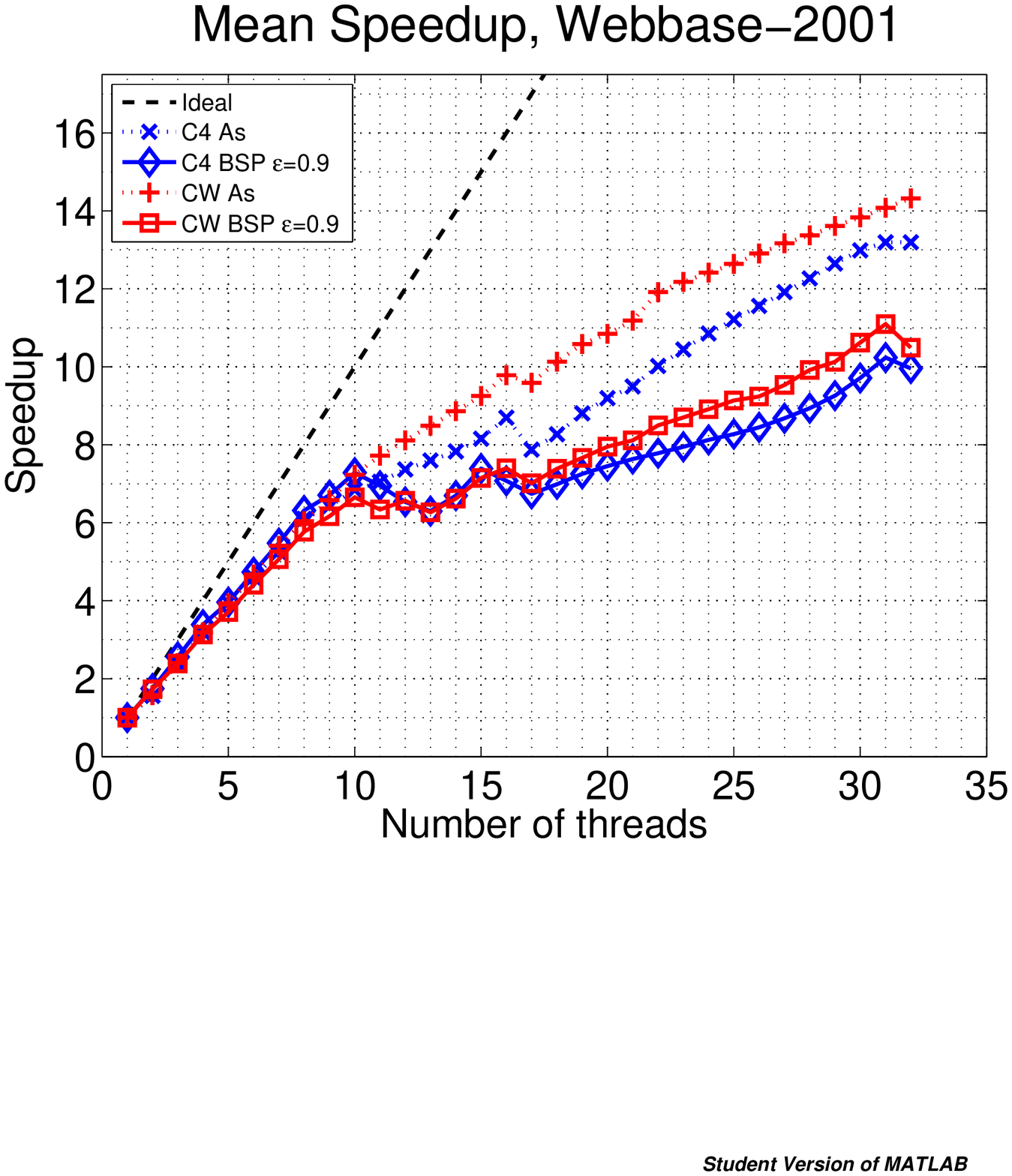}
      \caption{ Mean speedup, WebBase, $\epsilon = 0.9$}
      \label{fig:speedups_wb01_09}
    \end{subfigure} &
    \begin{subfigure}[b]{0.49\textwidth}
      \includegraphics[width=185pt]{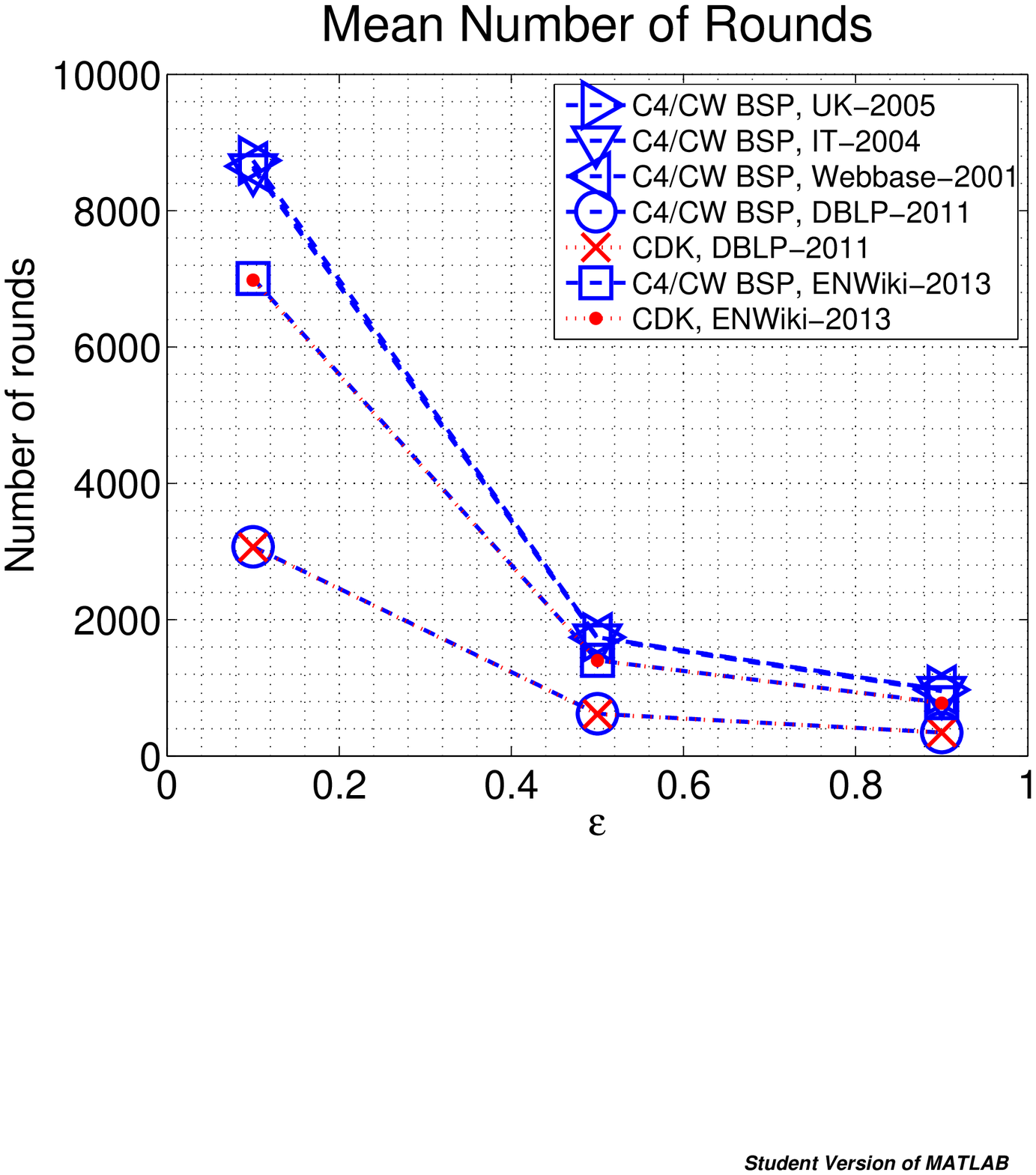}
      \caption{ Mean number of synchronization rounds for BSP algorithms}
      \label{fig:numround}
    \end{subfigure} \\
    \begin{subfigure}[b]{0.49\textwidth}
      \includegraphics[width=185pt]{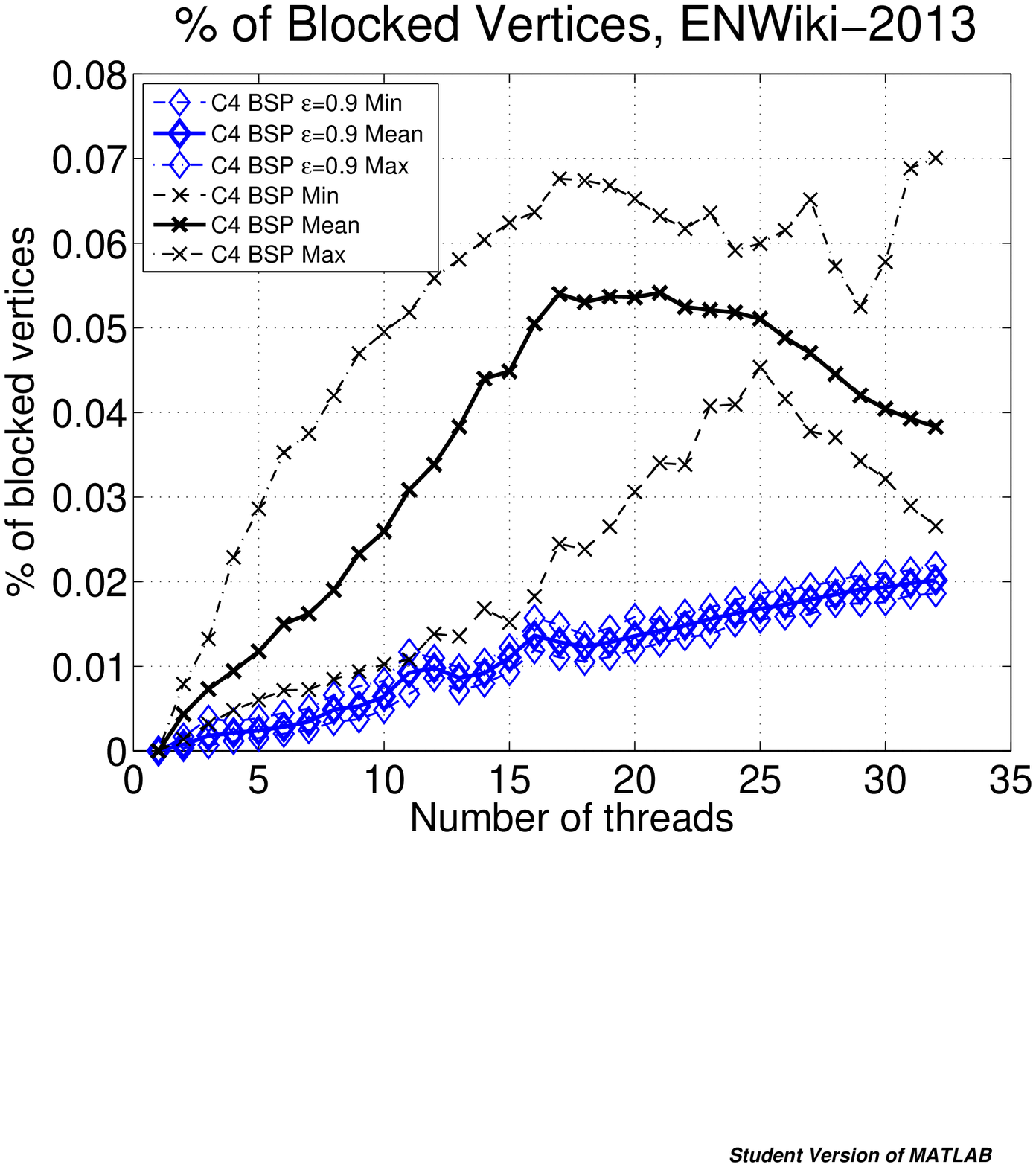}
      \caption{ Percent of blocked vertices for C4, ENWiki-2013. BSP run with $\epsilon=0.9$.}
      \label{fig:blktrans_ew13-09}
    \end{subfigure} &
    \begin{subfigure}[b]{0.49\textwidth}
      \includegraphics[width=185pt]{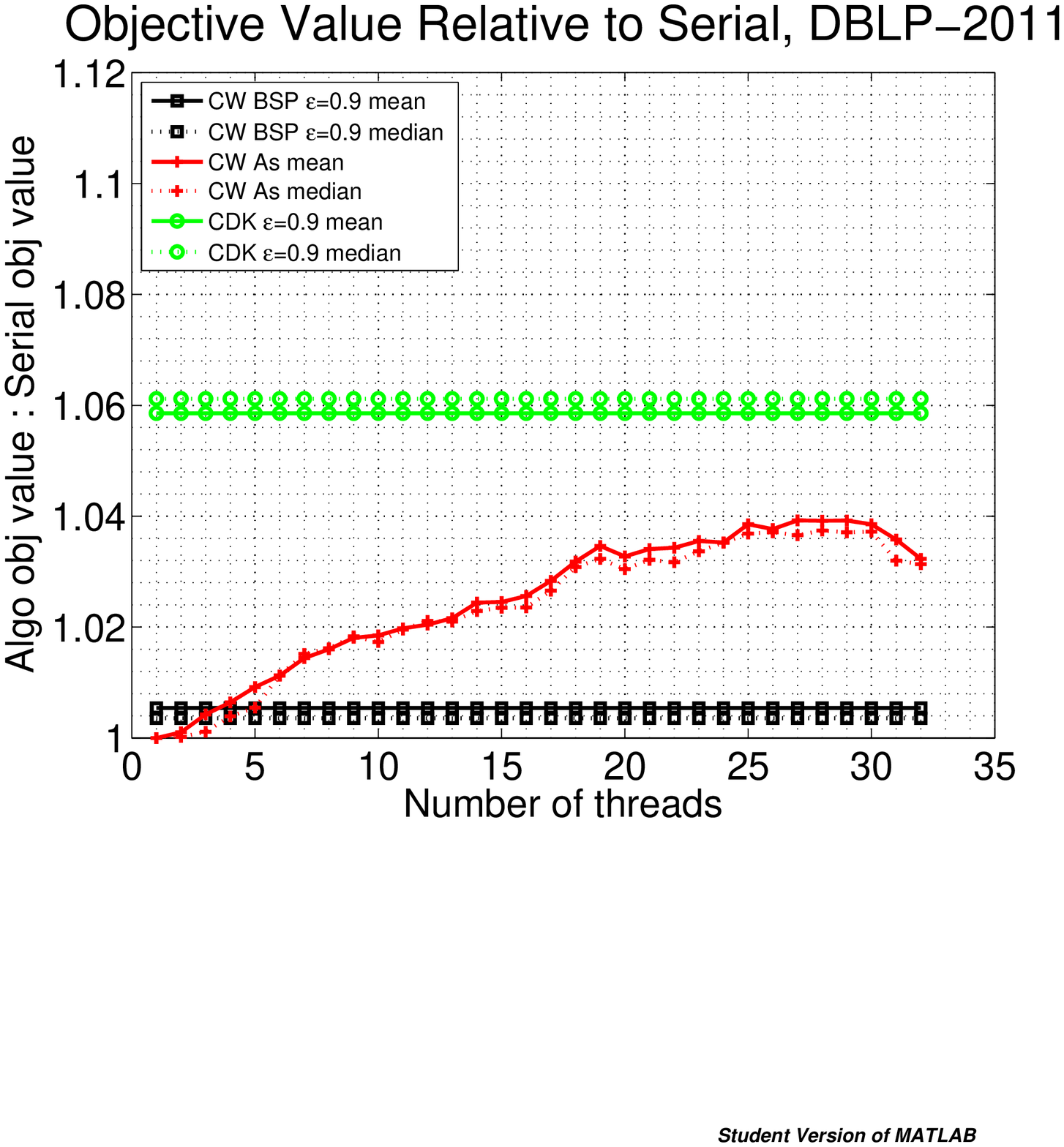}
      \caption{ Median objective values, DBLP-2011. CW BSP and CDK run with $\epsilon = 0.9$}
      \label{fig:objvalue_db11-09}
    \end{subfigure} \\
  \end{tabular}
  \caption{In the above figures, `CW' is short for \CW{}, `BSP' is short for the bulk-synchronous variants of the parallel algorithms, and `As' is short for the asynchronous variants.}
\end{figure}

\subsection{Runtimes}
\CC{} and \CW{} are initially slower than serial, due to the overheads required for atomic operations in the parallel setting.
However, all our parallel algorithms outperform serial \KC{} with 3-4 threads.
As more threads are added, the asychronous variants become faster than their BSP counterparts as there are no synchronization barrriers.
The difference between BSP and asychronous variants is greater for smaller $\epsilon$.
\CW{} is also always faster than \CC{} since there are no coordination overheads.
\subsection{Speedups}
The asynchronous algorithms are able to achieve a speedup of 13-15x on 32 threads.
The BSP algorithms have a poorer speedup ratio, but nevertheless achieve 10x speedup with $\epsilon=0.9$.
\subsection{Synchronization rounds}
The main overhead of the BSP algorithms lies in the need for synchronization rounds.
As $\epsilon$ increases, the amount of synchronization decreases, and with $\epsilon=0.9$, our algorithms have less than 1000 synchronization rounds, which is small considering the size of the graphs and our multicore setting.
\subsection{Blocked vertices}
Additionally, \CC{} incurs an overhead in the number of vertices that are blocked waiting for earlier vertices to complete.
We note that this overhead is extremely small in practice---on all graphs, less than 0.2\% of vertices are blocked.
On the larger and sparser graphs, this drops to less than 0.02\% (i.e., 1 in 5000) of vertices.
\subsection{Objective value}
By design, the \CC{} algorithms also return the same output (and thus objective value) as serial \KC{}.
We find that \CW{} BSP is at most 1\% worse than serial across all graphs and values of $\epsilon$.
The behavior of asynchronous \CW{} worsens as threads are added, reaching 15\% worse than serial for one of the graphs.
Finally, on the smaller graphs we were able to test CDK on, we find that CDK returns a worse median objective value than both \CW{} variants.

%% file: theory.tex
\section{Proofs of Theoretical Guarantees}

\subsection{Number of rounds for \CC{} and \CW{}}

\lemnumrounds*
\begin{proof}
We split our proof in two parts.

For \CW{}, we wish to upper bound the probability
\begin{align*}
q_t = 
\Pr\left\{\text{$v$ not clustered by round $i+t$}\left| \deg_{i+j}(v) \ge \frac{\Delta_i}{2}, 1\le j\le t \right.\right\}.
\end{align*}
Observe that the above event happens either if no neighbors of $v$ become activated by round $i+t$, or if $v$ itself does not become activated.
Hence, $q_t$ can be upper bounded by the probability that no neighbors of $v$ become activated by round $i+t$.

In the following, let $d_{i+j}$ denote the degree of vertex $v$ at roudn $i+j$;
for simplicity we drop the round indices on $n$ and $P$.
The probability, per round, that no neighbors of $v$ become activated is equal to\footnote{This follows from a simple calculation on the pdf of the hypergeometric distribution.}
\begin{align*}
&\frac{{n-d_{i+j}\choose P}}{{n\choose P}}
= \frac{ (n-P)!}{ (n-P-d_{i+j})! }\cdot \frac{(n-d_{i+j})!}{n!}\\
&= \frac{\prod_{t=1}^{d_{i+j}}(n-d_{i+1}+t-P)}{\prod_{t=1}^{d_{i+j}}(n-d_{i+1}+t)}= \prod_{t=1}^{d_{i+j}} \frac{n-d_{i+1}+t-P}{n-d_{i+1}+t}\\
&= \prod_{t=1}^{d_{i+j}}\left(1- \frac{P}{n-d_{i+1}+t}\right)\le \left(1- \frac{P}{n}\right)^{d_{i+j}}\\
&\le \left(1- \frac{\epsilon}{\Delta_i}\right)^{\Delta_i/2}= \left[\left(1- \frac{\epsilon}{\Delta_i}\right)^{\Delta_i/\epsilon}\right]^{\epsilon/2}
\le e^{-\epsilon/2}.
\end{align*}
where the last inequality is due to the fact that 
$$(1-x)^{1/x}<e^{-1}\text{ for all }x\le 1.$$
Therefore, the probability of vertex $v$ failing to be clustered after $t$ rounds is at most
$q_t\le e^{-t\cdot \epsilon /2}.$
Hence, we have that for any round $i$, the probability that any vertex has degree more than $\Delta_i/2$ after $t$ rounds is at most $n\cdot e^{-t\cdot \epsilon /2}$, due to a simple union bound.
If we want that that probability to be smaller than $\delta$, then
\begin{align*}
n\cdot e^{-t\cdot \epsilon /2} &< \delta
\Leftrightarrow \ln n-t\cdot \epsilon /2 < \ln(\delta)
\Leftrightarrow t > \frac{2}{\epsilon}\cdot\ln(n/\delta)
\end{align*}
Hence, with probability $1-\delta$, after $\frac{2}{\epsilon}\cdot\ln(n/\delta)$ rounds either all nodes of degree greater than $\Delta/2$ are clustered, or the maximum degree is decreased by half.
Applying this argument $\log\Delta$ times yields the result, as the maximum degree of the remaining graph becomes $1$.

For \CC{} the proof follows simply from the analogous proof of \cite{blelloch2012greedy}.
Consider any round of the algorithm, and break it into $k$ steps (each step, for each vertex in $\mathcal{A}$ that becomes a cluster center).
Let $v$ be a vertex that has degree at most $\Delta/2$, and is not active.
During step $1$ of round $1$, the probability that $v$ is not adjacent to $\pi(1)$ is at most $1-\frac{\epsilon}{2n}$.
If $v$ is not selected at step $1$, then during step $2$ of round $1$, the probability that $v$ is not adjacent to the next cluster center is again at most $1-\frac{\epsilon}{2n}$.
After processing all vertices in $\mathcal{A}$, during the first round, either $v$ was clustered, or its degree became strictly less than $\Delta/2$, or the probability that neither of the previous happened is at most $(1-\frac{\epsilon}{2n})^{\frac{\epsilon\Delta}{n}}\le 1-\epsilon/2$.
It is easy to see that after $O(\frac{1}{\epsilon}\log n)$ rounds  vertex $v$ will have either been clustered or its degree would be smaller than $\Delta/2$.
Union bounding for $n$ vertices and all rounds, we get that the max degree of the remaining graph gets halved after $O(\frac{1}{\epsilon}\log n)$ rounds, hence the total number of rounds needed is at most $O(\frac{1}{\epsilon}\log n \log \Delta)$, with high probability.
\end{proof}

\subsection{Running times}
In this section, we prove the running time theorem for our Algorithms.
We first present the following recent graph-theoretic result.
\begin{theo}[Theorem 1 in \cite{krivelevich2014phase}]\label{thm:krivelevich}
Let $G$ be an undirected graph on $n$ vertices, with maximum vertex degree $\Delta$.
Let us sample each vertex independently with probability 
$p = \frac{1-\epsilon}{\Delta}$
and define as $G'$ the induced subgraph on the activated vertices.
Then, the largest connected component of the resulting graph $G'$ has size at most 
$O(\frac{1}{\epsilon^2} \log n)$
with high probability.
\label{theo:kriv}
\end{theo}

To apply Theorem \ref{thm:krivelevich}, we first need to convert it into a result for sampling without replacement (instead of i.i.d. sampling).
\begin{lem}
Let us define two sequences of binary random variables $\{X_i\}_{i=1}^n, \{Y_i\}_{i=1}^n$.
The first sequence comprises $n$ i.i.d. Bernoulli random variables with probability $p$,
and the second sequence a random subset of $B$ random variables is set to $1$ without replacement, 
where $B$ is integer that satisfies 
$$(n+1)\cdot p-1\le B < (n+1)\cdot p.$$

Let us now define
$\rho_X = \Pr \left( f(X_1,\ldots, X_n) > C \right)$
for some $f$ (in our case this will be the largest connected component of a subgraph defined on the sampled vertices) and some number $C$,
and similarly define $\rho_Y$.
Let us further assume that we have an upper bound on the above probability 
$\rho_X\le \delta$. 
Then, $\rho_Y \le  n\cdot \delta$.
\end{lem}

\begin{proof}
By expanding $\rho_X$ using law of total probability we have
\begin{align}
\rho_X& =\sum_{b=0}^n \Pr \left( f(X_1,\ldots, X_n ) > C \left| \sum_{i=1}^n X_i = b\right.\right) \cdot \Pr\left( \sum_{i=1}^n X_i = b\right)\nonumber\\
&=\sum_{b=0}^n q_b \cdot \Pr\left( \sum_{i=1}^n X_i = b\right)
\end{align}
where $q_b$ is the probability that $f(X_1,\ldots, X_n ) > C$ given that a uniformly random subset of $b$ variables was set to $1$.
Moreover, we have 
\begin{align}
\rho_Y& =\sum_{b=0}^{n} \Pr \left( f(Y_1,\ldots, Y_n) > C \left| \sum_{i=1}^n Y_i = b\right.\right) \cdot \Pr\left( \sum_{i=1}^n Y_i =b\right )\nonumber\\
&\overset{(i)}{=}\sum_{b=0}^{n} q_b \cdot \Pr\left( \sum_{i=1}^n Y_i =b\right )\nonumber\\
&\overset{(ii)}{=} q_{B} \cdot 1
\end{align}
where $(i)$ comes form the fact that $\Pr \left( f(Y_1,\ldots, Y_n) > C \left| \sum_{i=1}^n Y_i = b\right.\right) $ is the same as the probability that that $f(X_1,\ldots, X_n ) > C$ given that a uniformly random subset of $b$ variables where set to $1$, and 
$(ii)$ comes from the fact that since we sample without replacement in $Y$,we have that $\sum_{i}^n Y_i = B $ always.

If we just keep the $b=B$ term in the expansion of $\rho_X$ we get
\begin{equation}
\rho_X=\sum_{b=0}^n q_b \cdot \Pr\left( \sum_{i=1}^n X_i = b\right)\ge q_{B} \cdot \Pr\left( \sum_{i=1}^n X_i = B\right) = \rho_Y \cdot \Pr\left( \sum_{i=1}^n X_i = B\right)
\label{eq:rho_X_LB}
\end{equation}
since all terms in the sum are non-negative numbers.
Moreover, since $X_i$s are Bernoulli random variables, 
then $\sum_{i=1}^n X_i$ is Binomially distributed with parameters $n$ and $p$.
We know that the maximum of the Binomial pmf with parameters $n$ and $p$ occurs at $\Pr\left(\sum_i X_i = B\right)$ where $B$ is the integer that satisfies $(n+1)\cdot p-1\le B < (n+1)\cdot p$.
Furthermore we know that the maximum value of the Binomial pmf cannot be less than $\frac{1}{n}$, that is
\begin{equation}
\Pr\left( \sum_{i=1}^n X_i =B\right) \ge \frac{1}{n}.
\label{eq:maxBinomial}
\end{equation}
If we combine \eqref{eq:rho_X_LB} and \eqref{eq:maxBinomial} we get
$\rho_X\ge \rho_Y / n \Leftrightarrow \rho_Y \le  n\cdot \delta$. 
\end{proof}

\begin{cor}\label{cor:kriv2}
Let $G$ be an undirected graph on $n$ vertices, with maximum vertex $\Delta$.
Let us sample $\epsilon \cdot \frac{n}{\Delta}$ vertices without replacement,
and define as $G'$ the induced subgraph on the activated vertices.
Then, the largest connected component of the resulting graph $G'$ has size at most 
$$O\left(\frac{\log n}{\epsilon^2}\right) $$
with high probability.
\end{cor}

We use this in the proof of our theorem that follows.

\thmrunningtime*
\begin{proof}
We start with analyzing \CC{}, as the running time of \CW{} follows from a similar, and simpler analysis. 
Observe, that we operate on Bulk Synchronous Parallel model: we sample a batch of vertices, $P$ cores asynchronously process the vertices in the batch, and once the batch is empty there is a bulk synchronization step.
The computational effort spent by \CC{} can be split in three parts: i) computing the maximum degree, ii) creating the clusters, per batch, iii) syncronizing at the end of each batch.

\paragraph{Computing $\Delta$ and synchronizing cost} Computing $\Delta_i$ at the beginning of each batch, can be implemented in time $\frac{m_i}{P}+\log P$, where each thread picks $n_i/P$ vertices and computes locally their degrees, and inserts it to a sorted data structure ({\it e.g.}, a B-tree that admits parallel operations), and then we get the largest item in logarithmic time.
Moreover, the third part of the computation, i.e., synchronization among cores, can be done in $O(P)$.
A little more involved argument is needed for establishing the running time of the second part, where the algorithms create the clusters.

\paragraph{Clustering cost}  For a single vertex $v$ sampled by a thread, the time required by the thread to process that vertex is the sum of the time needed to 
1) wait  inside the \text{attemptCluster} for preceding neighbors (by the order of $\pi$),
2) ``send" its $\pi(v)$ to its neighbors, if $v$ is a cluster center,
3) if $v$ is a cluster center, then for each $u$ neighbors it will attempt to update $\text{clusterID}(u)$; however, this thread  potentially competes with other threads that are attempting to write in $\text{clusterID}(u)$ at the same time.

Using Corollary \ref{cor:kriv2}, we can show that no more than $O(\log n)$ threads compete with each other at the same time, with high probability.
Observe, that in our sampling scheme of batches of vertices, we are taking the first $B_i= \frac{\epsilon}{\Delta_i}\cdot n_i$ elements of a random prefix $\pi$.
This is equivalent to sampling $B_i$ vertices without replacement from the graph $G_i$ of the current round.
The result in Corollary~\ref{cor:kriv2}, asserts that the largest connected component in the sampled subgraph is at most $O(\log n)$, with high probability. 
This directly implies that
a thread cannot be waiting for more than $O(\log n)$ other threads inside \texttt{attemptCluster}($v$).
Therefore, the time spent by each thread to wait on other threads in \texttt{attemptCluster}($v$) is upper bounded by the number of maximum threads that it can be neighbors with (which assuming $\epsilon$ is set to $1/2$) is at most $O(\log n)$, times the time it takes each of these threads to be done with their execution, which is at most $\Delta_i \log n $ (even assuming the worst case conflict pattern when updating at most $\Delta_i$ entries in the clusterID array).
Hence, for \CC{} the processing time of a single vertex is upper bounded by $O(\Delta_i \cdot \log^2n)$.

\paragraph{Job allocation}Now, observe that when each thread is done processing  vertex, it picks the next vertex from $\setA$ (if $\setA$ is not empty). 
This process essentially models a classical greedy task allocation to cores, that leads to a $2$ approximation in terms of the optimum weight allocation; here the optimum allocation leads to a max weight among cores that is at most equal to 
$\max(\Delta_i, B_i\Delta_i/P)$.
This implies that the running time on $P$ asynchronous threads of a single batch, is upperbounded by
$$O\left(\max\left(\Delta_i\log n, \frac{B_i\Delta_i\log^2 n}{P}\right)\right)=O\left(\max\left(\Delta_i\log n, \frac{n_i\log^2 n}{P}\right)\right).$$
Assuming, that the number of cores, is always less than the batch size (a reasonable assumption, as more cores, would not lead to further benefits), we obtain that the time for a single batch is 
$$O\left(\frac{E_i}{P}+\frac{n_i\log^2 n}{P}+P\right).$$

Observe that a difference in \CW{}, is that waiting is avoided, hence, the running time, per batch of \CW{} is
$$O\left(\frac{E_i}{P}+\frac{n_i}{P}+P\right).$$
Multiplying the above, with the number of rounds given by Lemma \ref{lem:rounds}, we obtain the theorem.

\end{proof}

\subsection{Approximation Guarantees}
One can view the execution of ClusterWild! on $G$ as having \KC{} run on a ``noisy version" of $G$.
A main issue is that \KC{} never allows two neighbors in the original graph to become cluster centers.
Hence, since ClusterWild! ignores these edges among active vertices, one can view these edges as ``adverserially" deleted.
The major technical contribution of this work is to quantify how these ``ignored" edges affect the quality of the output solution.
The following simple lemma presented in our main text, is useful in quantifying the cost of the output clustering for {\it any} peeling algorithm.
\lemalgocost*
\begin{proof}
Consider the first step of the algorithm, for simplicity, and without loss of generality.
Let us define as $T_{\text{in}}$ the  number of vertex pairs inside $\mathcal{C}_v$ that are not neighbors (i.e., they are joined by a negative edge).
Moreover, let $T_{\text{out}}$ denote the number of vertices outside $\mathcal{C}_v$ that are neighbors with vertices inside $\mathcal{C}_v$.
Then, the number of disagreements (i..e, number of misplaced pairs of vertices) generated by cluster $\mathcal{C}_v$, is equal to $T_{\text{in}}+T_{\text{out}}$.

Observe that all the $T_{\text{in}}$ edges are negative, and all $T_{\text{out}}$ are positive ones.
Let for example $(u,w)$ be one of the $T_{\text{in}}$ negative edges inside $\Cv$, hence both $u,w$ belong to $\Cv$ (i.e., are neighbors with $v$).
Then, $(u,v,w)$ forms a {\it bad triangle}.
Similarly, for every edge that is incident to a vertex in $\Cv$, with one end point say $u'\in \Cv$ and one $w'\in V\backslash v$, the triangle formed by $(v,u',w')$, is also a bad triangle.

Hence, all edges that are accounted for in the final cost of the algorithm (i..e, total number of disagreements) are equal to the $T_{\text{in}}+T_{\text{out}}$ bad triangles that include these edges and each cluster center per round.
\end{proof}

Let us now consider the set of all cluster centers generated by ClusterWild!; call these vertices $\mathcal{C}_{\text{CW}}$.
Then, consider the graph $G'$ that is generated by deleting all edges between $\mathcal{C}_{\text{CW}}$.
Observe that this is a random graph, since the set of edges deleted depends on the specific random sampling that is performed in ClusterWild!.
We will use the following simple technical proposition to quantify how many more bad triangles $G'$ has compared to $G$.

\begin{prop}
Given any graph $G$ with positive and negative edges, then let us obtain a graph $G_e$ where we have removed a single edge, $e$ from $G$.
Then, the $G_e$ has at most $\Delta$ more bad triangles compared to $G$.
\end{prop}

\begin{proof}
Let $(i,j,k)$ be a bad triangle in $G$ but not in $G_e$.
Then it must be the case that $e \in t$.
WLOG let $e = (i,j)$, and so $k \in N(i) \cup N(j)$.
Since $|N(i) \cup N(j)| \leq 2\max(\deg_i, \deg_j) \leq 2\Delta$, there can be at most $\Delta$ new bad triangles in $G_e$.
\end{proof}

The above proposition is used to establish the $\tau_{\text{new}}$ bound for Lemma~\ref{lem:cwcost}.
Now, assume a random permutation $\pi$ for which we run \CW{}, and let $\hat{\mathcal{A}}=\cup_{r=1}^R\mathcal{A}_r$ denote the union of all active sets of vertices, for each round $r$ of the algorithm.
Moreover, let $\hat G$, denote the graph that is missing all edges between the vertices in the sets $\mathcal{A}_r$.
A simple way to bound the clustering error of \CW{}, is splitting it in to two terms: the number of old bad triangles of $G$ adjacent to active vertices (i.e., we need to bound the expectation of the event that an active vertex is adjacent to an ``old" triangle), plus the number of {\it all} new triangles induced by ignoring edges.
Observe that this bound can be loose, since not all ``new" bad triangles of $\hat G$ count towards the clustering error, and some ``old" bad triangles can disappear.
However, this makes the analysis tractable.
Lemma \ref{lem:cwcost} then follows.

\lemcwcost*

%% file: implementation.tex
\section{Implementation Details}
Our implementation is highly optimized in our effort to have practically scalable algorithms.
We discuss these details in this section.

\subsection{Atomic and non-atomic variables in Java/Scala}
In Java/Scala, processors maintain their own local cache of variable values, which could lead to spinlocks in \CC{} or greater errors in \CW{}.
It is necessary to enforce a consistent view across all processors by the use of synchronization or AtomicReferences, but doing so will incur high overheads that render the algorithm not scalable.

To mitigate this overhead, we exploit a monoticity property of our algorithms---the clusterID of any vertex is a non-increasing value.
Thus, many of the checks in \CC{} and \CW{} may be sufficiently performed using only an outdated version of clusterID.
Hence, we may maintain both an inconsistent but cheap clusterID array as well as an expensive but consistent atomic clusterID array.
Most reads can be done using the cheap inconsistent array, but writes must propagate to the consistent atomic array.
Since each clusterID is written a few times but read often, this allows us to minimize the cost of synchronizing values without any substantial changes to the algorithm itself.

We point out that the same concepts may be applied in a distributed setting to minimize communication costs.

\subsection{Estimating but not computing $\Delta$}
As written, the BSP variants require a computation of the maximum degree $\Delta$ at each round.
Since this effectively involves a scan of all the edges, it can be an expensive operation to perform at each iteration.
We instead use a proxy $\hat\Delta$ which is initialized to $\Delta$ in the first round, and halved every $\frac{2}{\epsilon}\ln(n\log\Delta/\delta)$ rounds.

With a simple modification to Lemma \ref{lem:numrounds}, we can see that w.h.p. any vertex with degree greater than $\hat\Delta$ will either be clustered or have its degree halved after $\frac{2}{\epsilon}\ln(n\log\Delta/\delta)$ rounds, so $\hat\Delta$ upper-bounds $\Delta$ and our algorithms complete in logarithmic number of rounds.

\subsection{Lazy deletion of vertices and edges}
In practice, we do not remove vertices and edges as they are clustered, but simply skip over them when they are encountered later in the process.
We find that this approach decreases the runtimes and overall complexity of the algorithm.
(In particular, edges between vertices adjacent to cluster centers may never be touched in the lazy deletion scheme, but must nevertheless be removed in the proactive deletion approach.)
Lazy deletions also allow us to avoid expensive mutations of internal data structures.

\subsection{Binomial sampling instead of fixed-size batches}
Lazy deletion does introduce an extra complication, namely it is now more difficult to sample a fixed-size batch of $\epsilon n_i  / \Delta_i$ vertices, where $n_i$ is the number of remaining unclustered vertices.
This is because we do not maintain a separate set of $n_i$ unclustered vertices, nor explicitly compute the value of $n_i$.

We do, however, maintain a set of \emph{unprocessed} vertices, that is, a suffix of $\pi$ containing $n_i$ unclustered vertices and $m_i$ clustered vertices that have not been passed through by the algorithm.
We may therefore resort to an i.i.d. sampling of these vertices, choosing each with probability $\epsilon / \Delta_i$.
Since processing an unprocessed but clustered vertex has no effect, we effectively simulate an i.i.d. sampling of the $n_i$ unclustered vertices.

Furthermore, we do not have to actually sample each vertex---because $\pi$ is a uniform random permutation, it suffices to draw $B \sim Bin(n_i+m_i, \epsilon/\Delta_i)$ and extract the next $B$ elements from $\pi$ for processing, reducing the number of random draws from $n_i+m_i$ Bernoullis to a single Binomial.

All of our theorems hold in expectation when using i.i.d. sampling instead of fixed-size batches.

\subsection{Comment on CDK Implementation}
A crucial difference between the CDK algorithm and our algorithms lies in the fact that CDK might reject vertices from the active set, which are then placed back into the set of unclustered vertices for potential selection at later rounds.
Conversely, our algorithms ensure that the active set is always completely processed, so any vertex that has been selected will no longer be selected in an active set again.
We are therefore able to exploit a single random permutation $\pi$ and use the tricks with lazy deletions and binomial sampling that are not available to CDK, which instead has to perform the complete i.i.d. sampling.
We believe that this accounts for the largest difference in runtimes between CDK and our algorithms.

%% file: exptresults.tex
\section{Full experiment results}
\label{app:exptresults}

\begin{figure}[ht]
  \centering
  \begin{tabular}{ccc}

    \begin{subfigure}[b]{0.31\textwidth}
      \includegraphics[width=120pt]{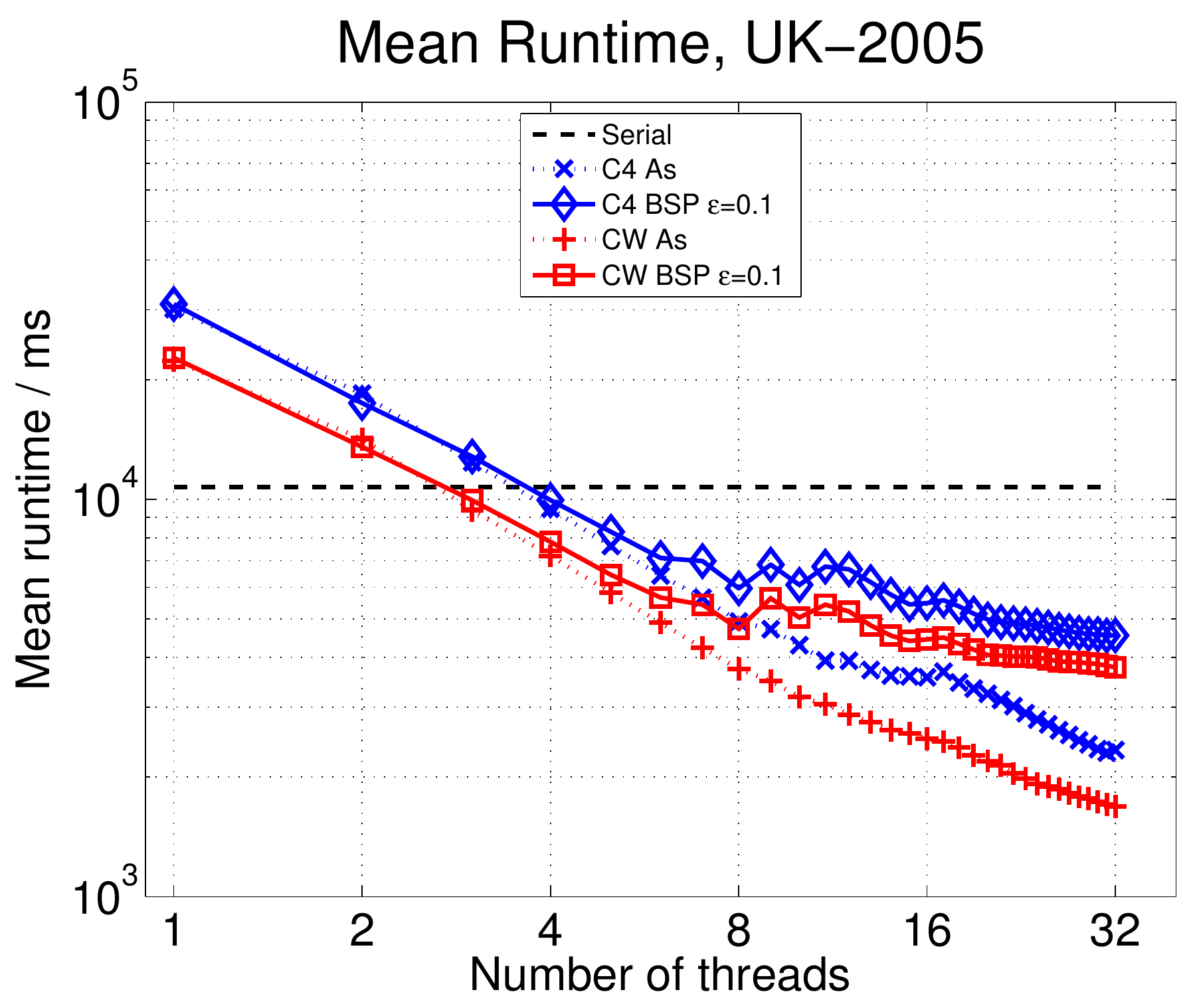}
      \caption{UK-2005, $\epsilon = 0.1$}
      \label{appfig:runtimes_uk05_ave_01}
    \end{subfigure} &
    \begin{subfigure}[b]{0.31\textwidth}
      \includegraphics[width=120pt]{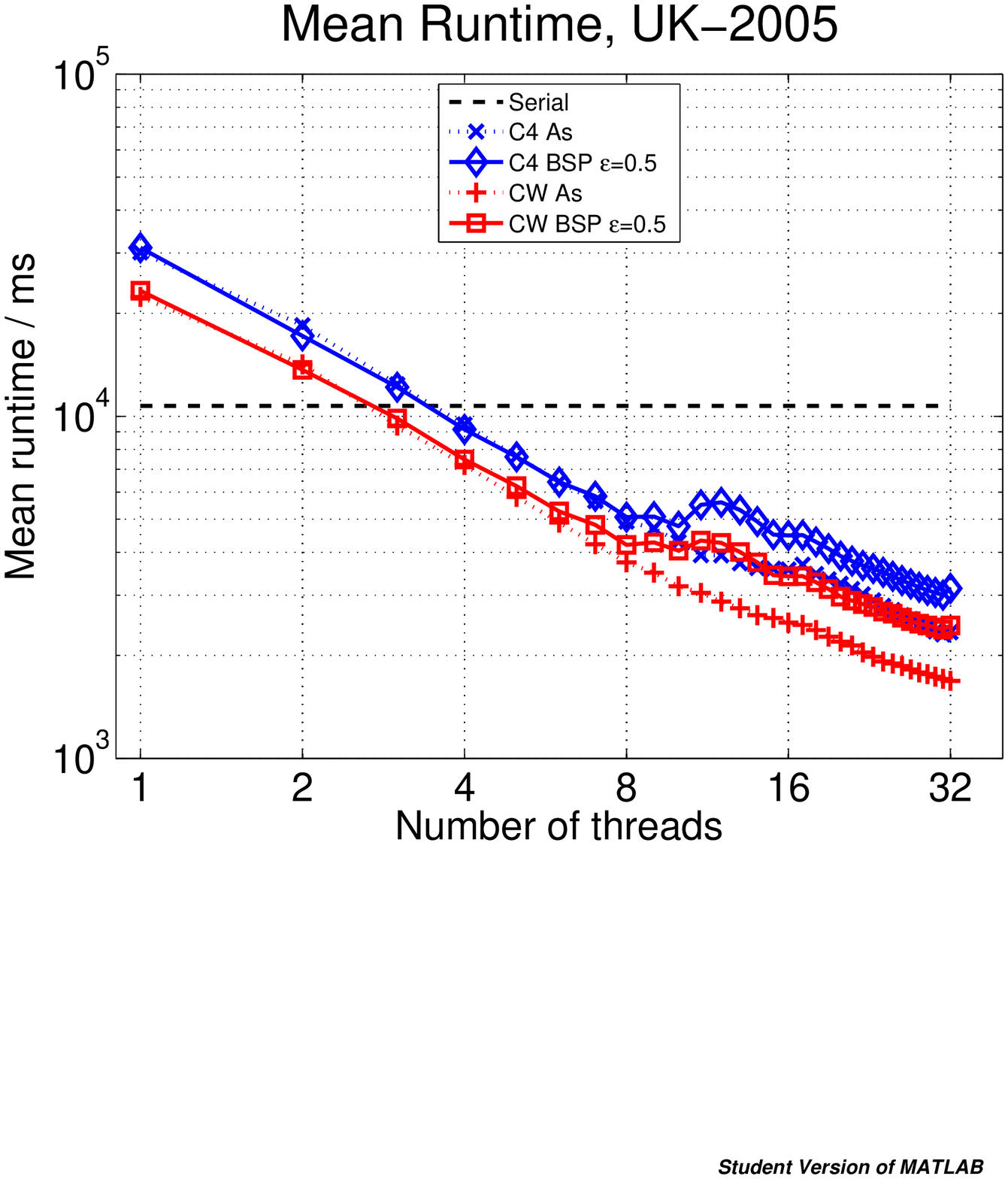}
      \caption{UK-2005, $\epsilon = 0.5$}
      \label{appfig:runtimes_uk05_ave_05}
    \end{subfigure} &
    \begin{subfigure}[b]{0.31\textwidth}
      \includegraphics[width=120pt]{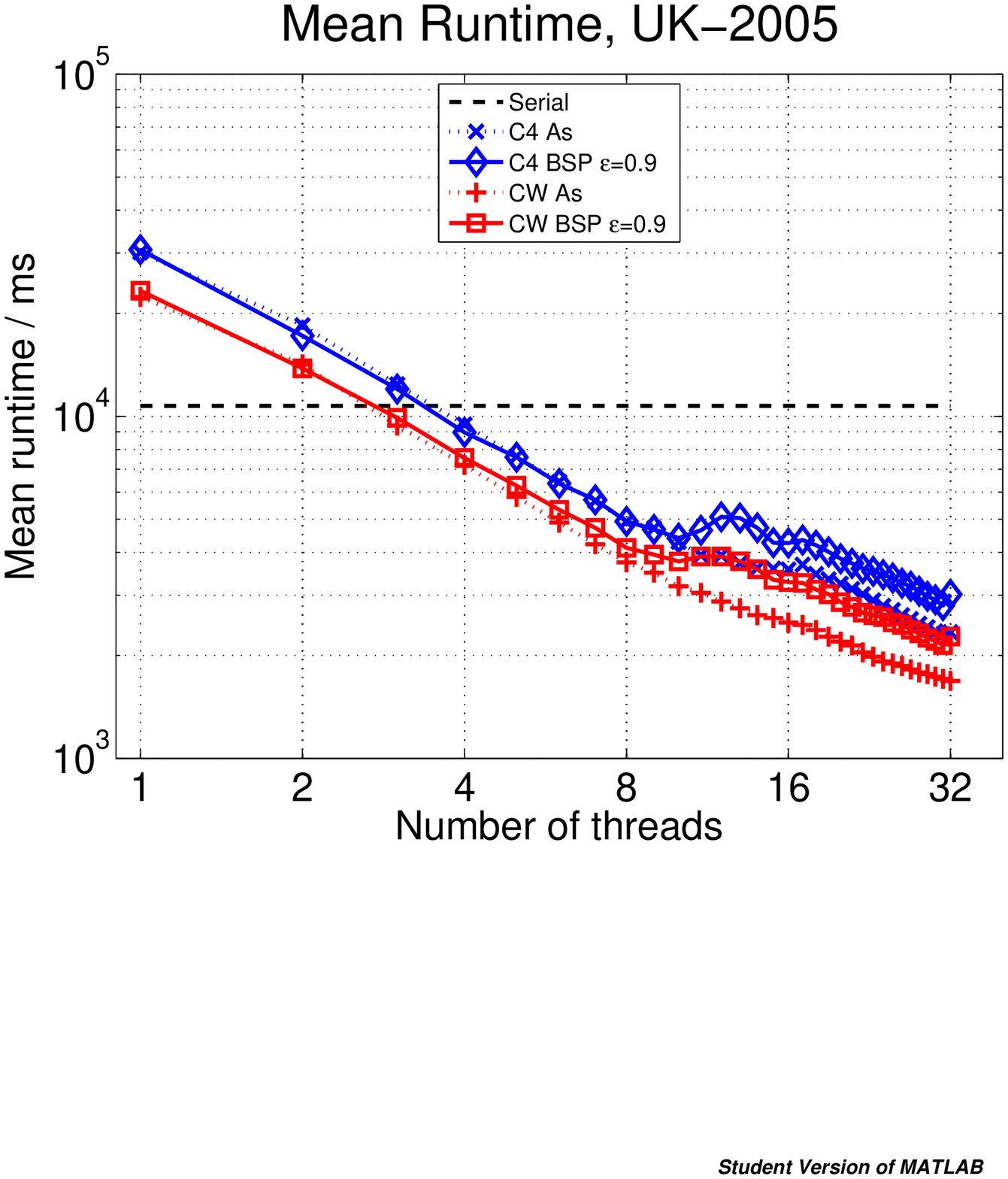}
      \caption{UK-2005, $\epsilon = 0.9$}
      \label{appfig:runtimes_uk05_ave_09}
    \end{subfigure} \\

    \begin{subfigure}[b]{0.31\textwidth}
      \includegraphics[width=120pt]{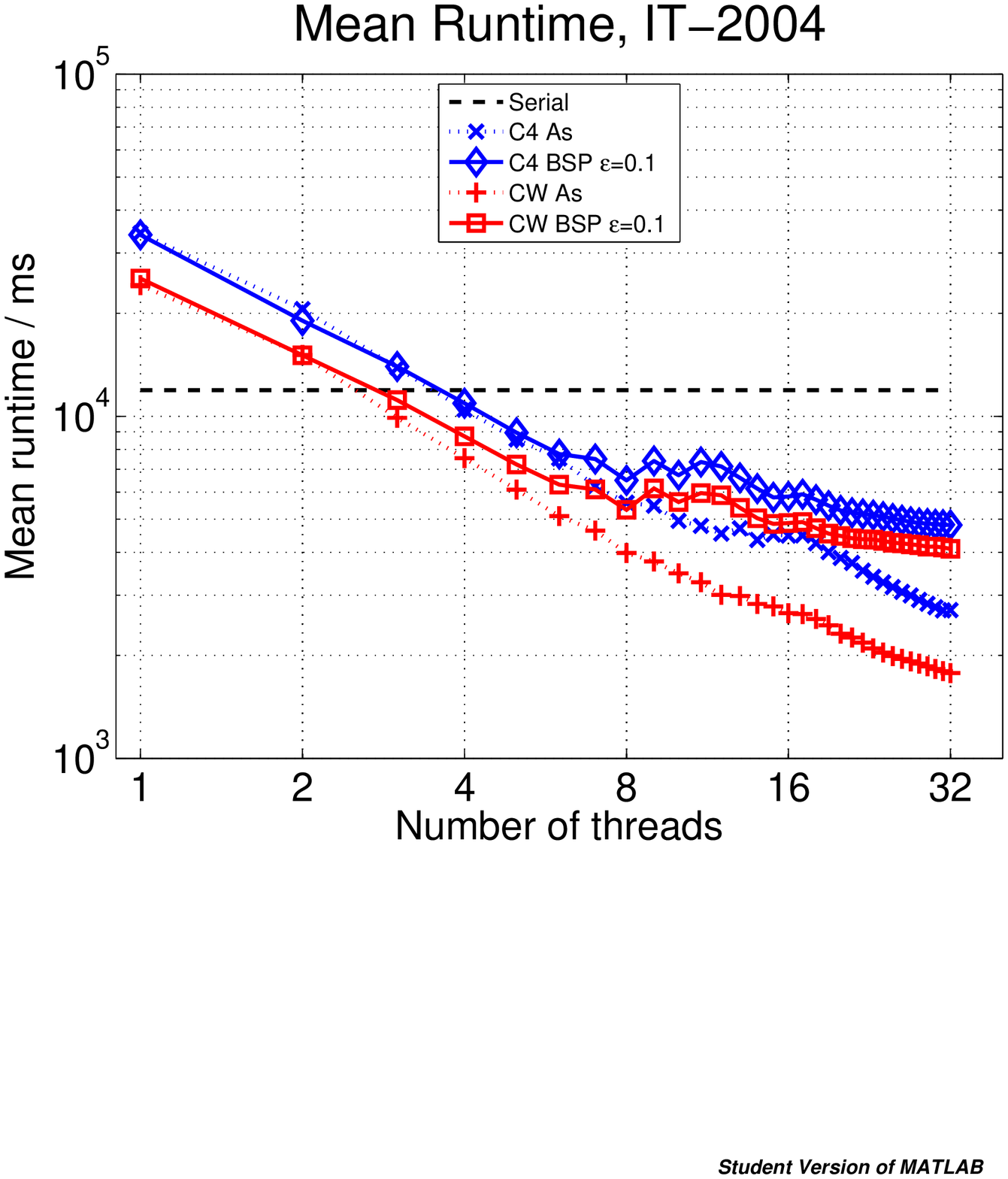}
      \caption{IT-2004, $\epsilon = 0.1$}
      \label{appfig:runtimes_it04_ave_01}
    \end{subfigure} &
    \begin{subfigure}[b]{0.31\textwidth}
      \includegraphics[width=120pt]{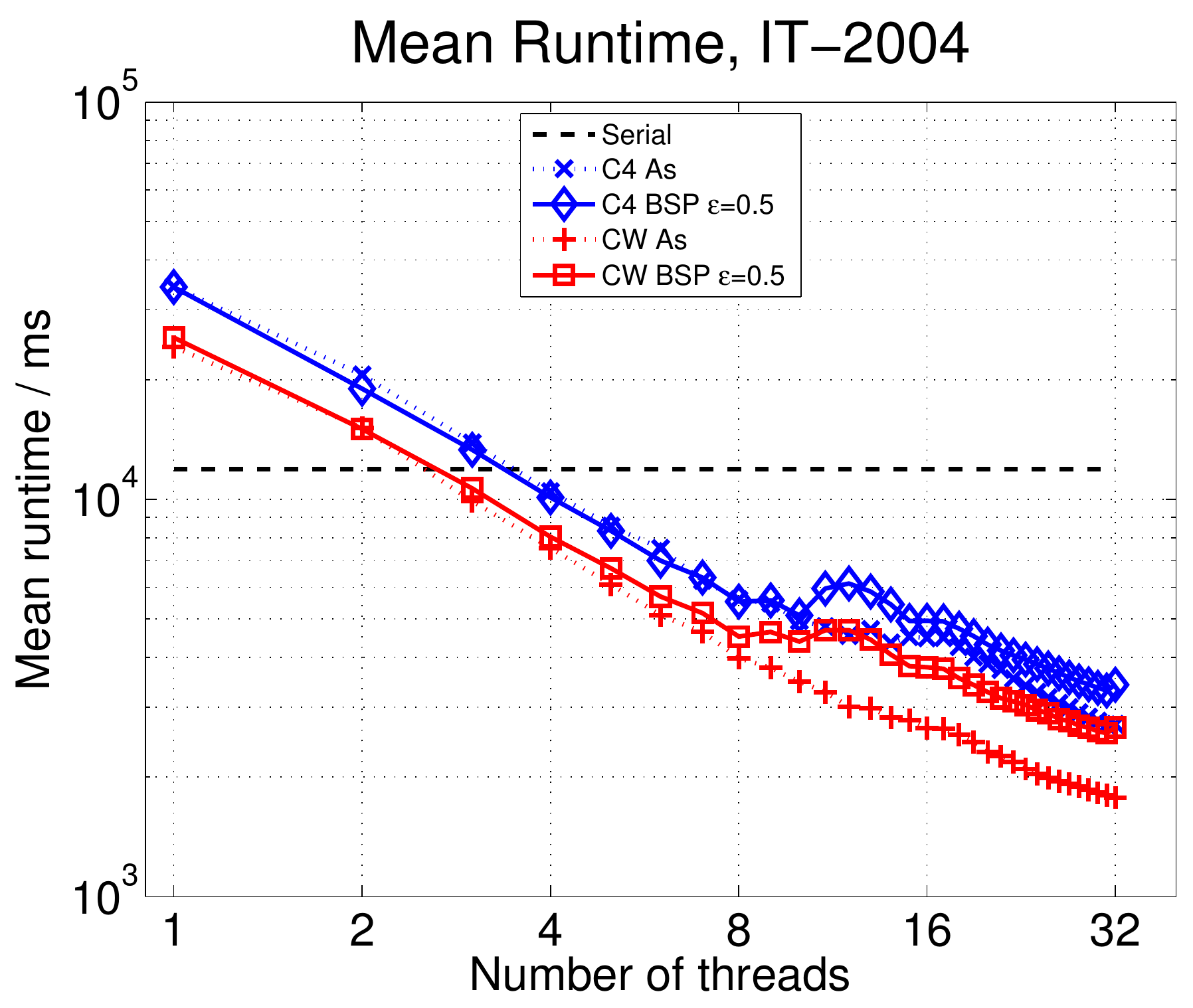}
      \caption{IT-2004, $\epsilon = 0.5$}
      \label{appfig:runtimes_it04_ave_05}
    \end{subfigure} &
    \begin{subfigure}[b]{0.31\textwidth}
      \includegraphics[width=120pt]{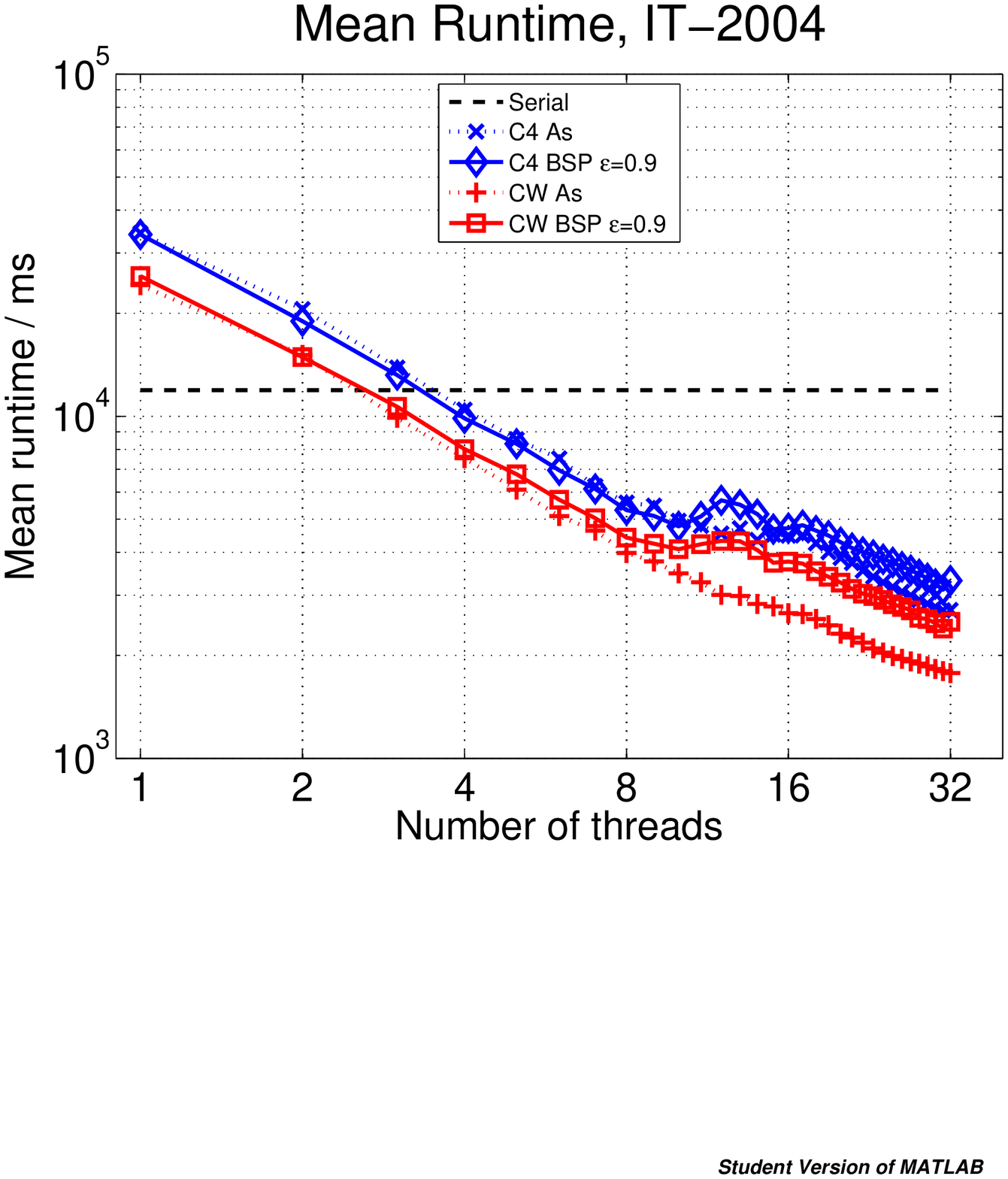}
      \caption{IT-2004, $\epsilon = 0.9$}
      \label{appfig:runtimes_it04_ave_09}
    \end{subfigure} \\

    \begin{subfigure}[b]{0.31\textwidth}
      \includegraphics[width=120pt]{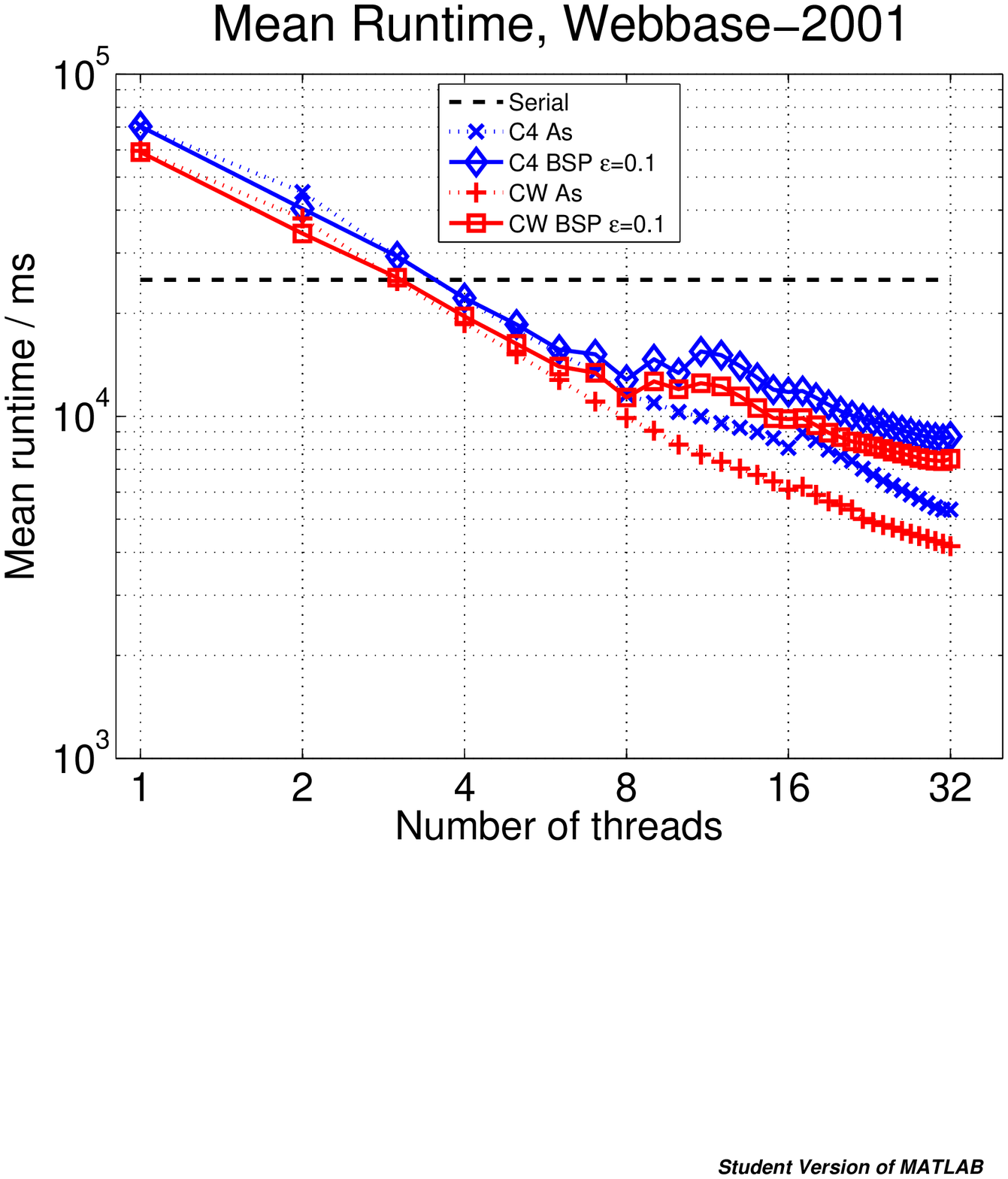}
      \caption{Webbase-2001, $\epsilon = 0.1$}
      \label{appfig:runtimes_wb01_ave_01}
    \end{subfigure} &
    \begin{subfigure}[b]{0.31\textwidth}
      \includegraphics[width=120pt]{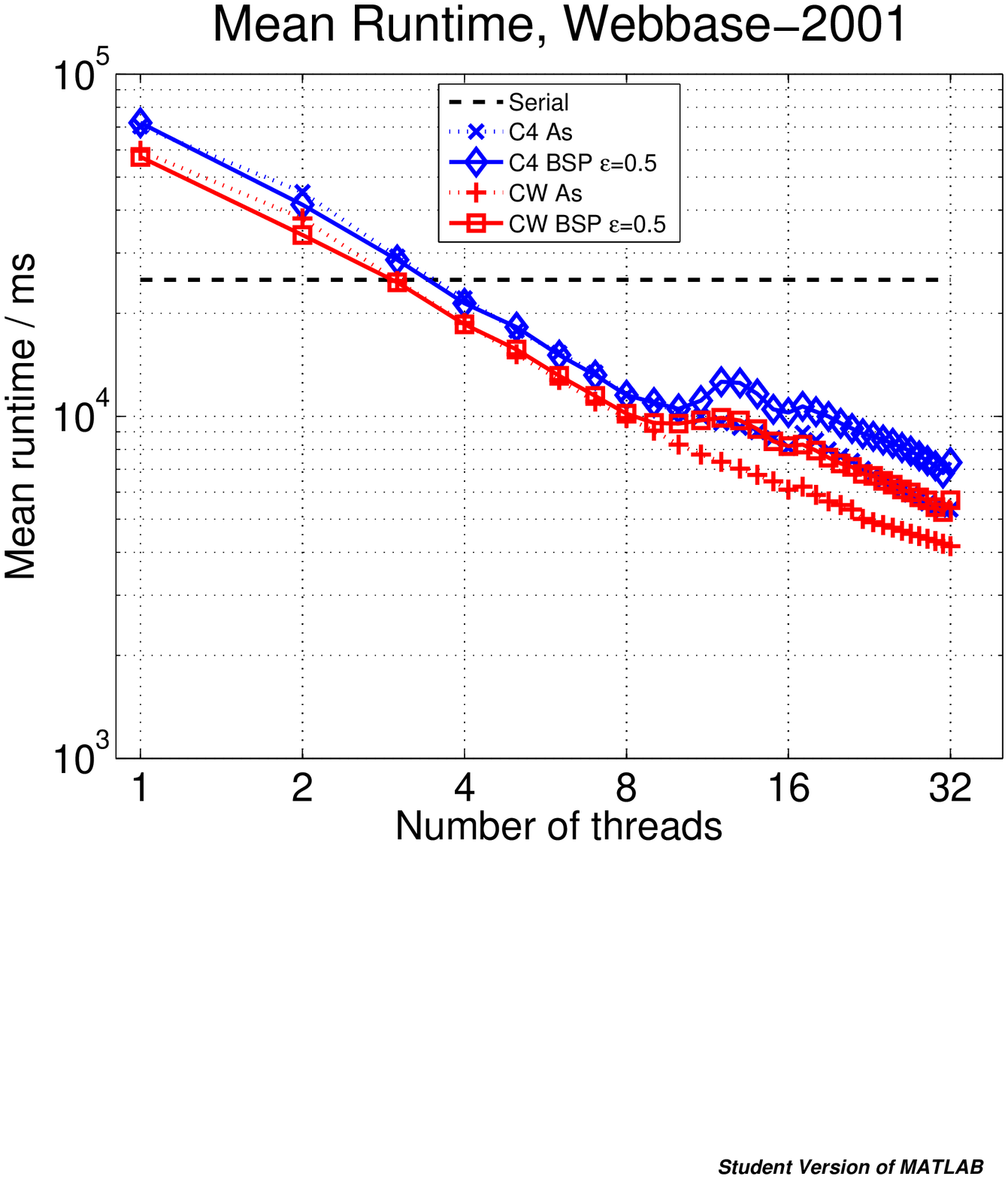}
      \caption{Webbase-2001, $\epsilon = 0.5$}
      \label{appfig:runtimes_wb01_ave_05}
    \end{subfigure} &
    \begin{subfigure}[b]{0.31\textwidth}
      \includegraphics[width=120pt]{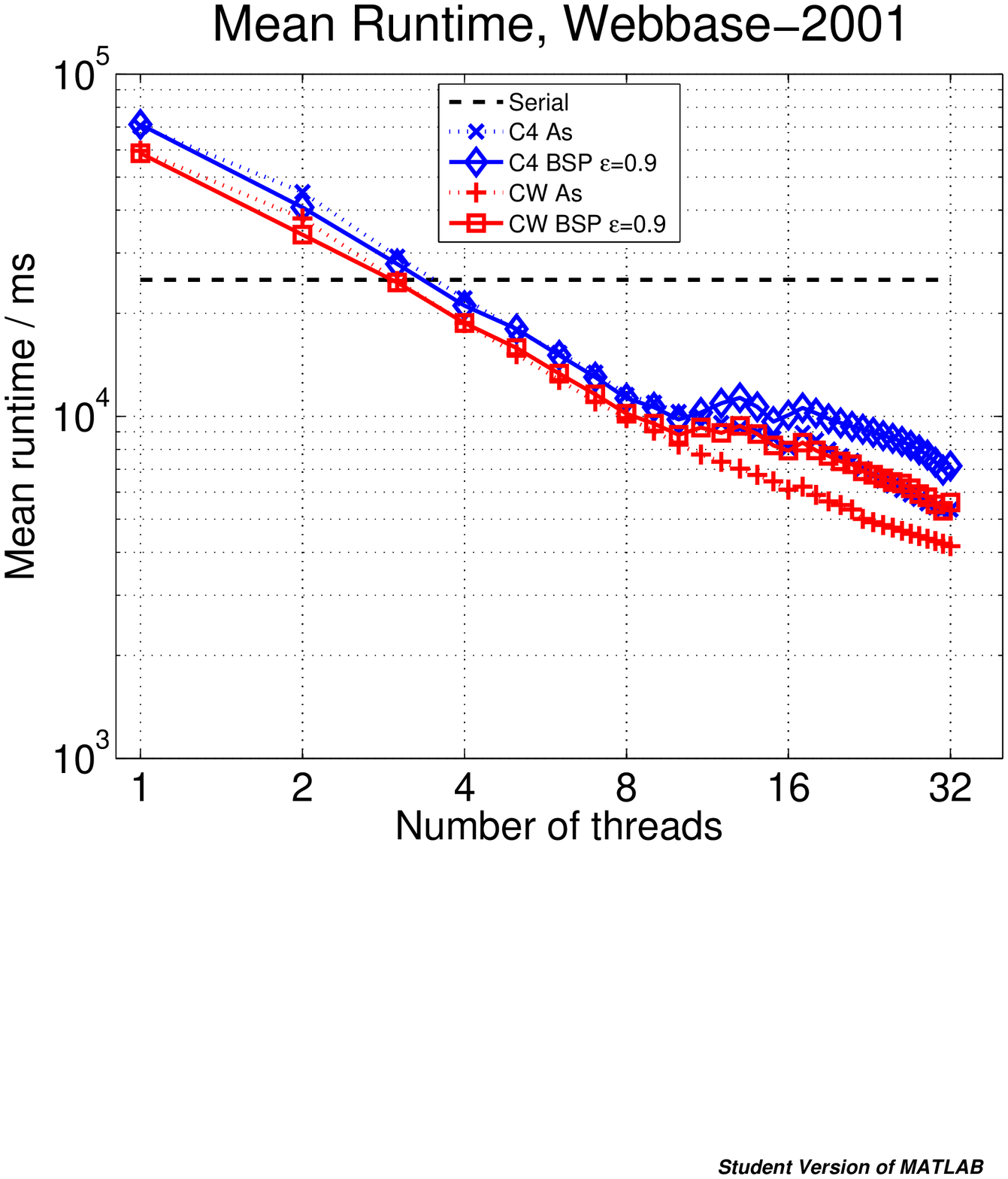}
      \caption{Webbase-2001, $\epsilon = 0.9$}
      \label{appfig:runtimes_wb01_ave_09}
    \end{subfigure} \\

    \begin{subfigure}[b]{0.31\textwidth}
      \includegraphics[width=120pt]{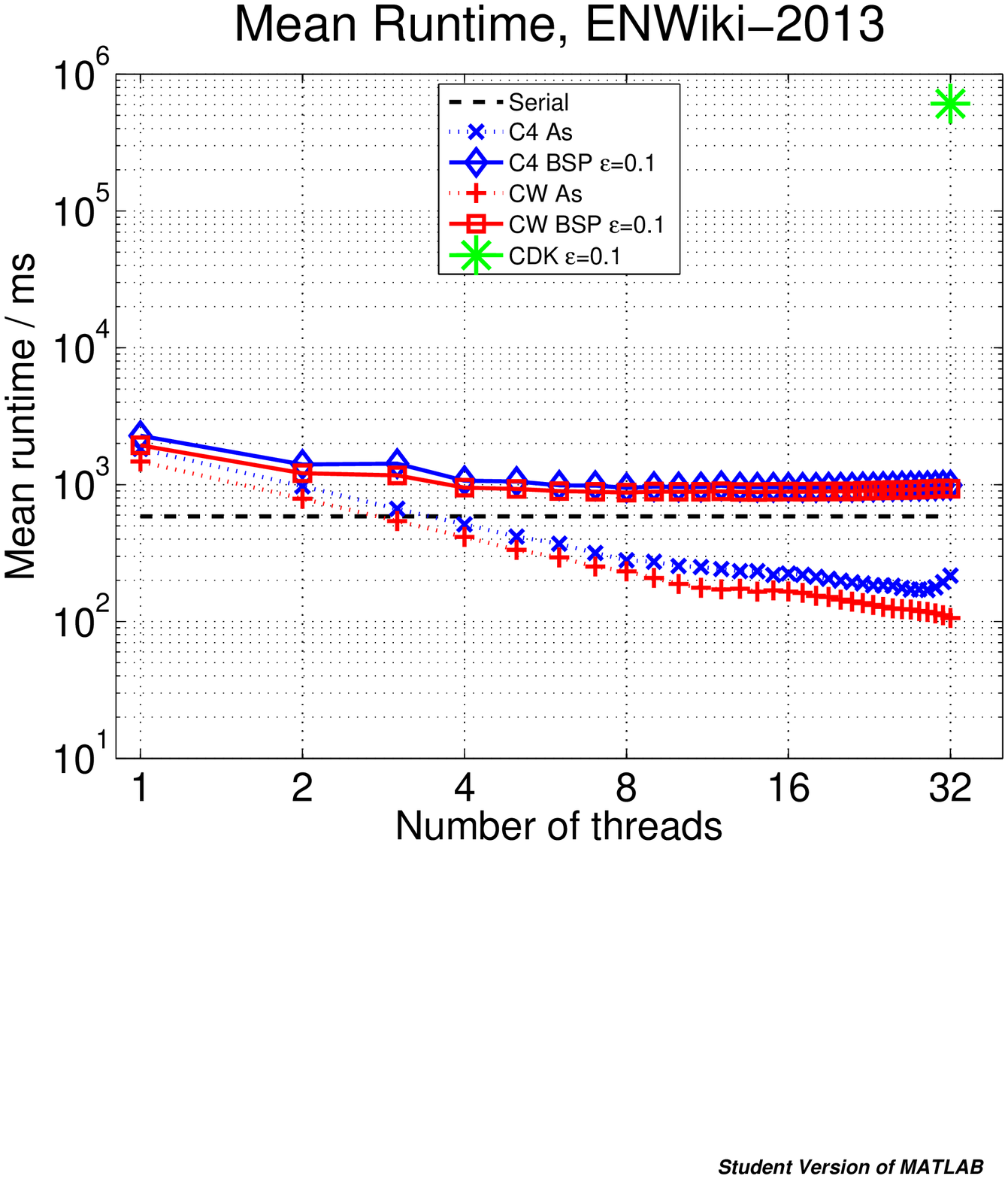}
      \caption{ENWiki-2013, $\epsilon = 0.1$}
      \label{appfig:runtimes_ew13_ave_01}
    \end{subfigure} &
    \begin{subfigure}[b]{0.31\textwidth}
      \includegraphics[width=120pt]{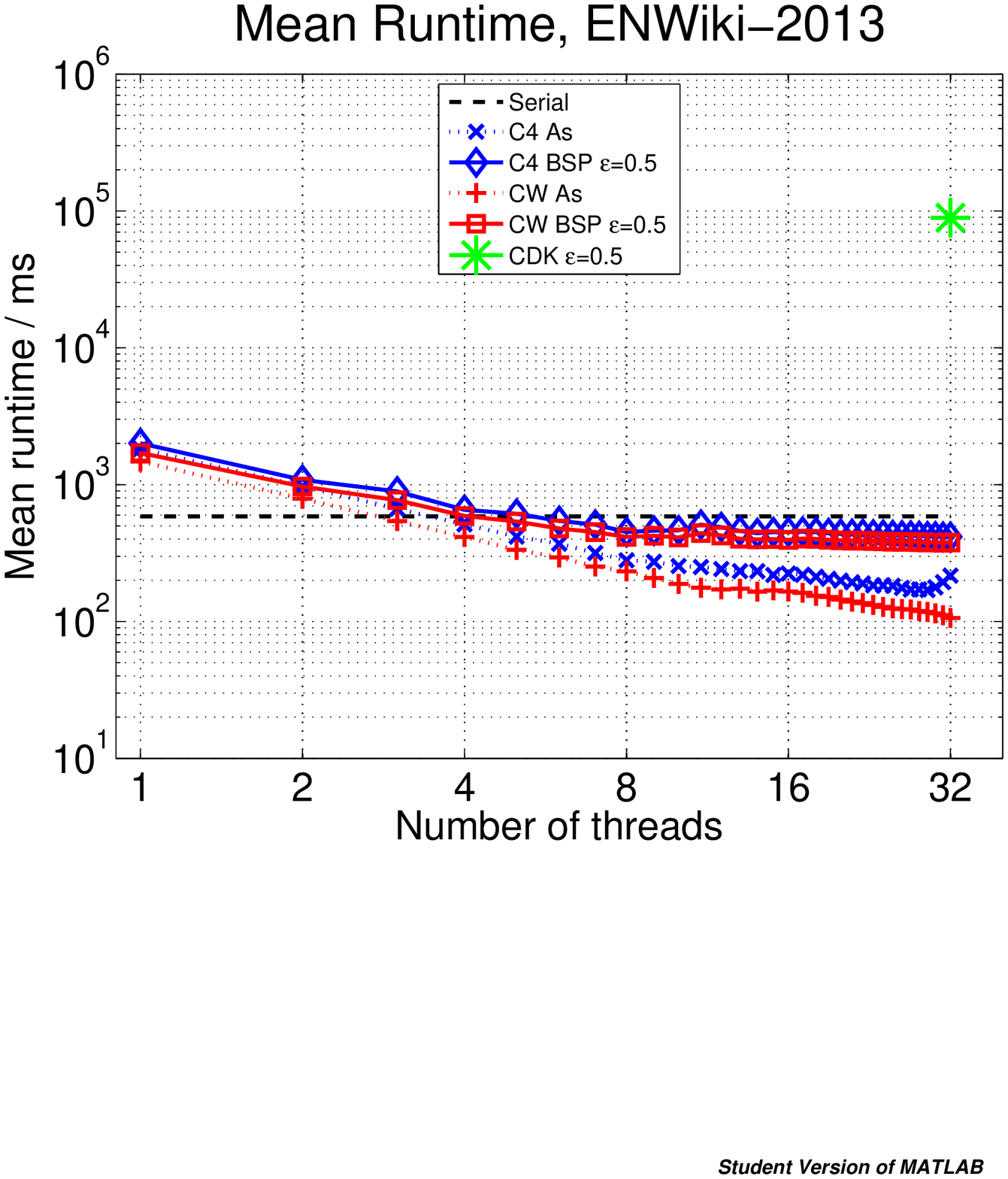}
      \caption{ENWiki-2013, $\epsilon = 0.5$}
      \label{appfig:runtimes_ew13_ave_05}
    \end{subfigure} &
    \begin{subfigure}[b]{0.31\textwidth}
      \includegraphics[width=120pt]{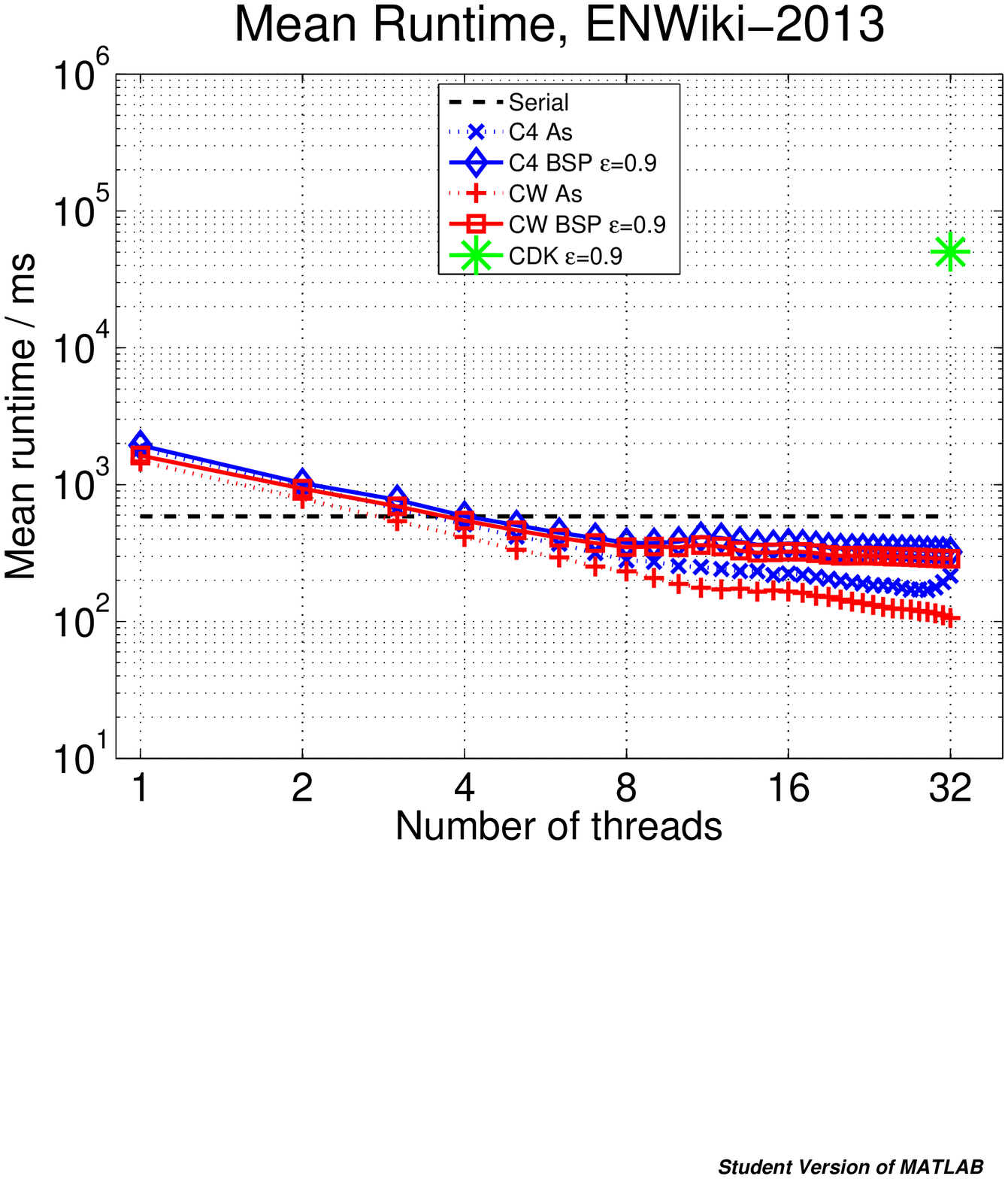}
      \caption{ENWiki-2013, $\epsilon = 0.9$}
      \label{appfig:runtimes_ew13_ave_09}
    \end{subfigure} \\

    \begin{subfigure}[b]{0.31\textwidth}
      \includegraphics[width=120pt]{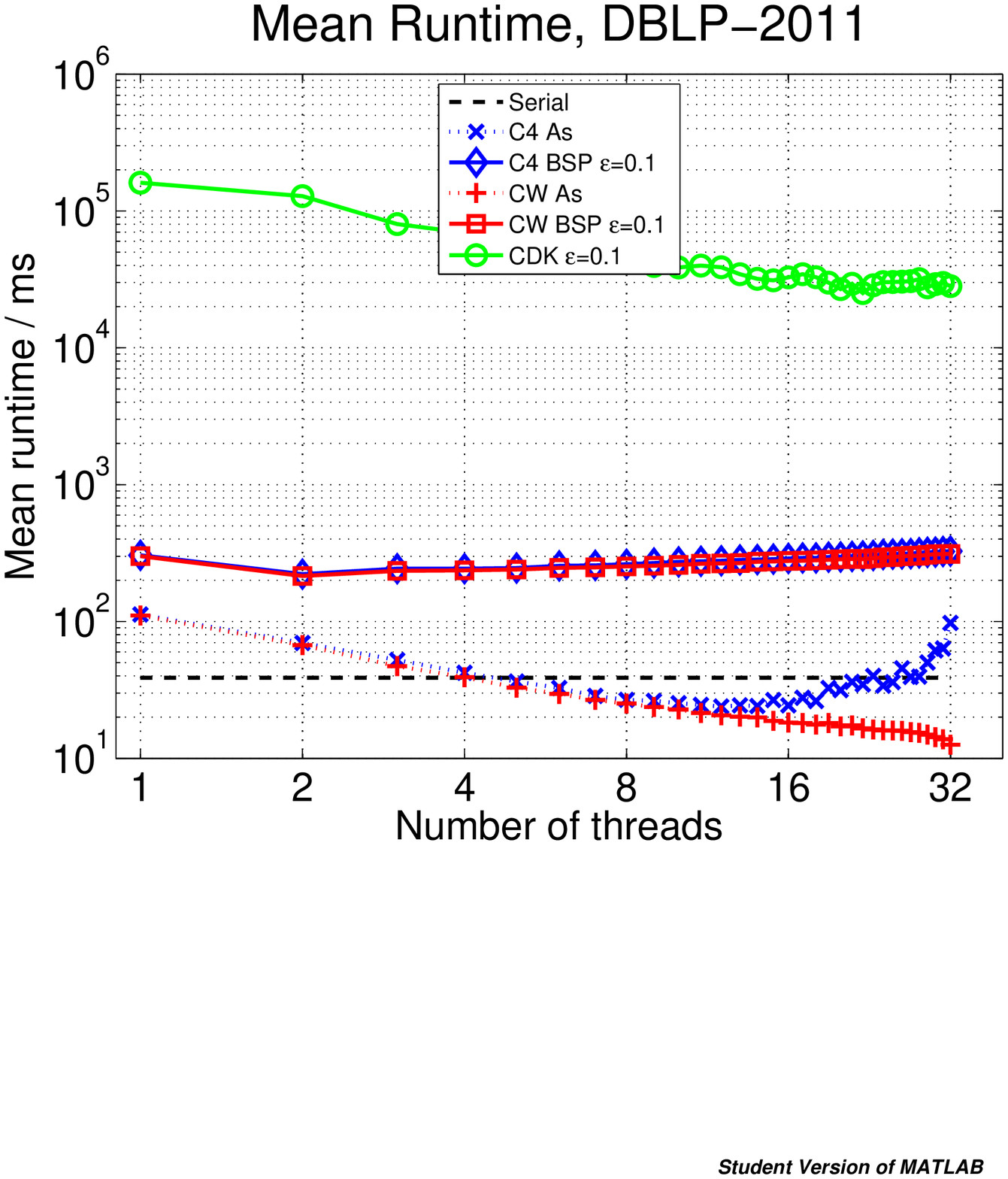}
      \caption{DBLP-2011, $\epsilon = 0.1$}
      \label{appfig:runtimes_db11_ave_01}
    \end{subfigure} &
    \begin{subfigure}[b]{0.31\textwidth}
      \includegraphics[width=120pt]{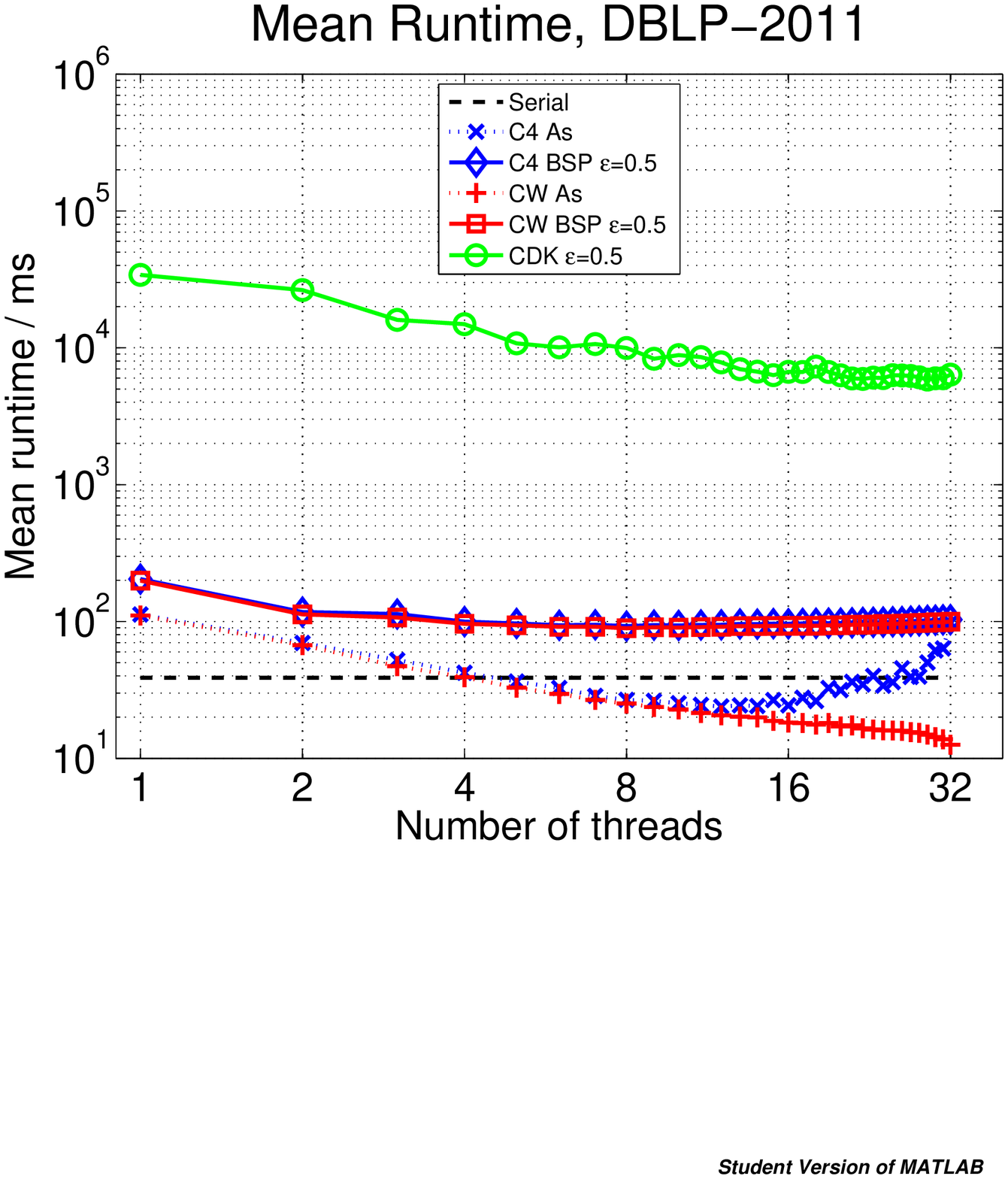}
      \caption{DBLP-2011, $\epsilon = 0.5$}
      \label{appfig:runtimes_db11_ave_05}
    \end{subfigure} &
    \begin{subfigure}[b]{0.31\textwidth}
      \includegraphics[width=120pt]{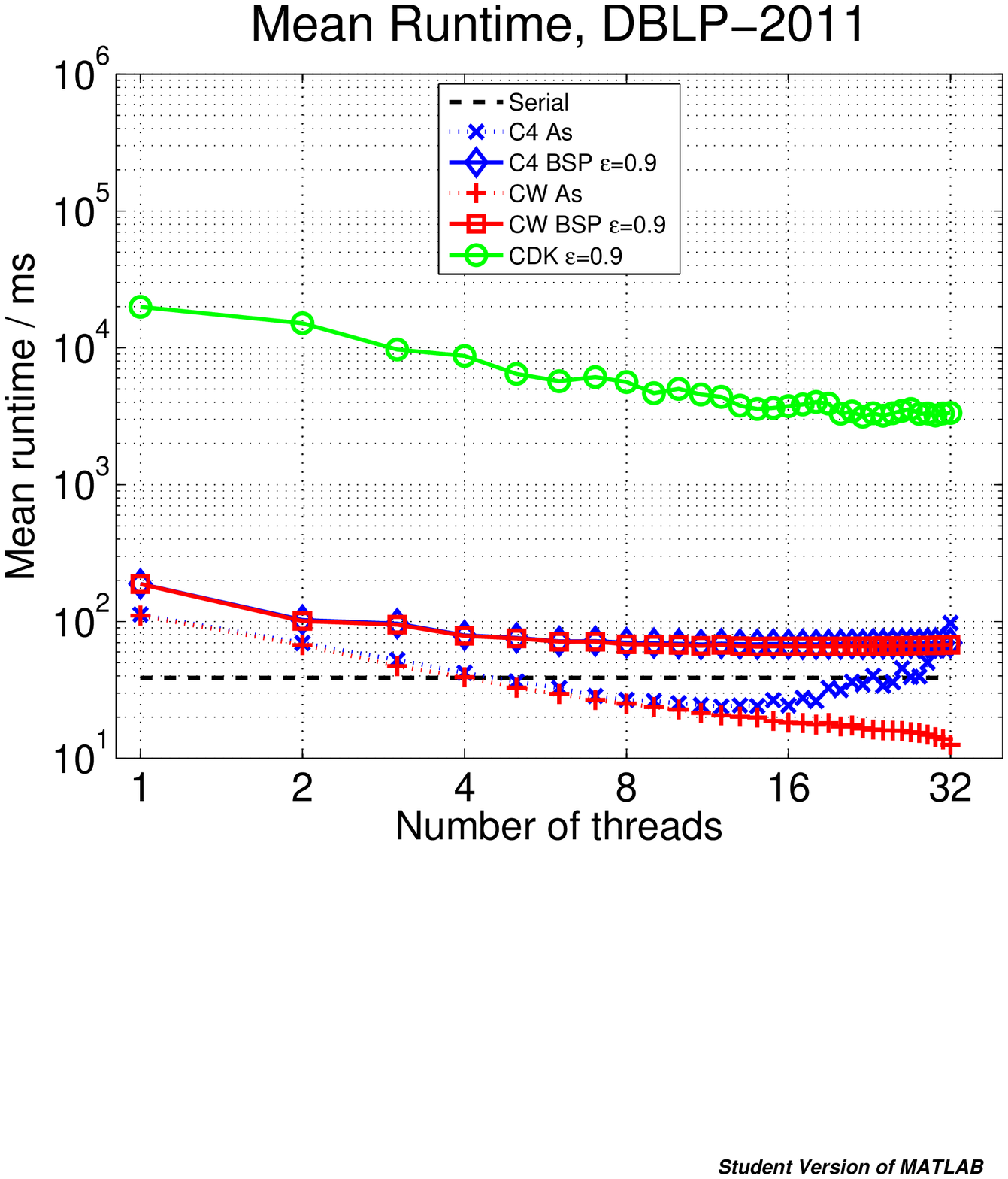}
      \caption{DBLP-2011, $\epsilon = 0.9$}
      \label{appfig:runtimes_db11_ave_09}
    \end{subfigure} \\

  \end{tabular}
  \caption{\scriptsize
  Empirical mean runtimes.
  For short, `CW' is \CW{} and `As' refers to the asynchronous variants.
  On larger graphs, our parallel algorithms on 3-4 threads are faster than serial \KC{}.
  On the smaller graphs, the BSP variants have expensive synchronization barriers (relative to the small amount of actual work done) and do not necessary run faster than serial \KC{}; the asynchronous variants do outperform serial \KC{} with 4-5 threads.
  We were only able to run CDK on the smaller graphs, for which CDK was 2-3 orders of magnitude slower than serial.
  Note also that the BSP variants have improved runtimes for larger $\epsilon$.}
\end{figure}

\begin{figure}[ht]
  \centering
  \begin{tabular}{ccc}

    \begin{subfigure}[b]{0.31\textwidth}
      \includegraphics[width=120pt]{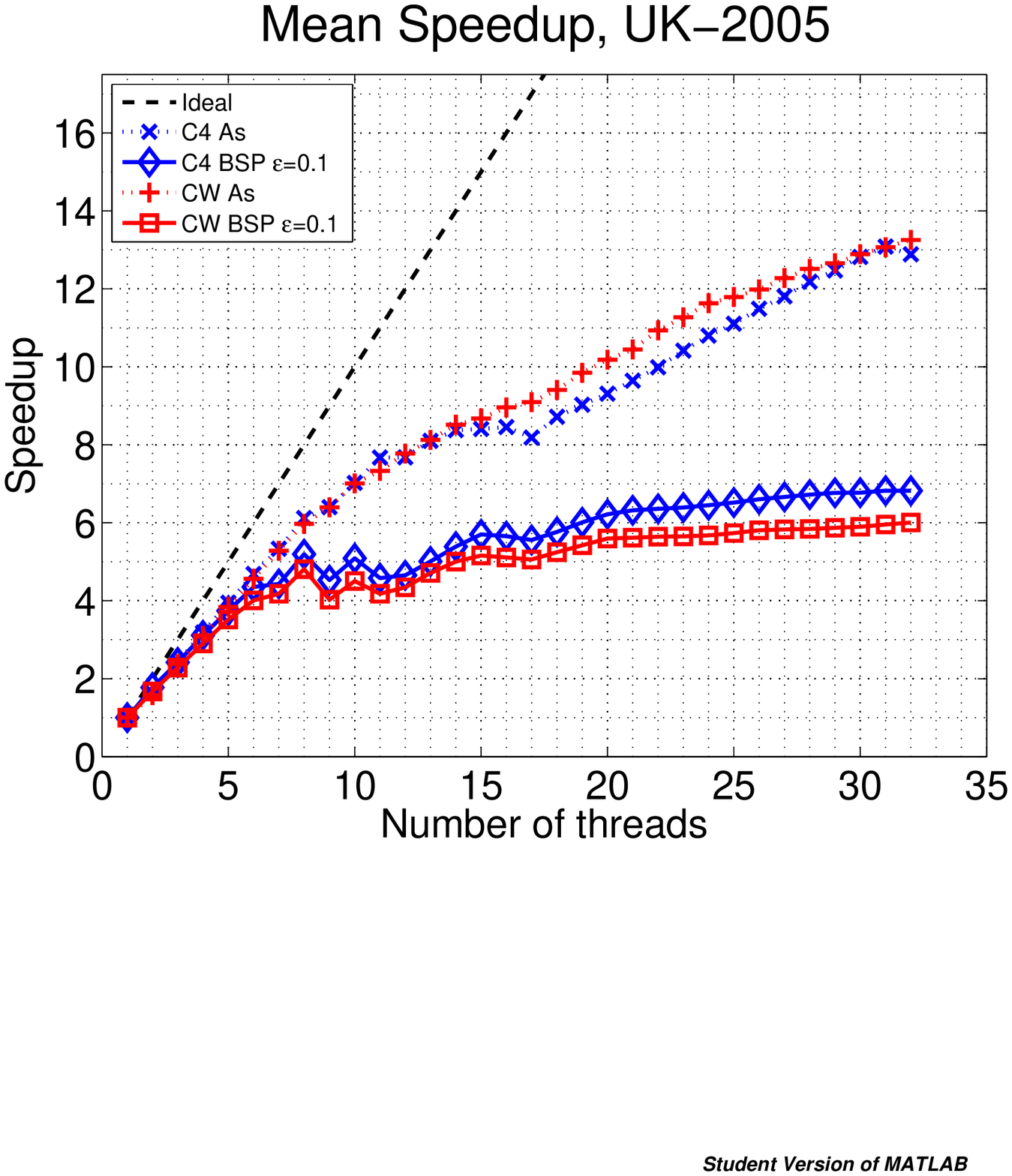}
      \caption{UK-2005, $\epsilon = 0.1$}
      \label{appfig:speedups_uk05_01}
    \end{subfigure} &
    \begin{subfigure}[b]{0.31\textwidth}
      \includegraphics[width=120pt]{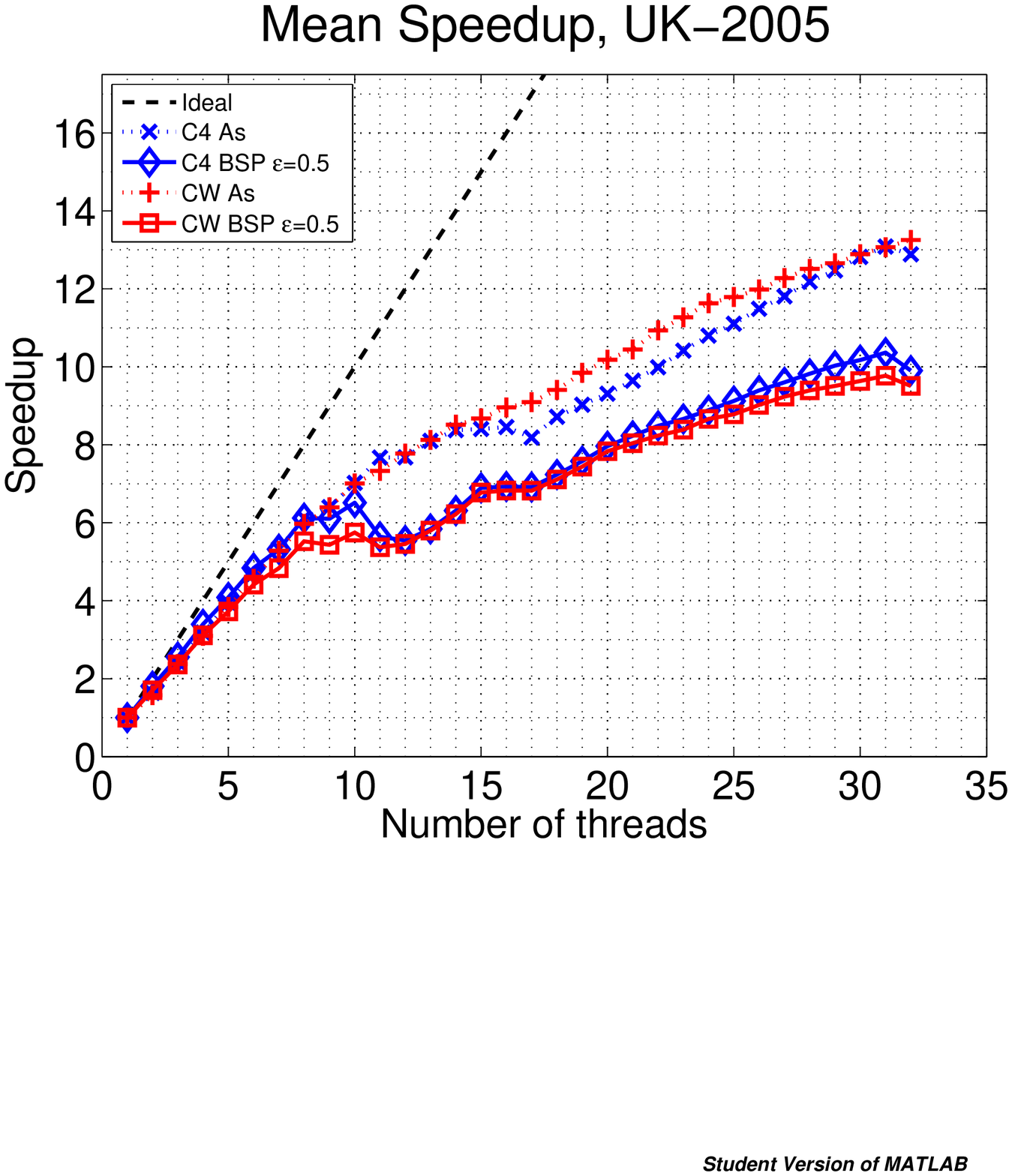}
      \caption{UK-2005, $\epsilon = 0.5$}
      \label{appfig:speedups_uk05_05}
    \end{subfigure} &
    \begin{subfigure}[b]{0.31\textwidth}
      \includegraphics[width=120pt]{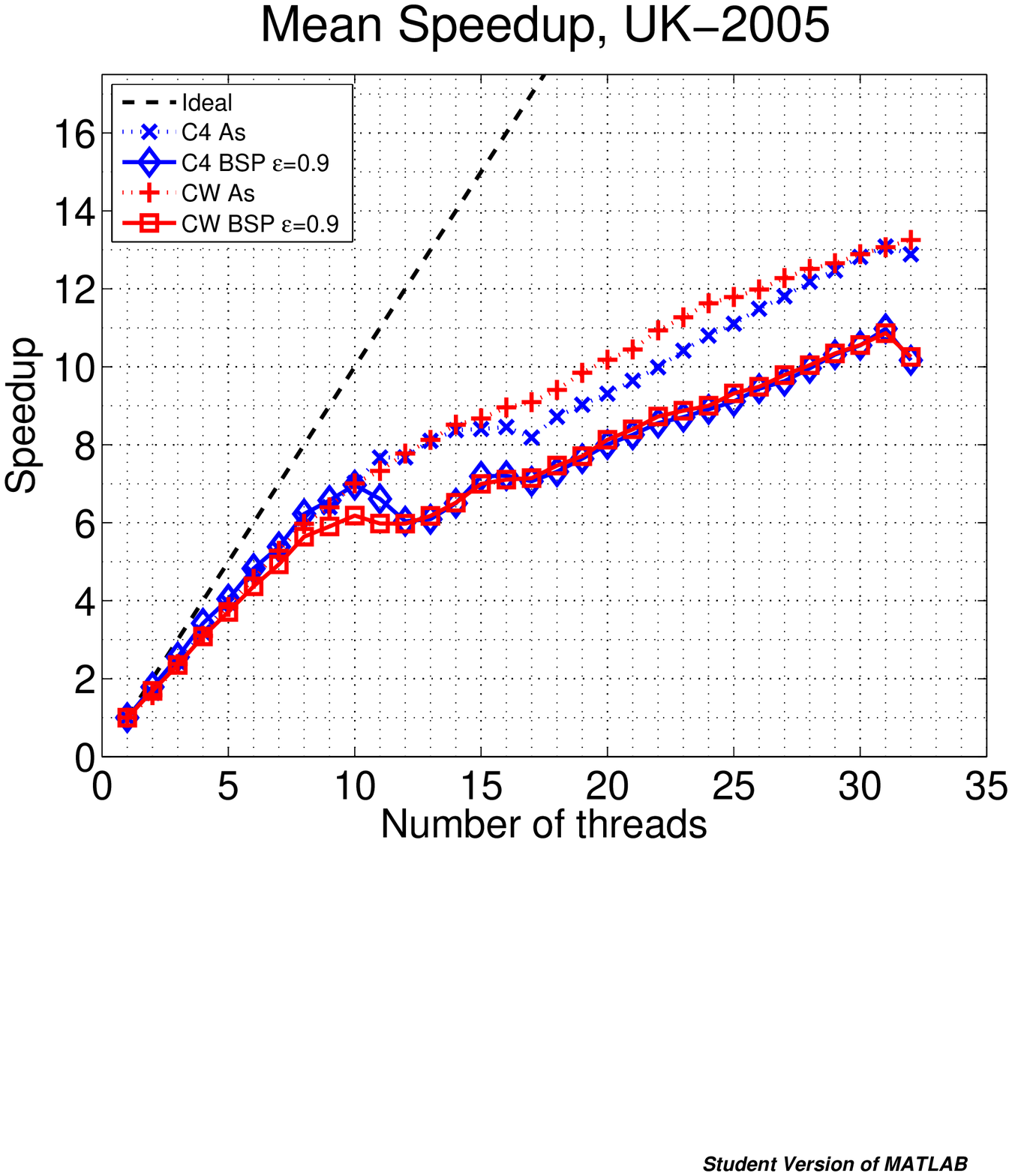}
      \caption{UK-2005, $\epsilon = 0.9$}
      \label{appfig:speedups_uk05_09}
    \end{subfigure} \\

    \begin{subfigure}[b]{0.31\textwidth}
      \includegraphics[width=120pt]{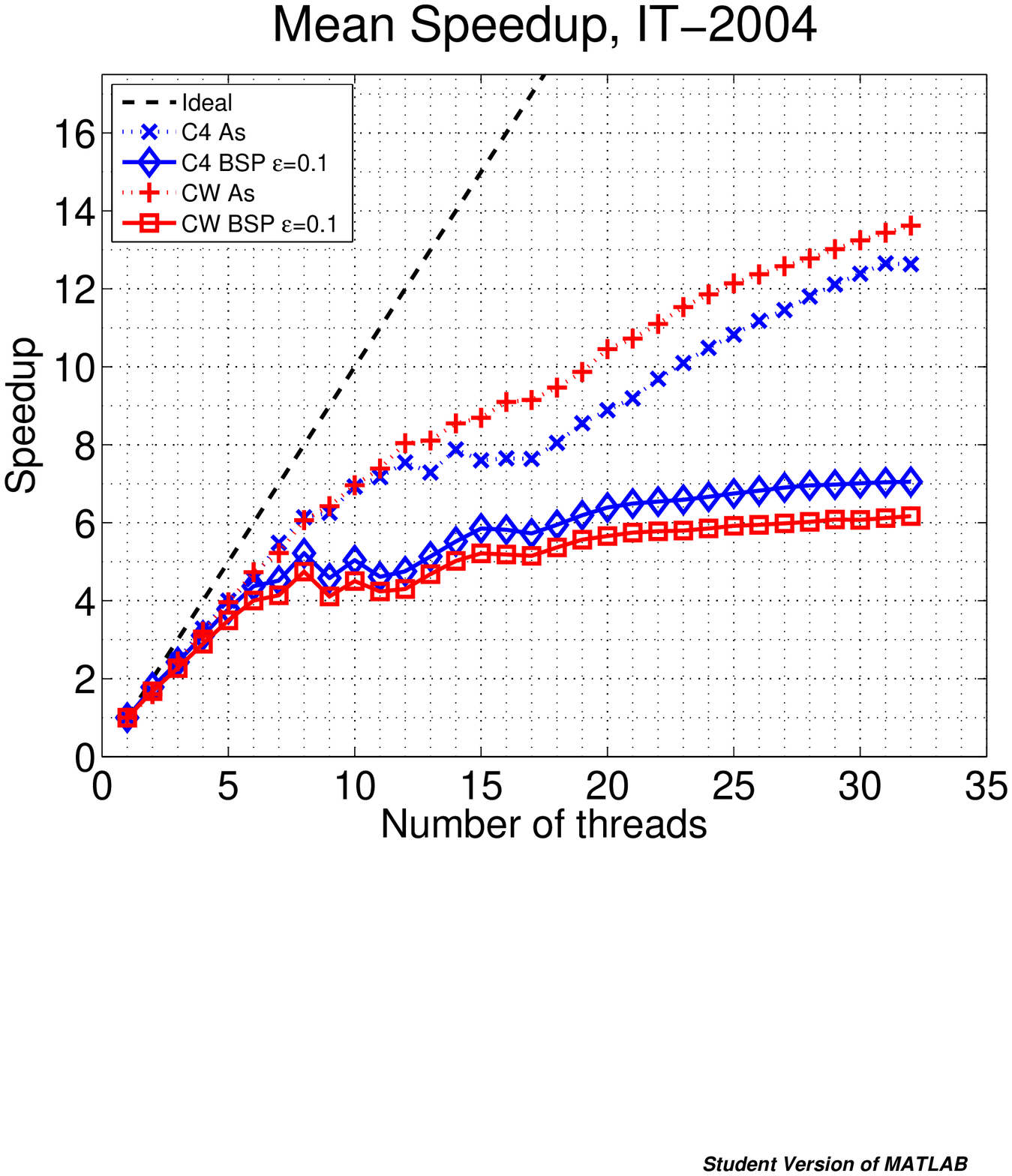}
      \caption{IT-2004, $\epsilon = 0.1$}
      \label{appfig:speedups_it04_01}
    \end{subfigure} &
    \begin{subfigure}[b]{0.31\textwidth}
      \includegraphics[width=120pt]{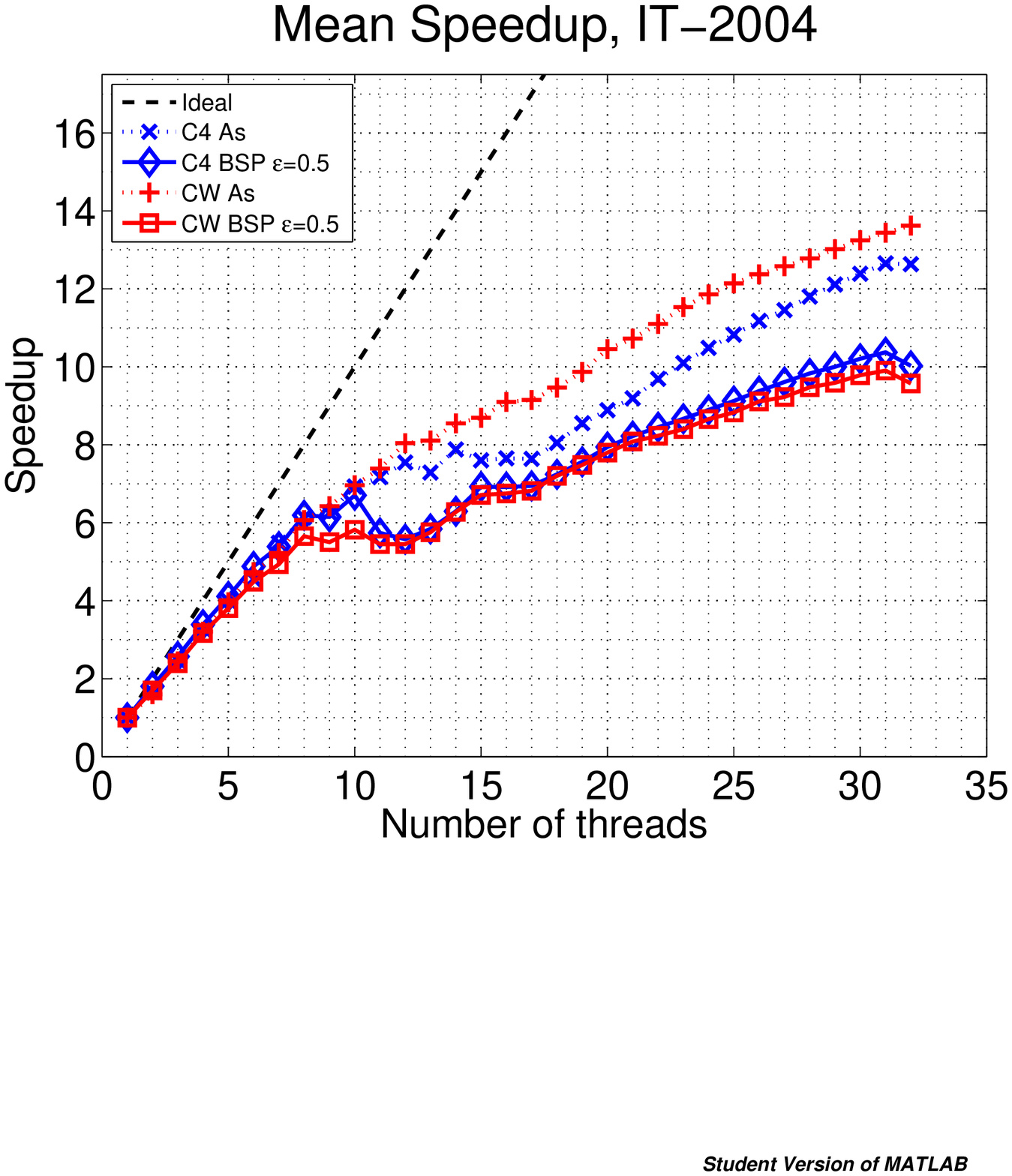}
      \caption{IT-2004, $\epsilon = 0.5$}
      \label{appfig:speedups_it04_05}
    \end{subfigure} &
    \begin{subfigure}[b]{0.31\textwidth}
      \includegraphics[width=120pt]{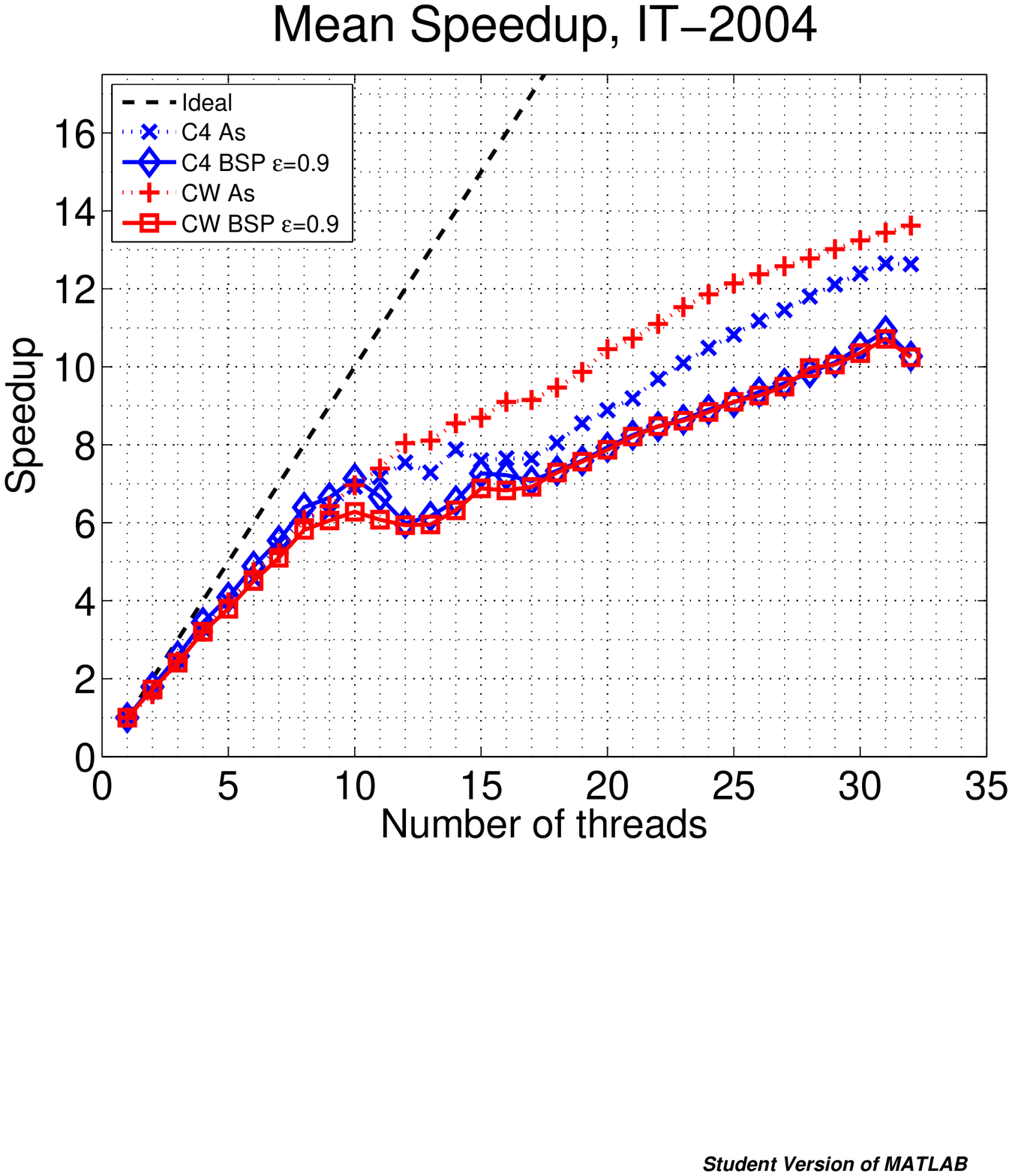}
      \caption{IT-2004, $\epsilon = 0.9$}
      \label{appfig:speedups_it04_09}
    \end{subfigure} \\

    \begin{subfigure}[b]{0.31\textwidth}
      \includegraphics[width=120pt]{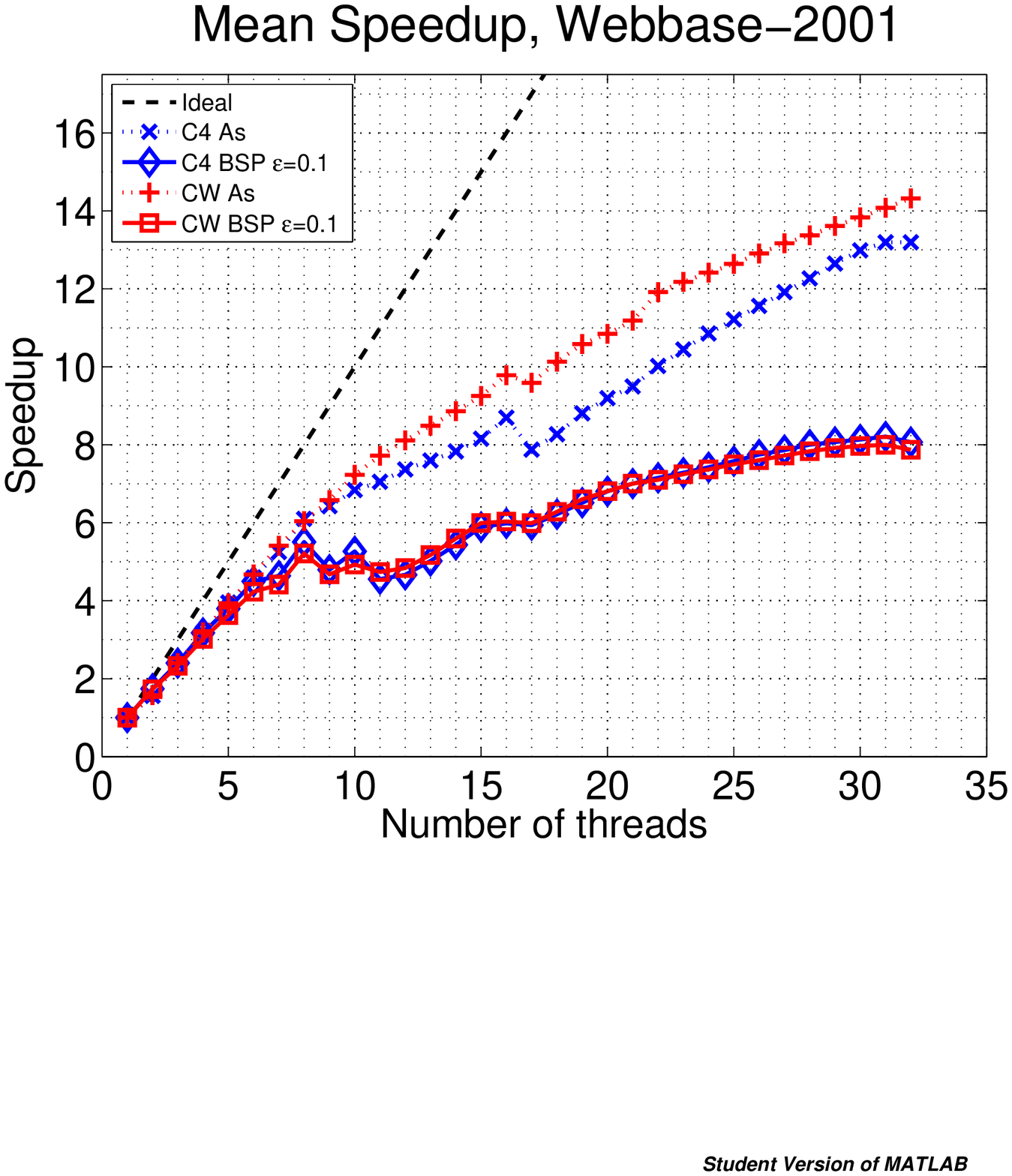}
      \caption{Webbase-2001, $\epsilon = 0.1$}
      \label{appfig:speedups_wb01_01}
    \end{subfigure} &
    \begin{subfigure}[b]{0.31\textwidth}
      \includegraphics[width=120pt]{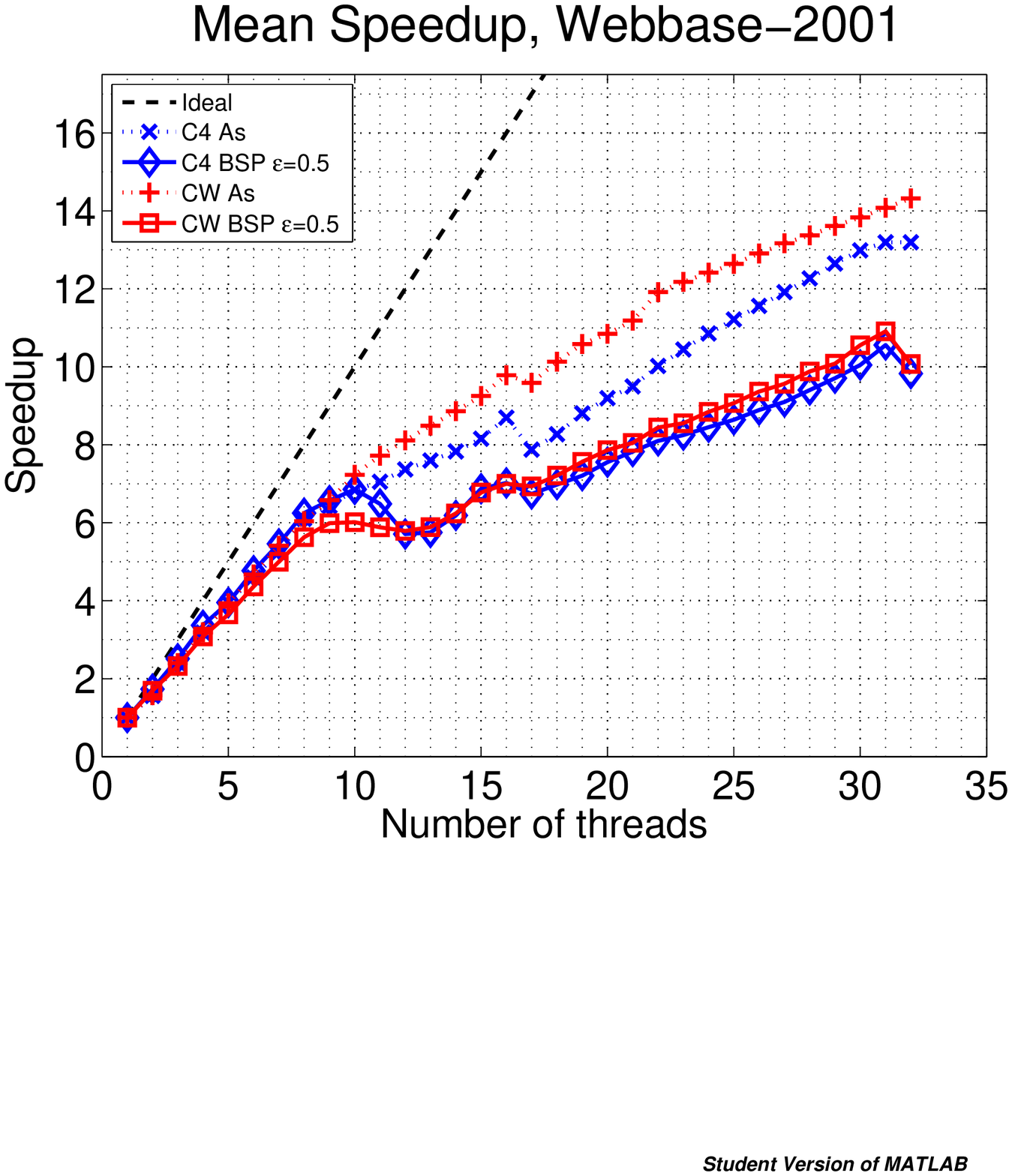}
      \caption{Webbase-2001, $\epsilon = 0.5$}
      \label{appfig:speedups_wb01_05}
    \end{subfigure} &
    \begin{subfigure}[b]{0.31\textwidth}
      \includegraphics[width=120pt]{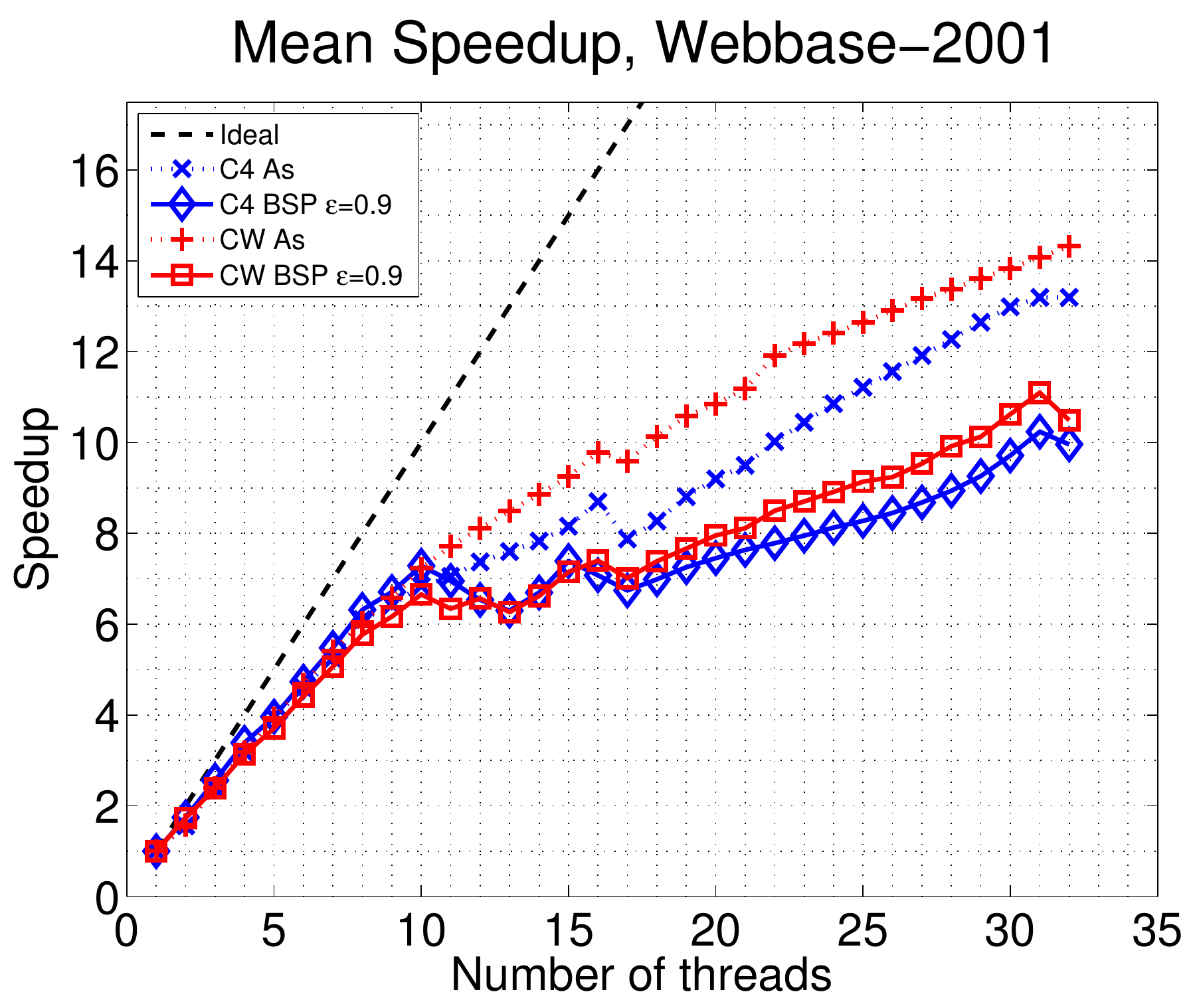}
      \caption{Webbase-2001, $\epsilon = 0.9$}
      \label{appfig:speedups_wb01_09}
    \end{subfigure} \\

    \begin{subfigure}[b]{0.31\textwidth}
      \includegraphics[width=120pt]{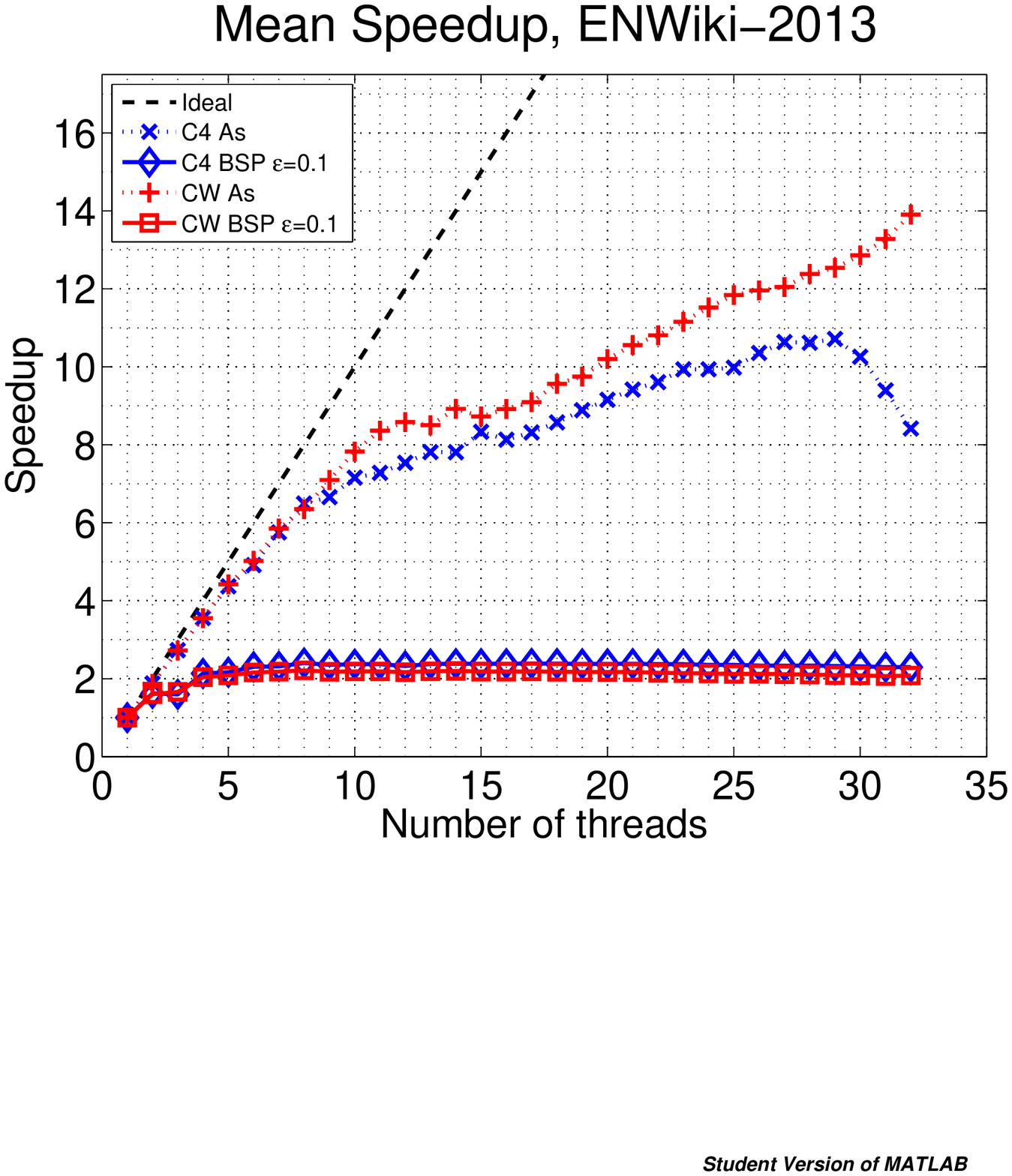}
      \caption{ENWiki-2013, $\epsilon = 0.1$}
      \label{appfig:speedups_ew13_01}
    \end{subfigure} &
    \begin{subfigure}[b]{0.31\textwidth}
      \includegraphics[width=120pt]{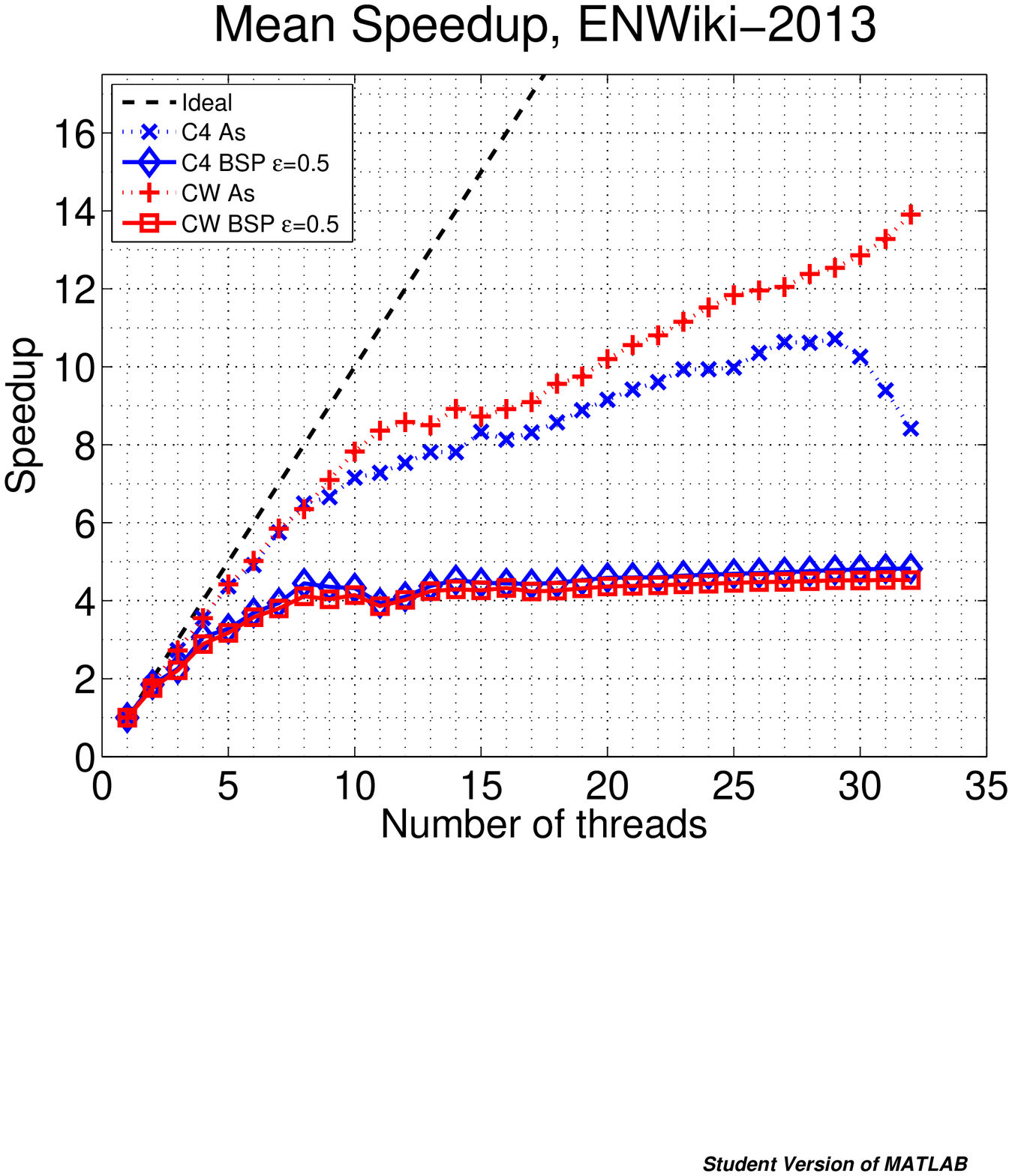}
      \caption{ENWiki-2013, $\epsilon = 0.5$}
      \label{appfig:speedups_ew13_05}
    \end{subfigure} &
    \begin{subfigure}[b]{0.31\textwidth}
      \includegraphics[width=120pt]{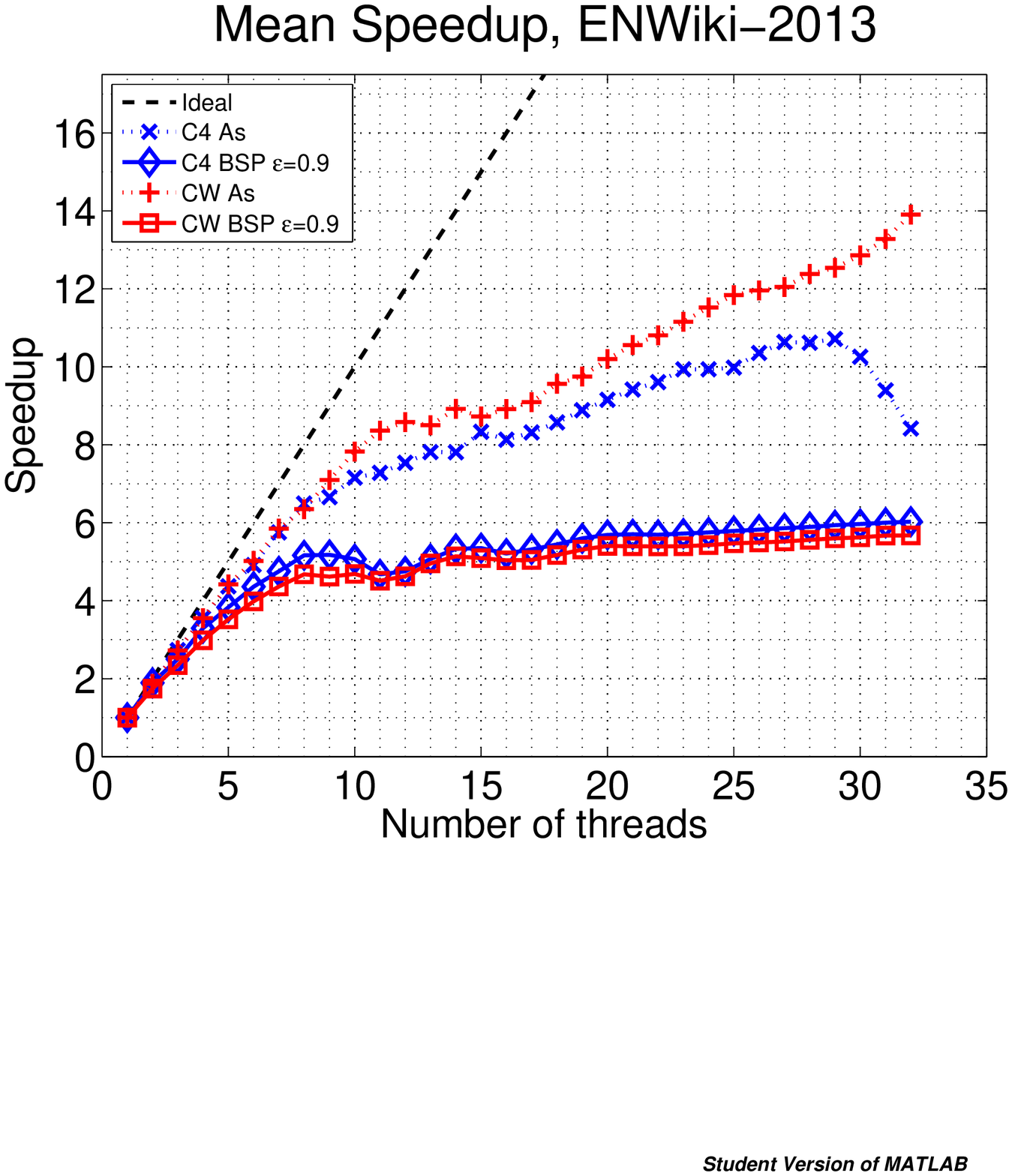}
      \caption{ENWiki-2013, $\epsilon = 0.9$}
      \label{appfig:speedups_ew13_09}
    \end{subfigure} \\

    \begin{subfigure}[b]{0.31\textwidth}
      \includegraphics[width=120pt]{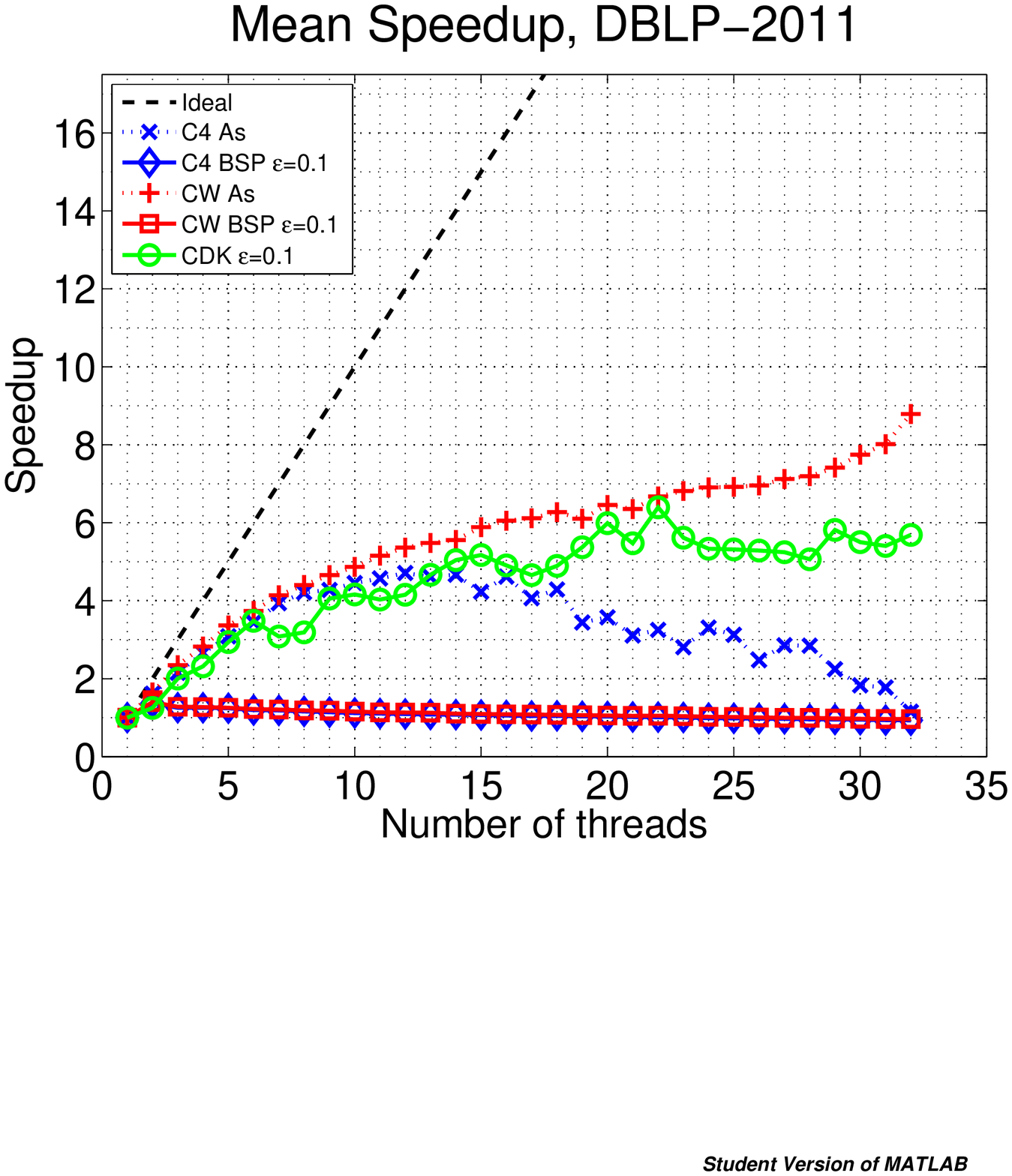}
      \caption{DBLP-2011, $\epsilon = 0.1$}
      \label{appfig:speedups_db11_01}
    \end{subfigure} &
    \begin{subfigure}[b]{0.31\textwidth}
      \includegraphics[width=120pt]{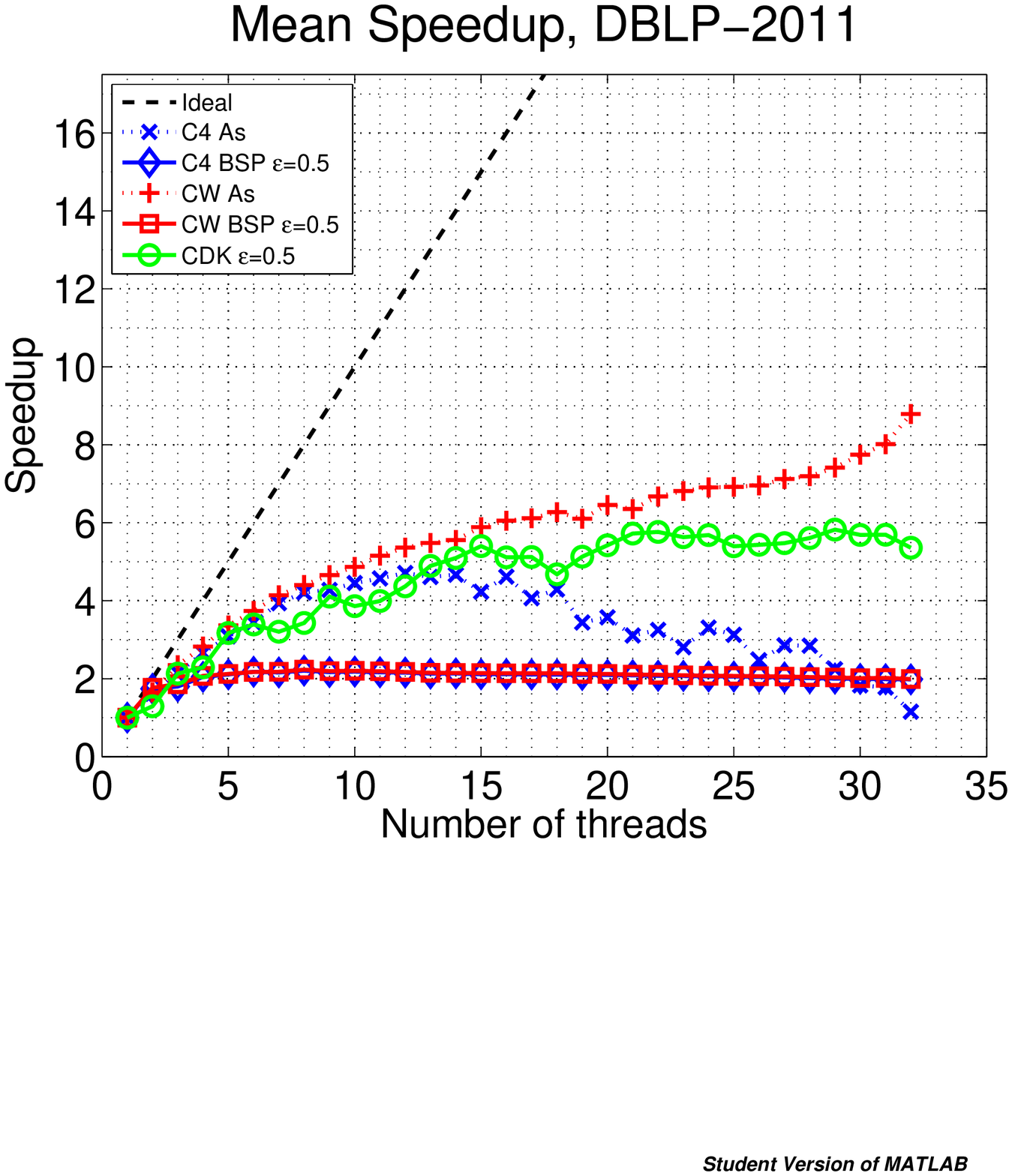}
      \caption{DBLP-2011, $\epsilon = 0.5$}
      \label{appfig:speedups_db11_05}
    \end{subfigure} &
    \begin{subfigure}[b]{0.31\textwidth}
      \includegraphics[width=120pt]{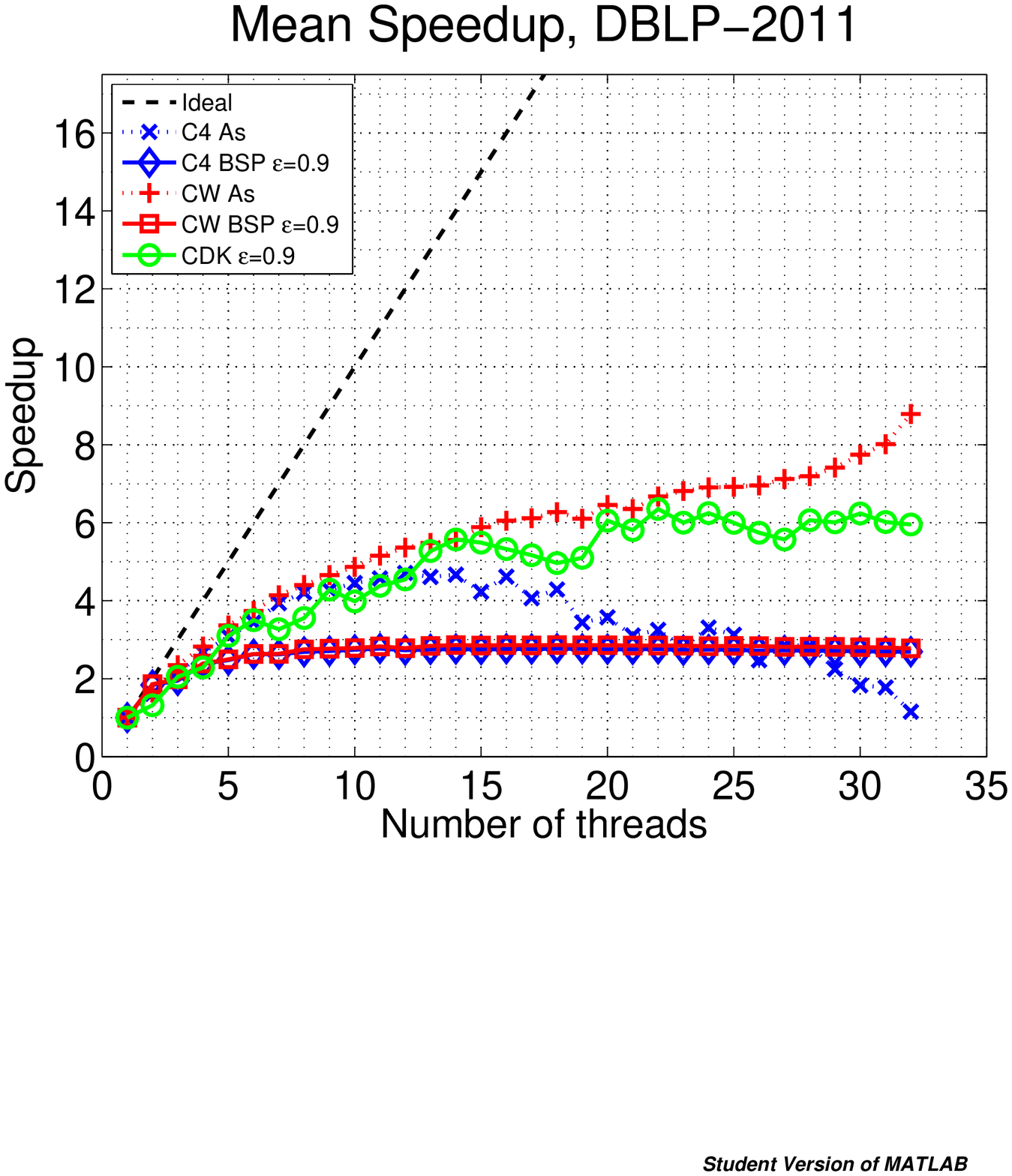}
      \caption{DBLP-2011, $\epsilon = 0.9$}
      \label{appfig:speedups_db11_09}
    \end{subfigure} \\

  \end{tabular}
  \caption{\scriptsize
  Empirical mean speedups.
  The best speedups (14x on large graphs) are achieved by asynchronous \CW{} which has the least coordination,
  followed by asynchronous \CC{} (13x on large graphs).
  The BSP variants achieve up to 10x speedups on large graphs, with better speedups as $\epsilon$ increases.
  On small graphs we obtain poorer speedups as the cost of any contention is magnified as the actual work done is comparatively small.
  There are a couple of kinks at 10 and 16 threads, which we postulate is due to NUMA and hyperthreading effects---the EC2 r3.8xlarge instances are equipped with 10-core Intel Xeon E5-2670 v2 (Ivy Bridge) processors with 32 vCPUs and hyperthreading.
  }
\end{figure}

\begin{figure}[ht]
  \centering
  \begin{tabular}{ccc}

    \begin{subfigure}[b]{0.31\textwidth}
      \includegraphics[width=120pt]{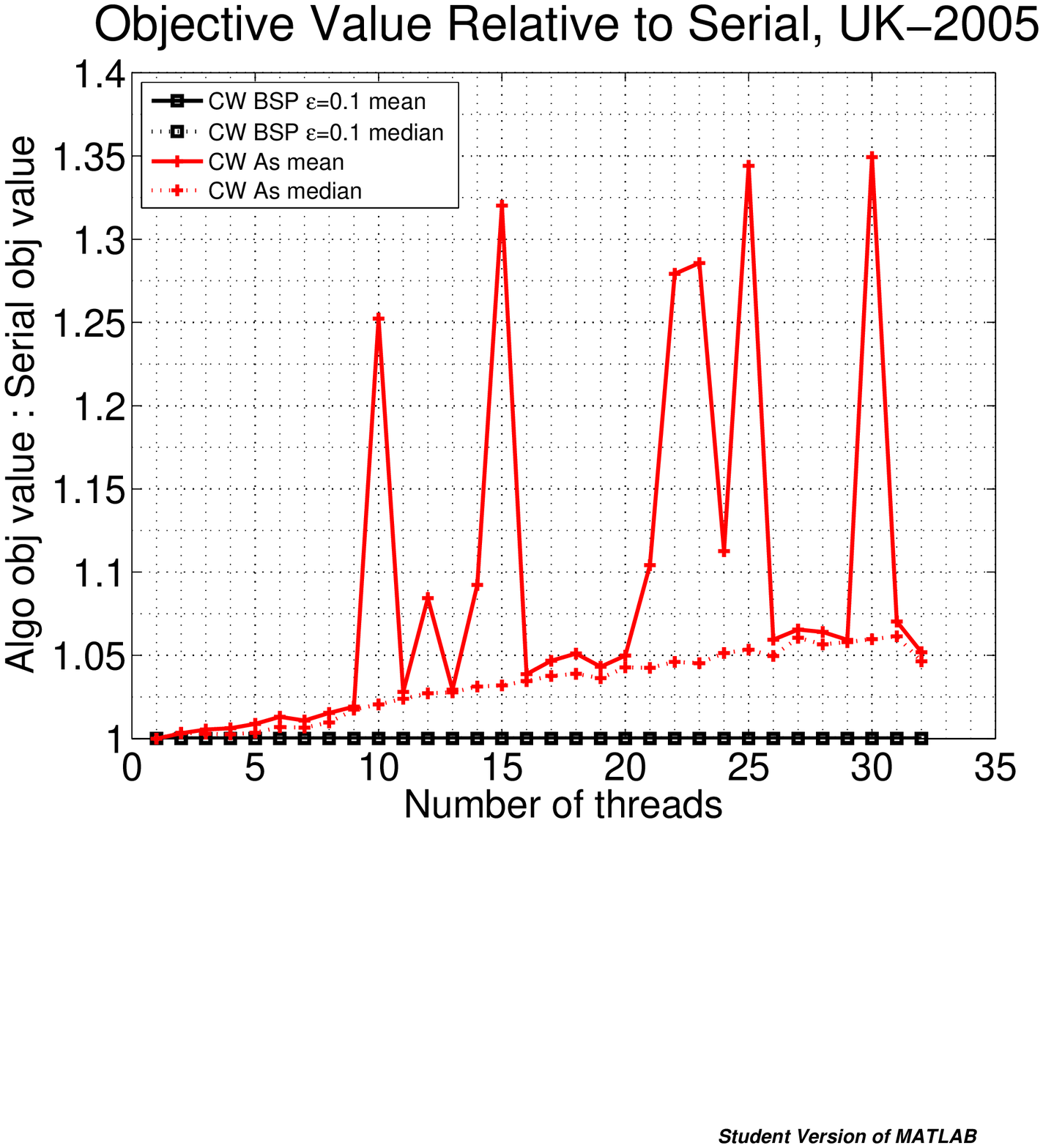}
      \caption{UK-2005, $\epsilon = 0.1$}
      \label{appfig:objvalue_uk05_01}
    \end{subfigure} &
    \begin{subfigure}[b]{0.31\textwidth}
      \includegraphics[width=120pt]{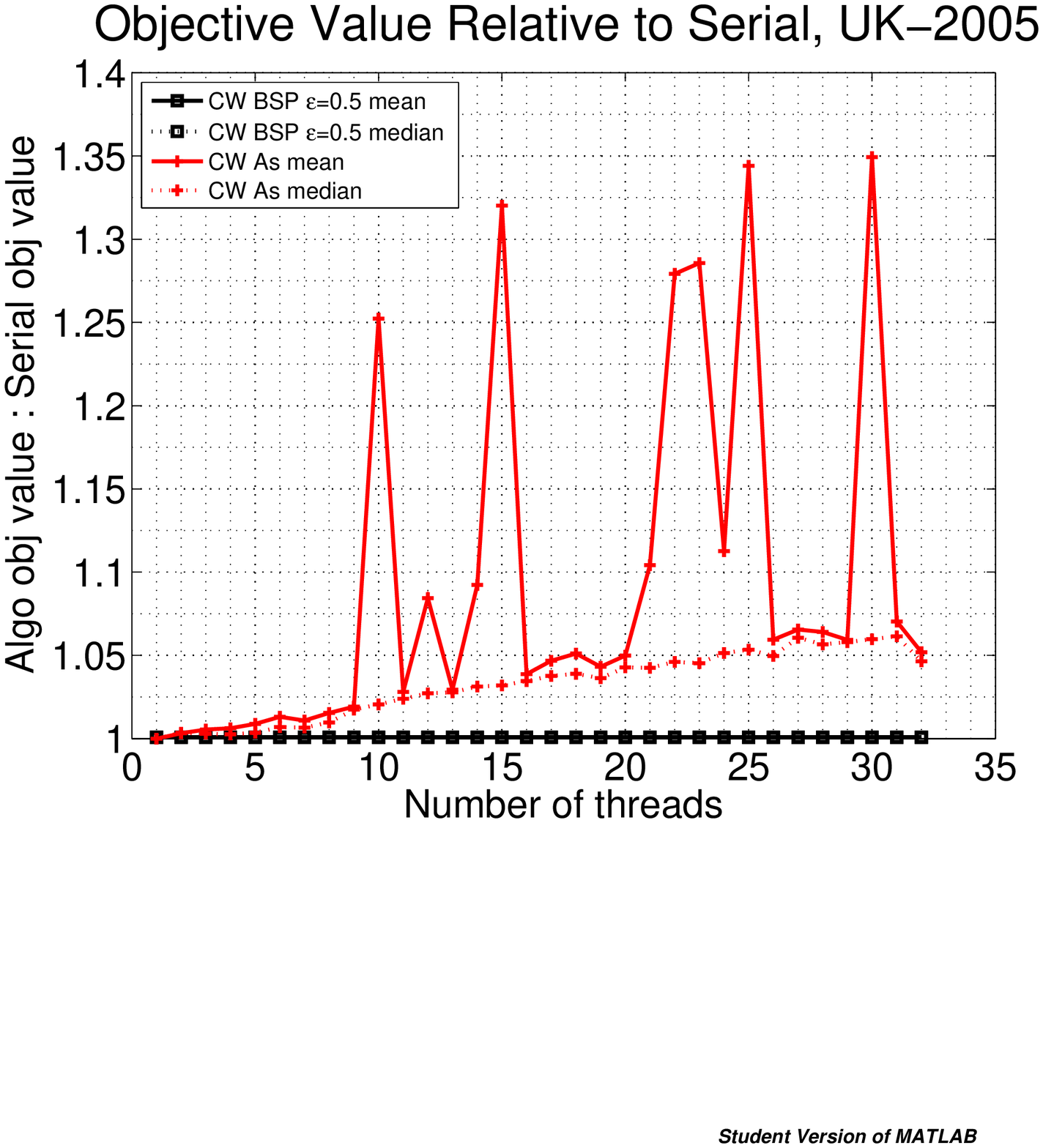}
      \caption{UK-2005, $\epsilon = 0.5$}
      \label{appfig:objvalue_uk05_05}
    \end{subfigure} &
    \begin{subfigure}[b]{0.31\textwidth}
      \includegraphics[width=120pt]{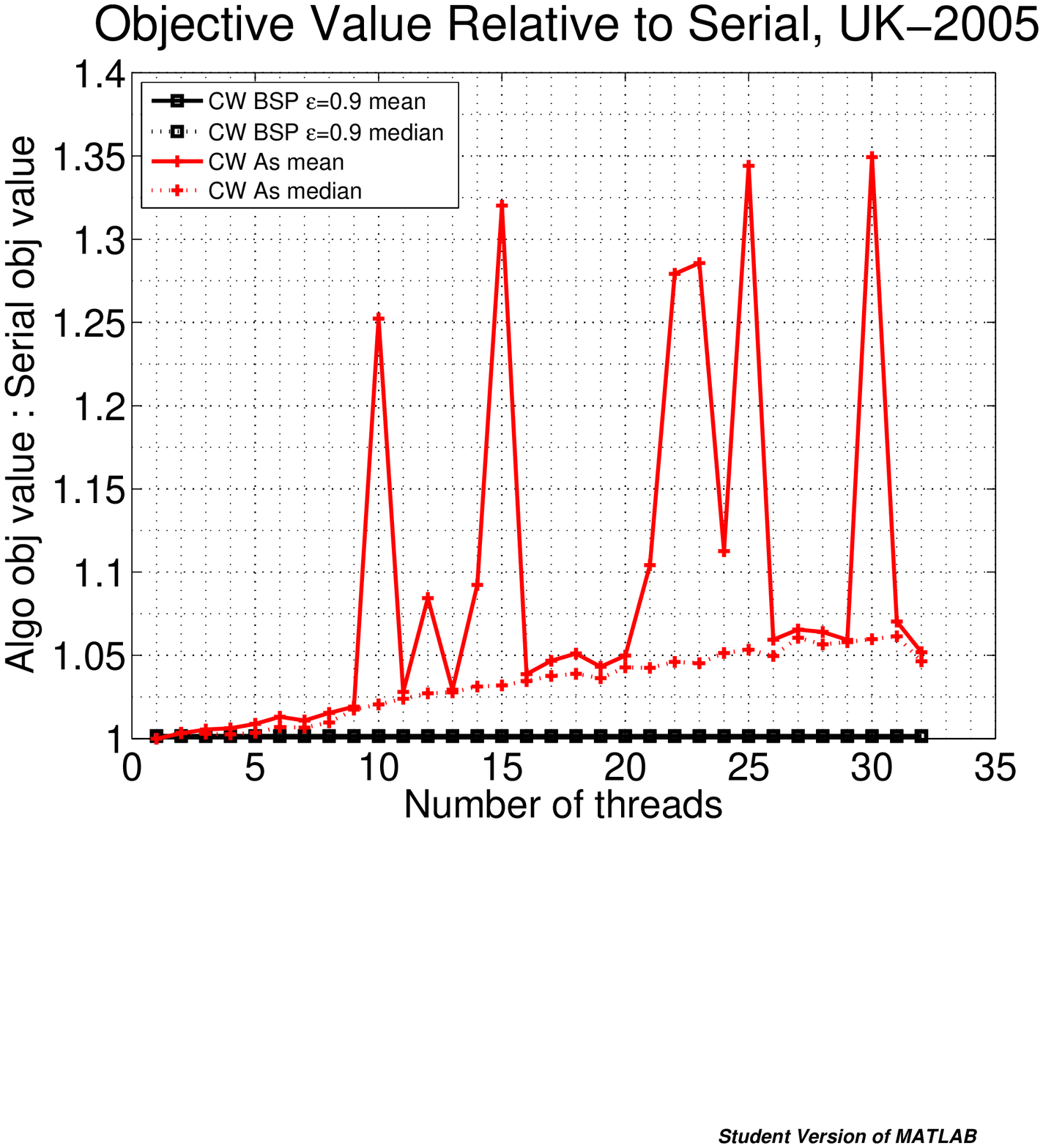}
      \caption{UK-2005, $\epsilon = 0.9$}
      \label{appfig:objvalue_uk05_09}
    \end{subfigure} \\

    \begin{subfigure}[b]{0.31\textwidth}
      \includegraphics[width=120pt]{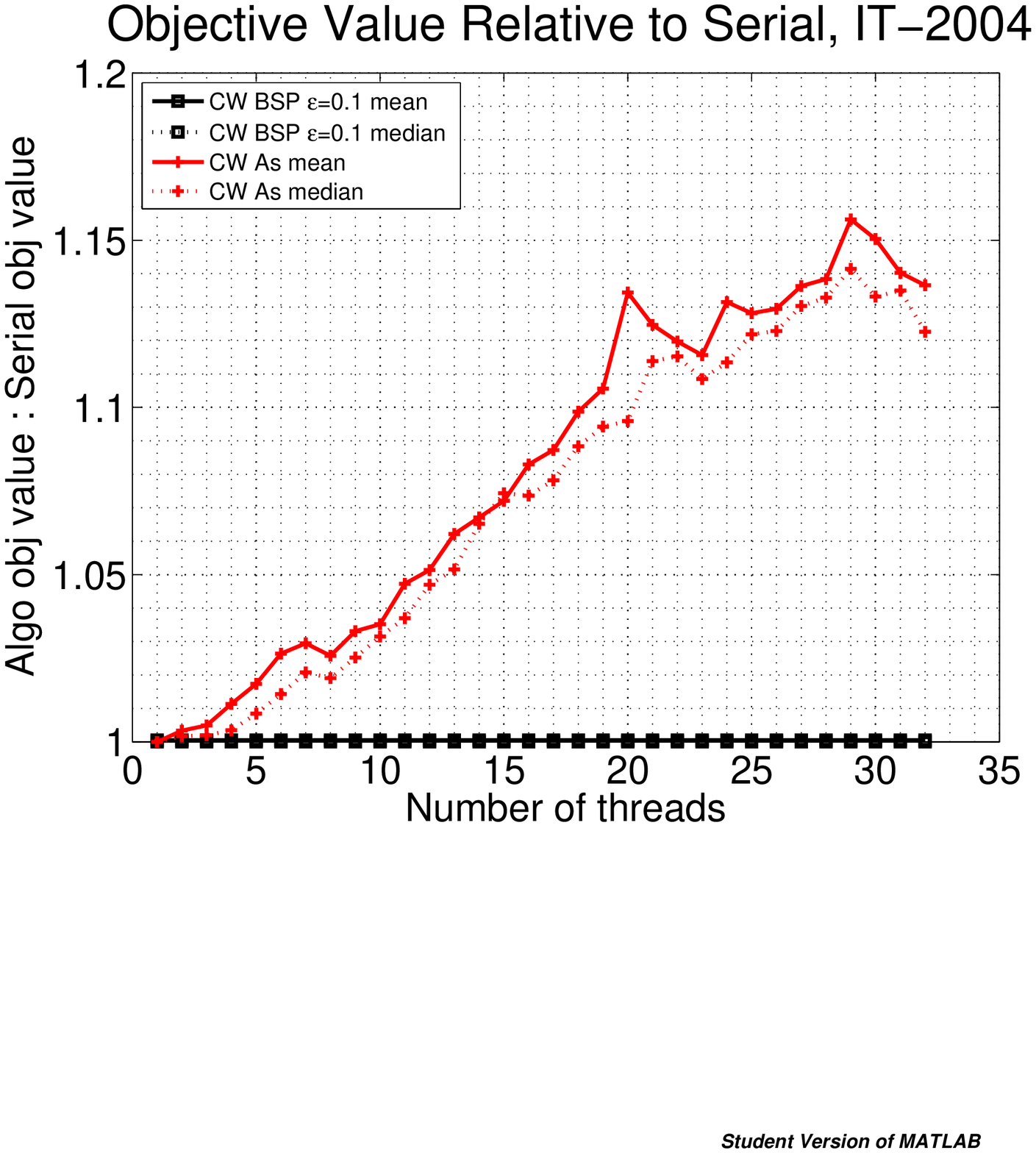}
      \caption{IT-2004, $\epsilon = 0.1$}
      \label{appfig:objvalue_it04_01}
    \end{subfigure} &
    \begin{subfigure}[b]{0.31\textwidth}
      \includegraphics[width=120pt]{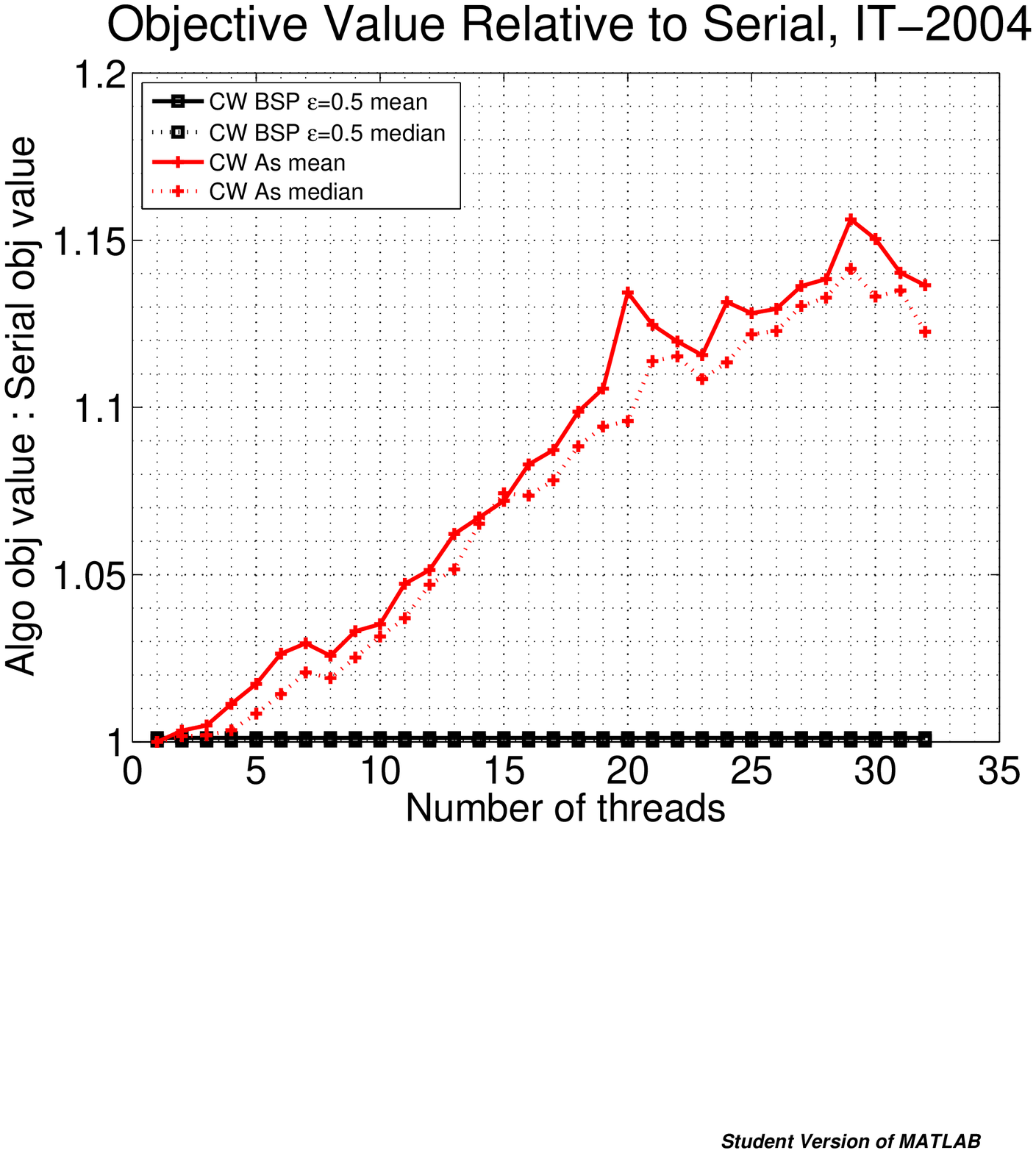}
      \caption{IT-2004, $\epsilon = 0.5$}
      \label{appfig:objvalue_it04_05}
    \end{subfigure} &
    \begin{subfigure}[b]{0.31\textwidth}
      \includegraphics[width=120pt]{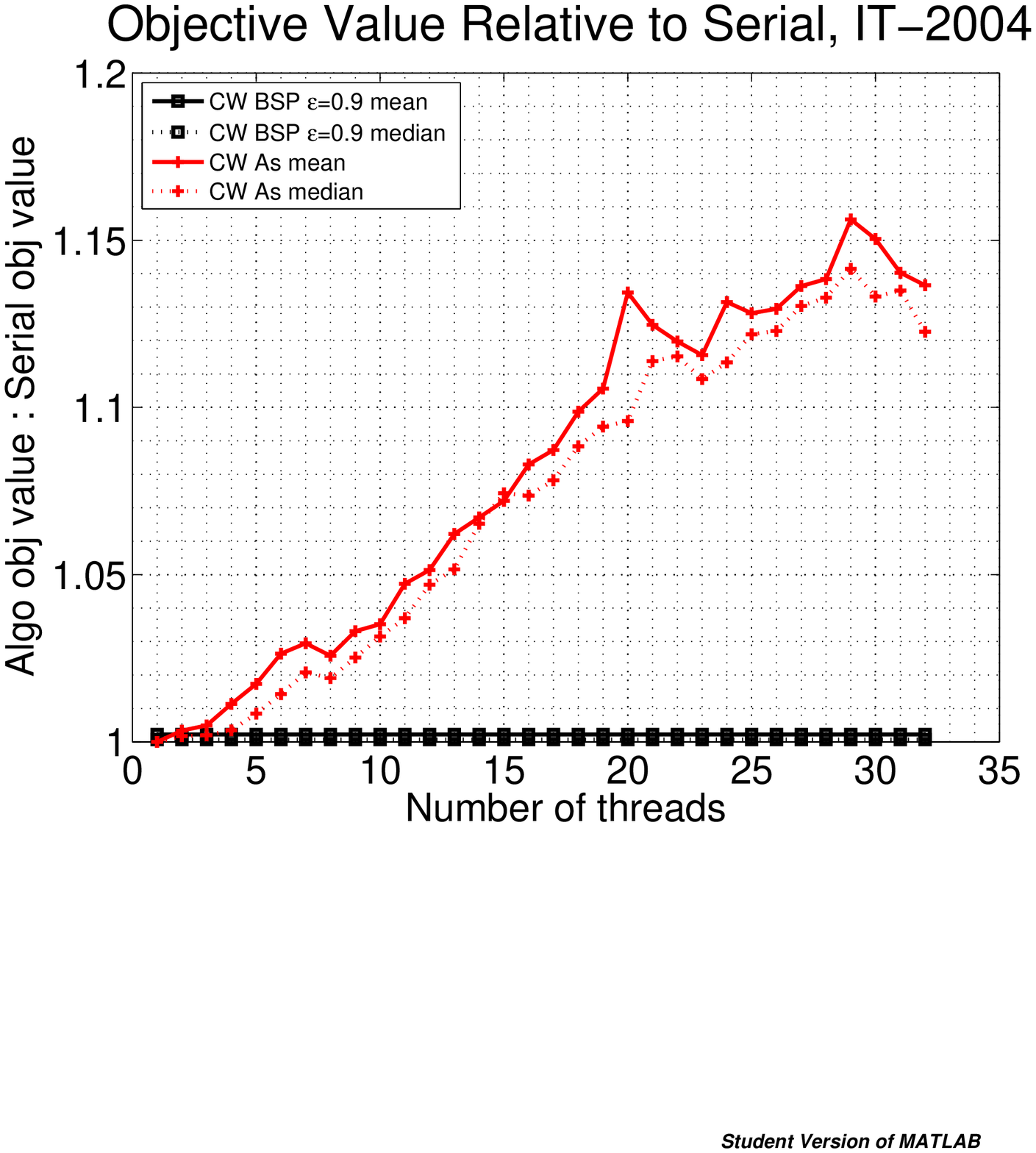}
      \caption{IT-2004, $\epsilon = 0.9$}
      \label{appfig:objvalue_it04_09}
    \end{subfigure} \\

    \begin{subfigure}[b]{0.31\textwidth}
      \includegraphics[width=120pt]{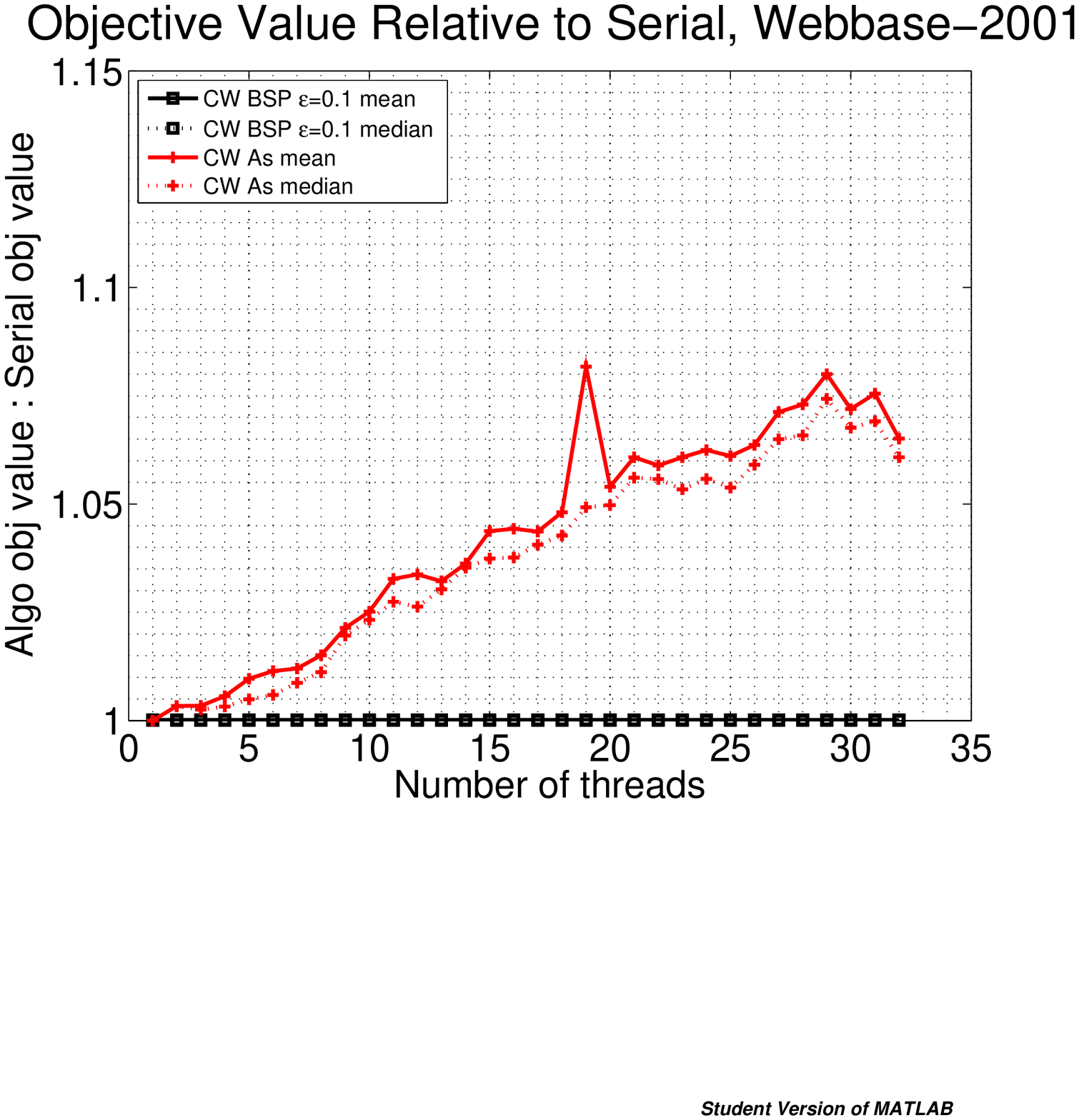}
      \caption{Webbase-2001, $\epsilon = 0.1$}
      \label{appfig:objvalue_wb01_01}
    \end{subfigure} &
    \begin{subfigure}[b]{0.31\textwidth}
      \includegraphics[width=120pt]{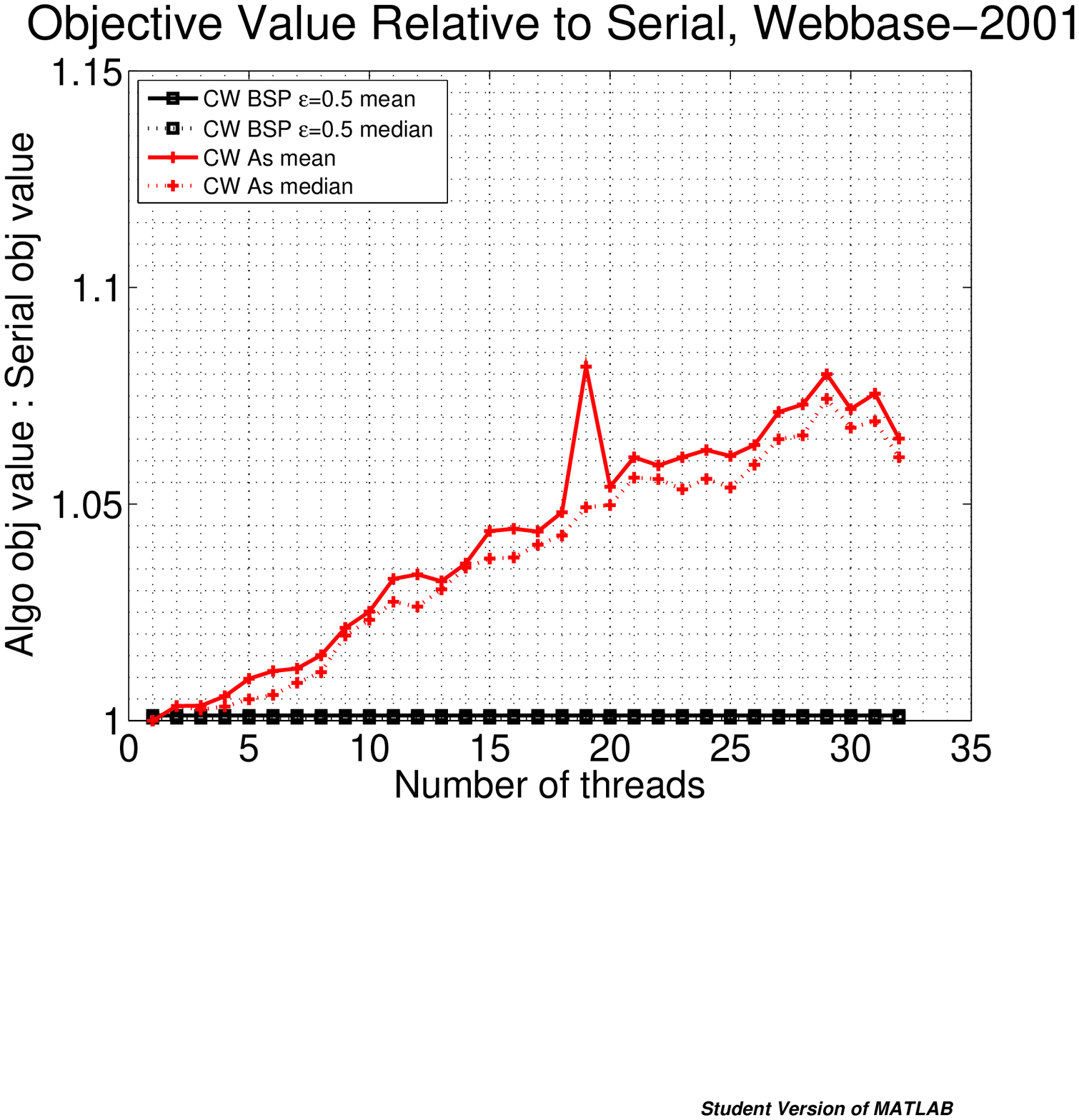}
      \caption{Webbase-2001, $\epsilon = 0.5$}
      \label{appfig:objvalue_wb01_05}
    \end{subfigure} &
    \begin{subfigure}[b]{0.31\textwidth}
      \includegraphics[width=120pt]{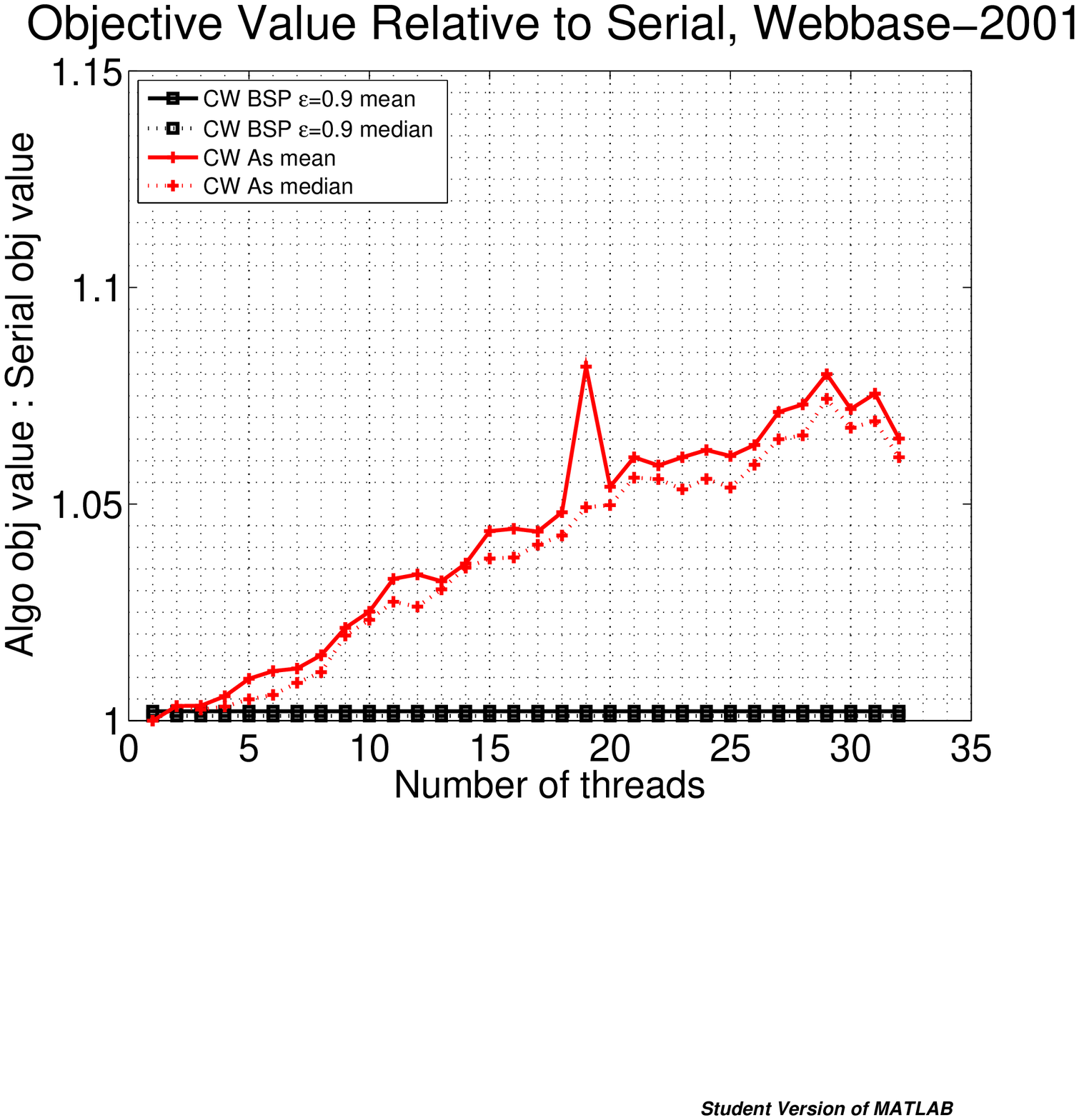}
      \caption{Webbase-2001, $\epsilon = 0.9$}
      \label{appfig:objvalue_wb01_09}
    \end{subfigure} \\

    \begin{subfigure}[b]{0.31\textwidth}
      \includegraphics[width=120pt]{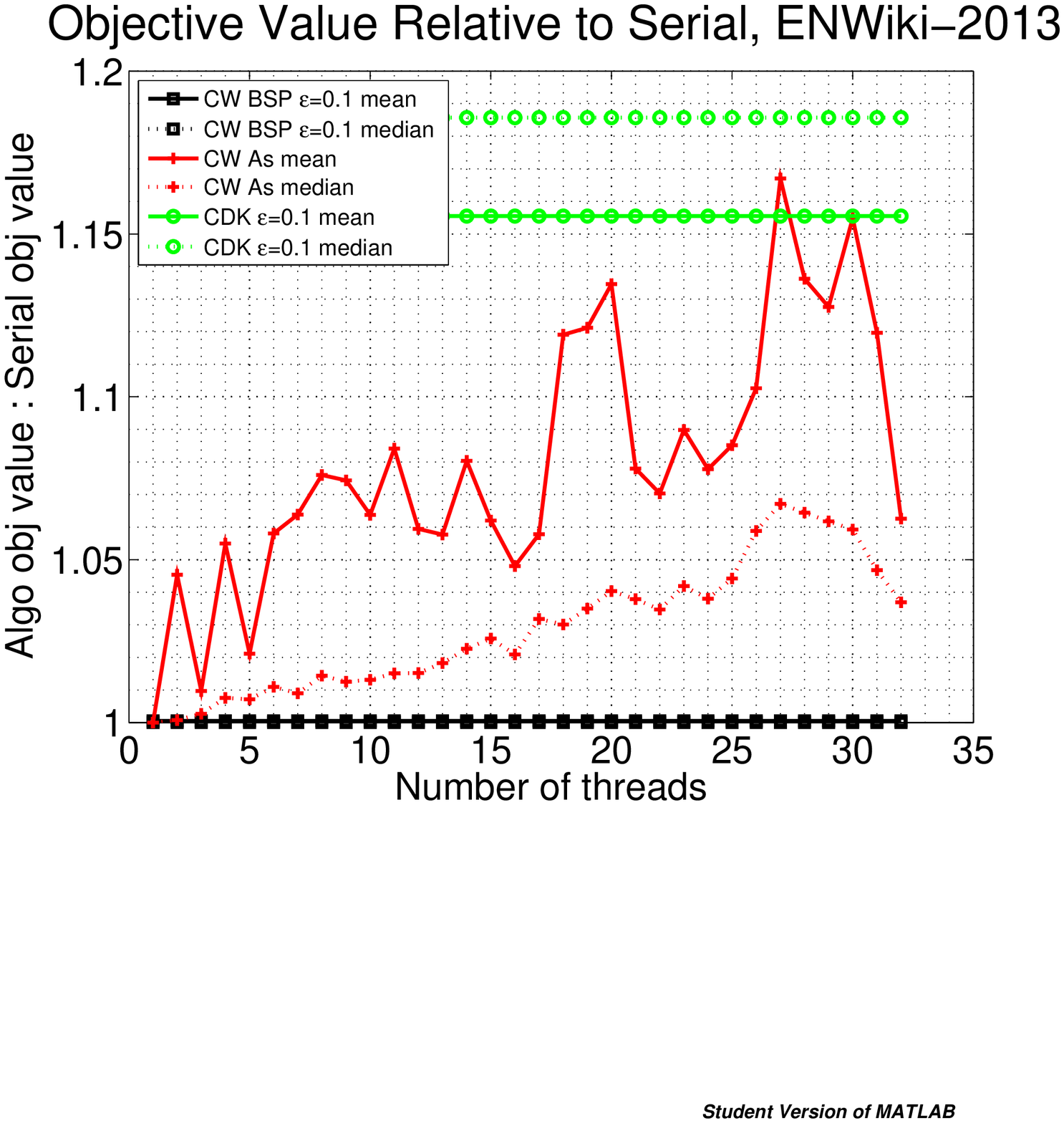}
      \caption{ENWiki-2013, $\epsilon = 0.1$}
      \label{appfig:objvalue_ew13_01}
    \end{subfigure} &
    \begin{subfigure}[b]{0.31\textwidth}
      \includegraphics[width=120pt]{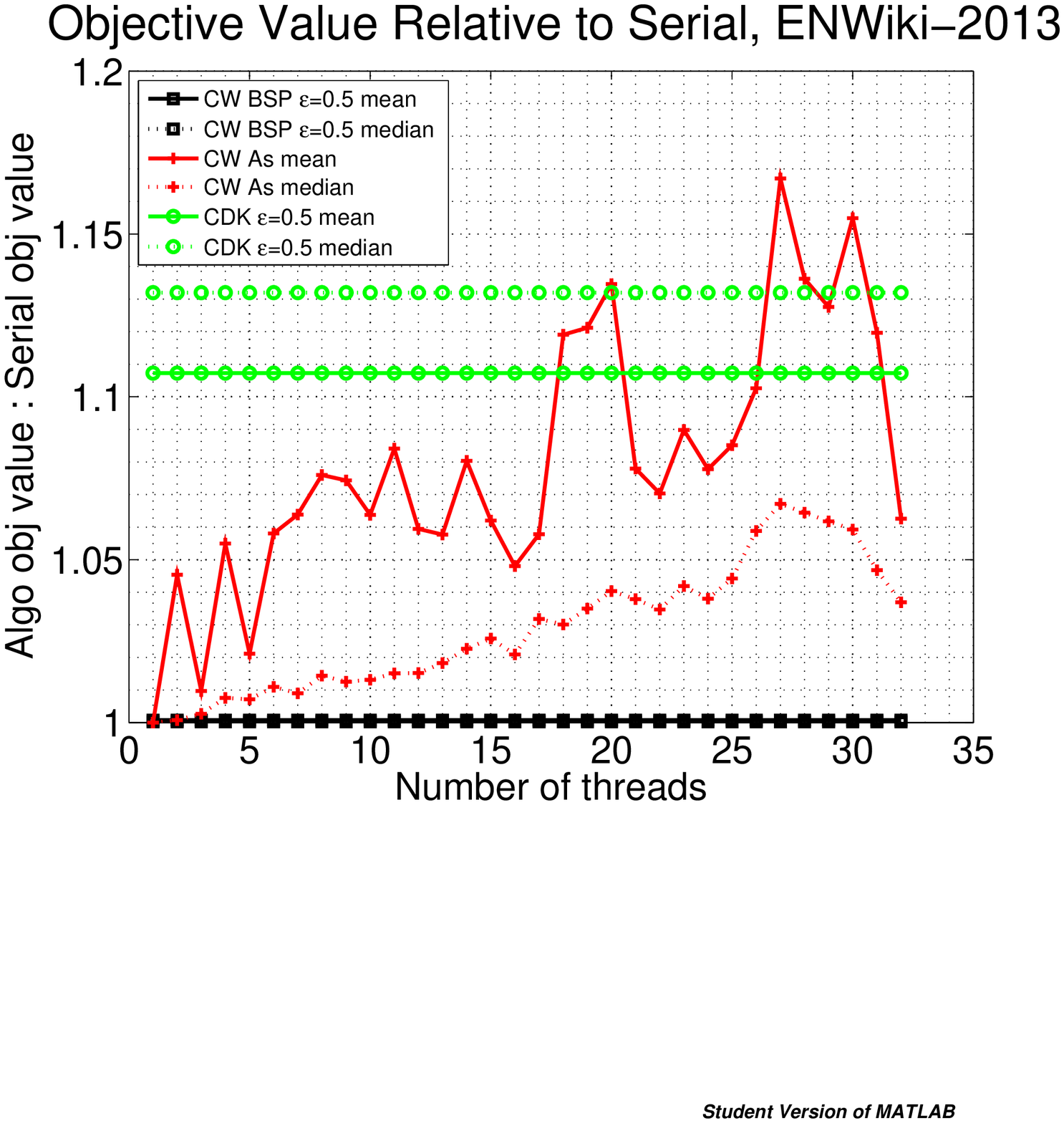}
      \caption{ENWiki-2013, $\epsilon = 0.5$}
      \label{appfig:objvalue_ew13_05}
    \end{subfigure} &
    \begin{subfigure}[b]{0.31\textwidth}
      \includegraphics[width=120pt]{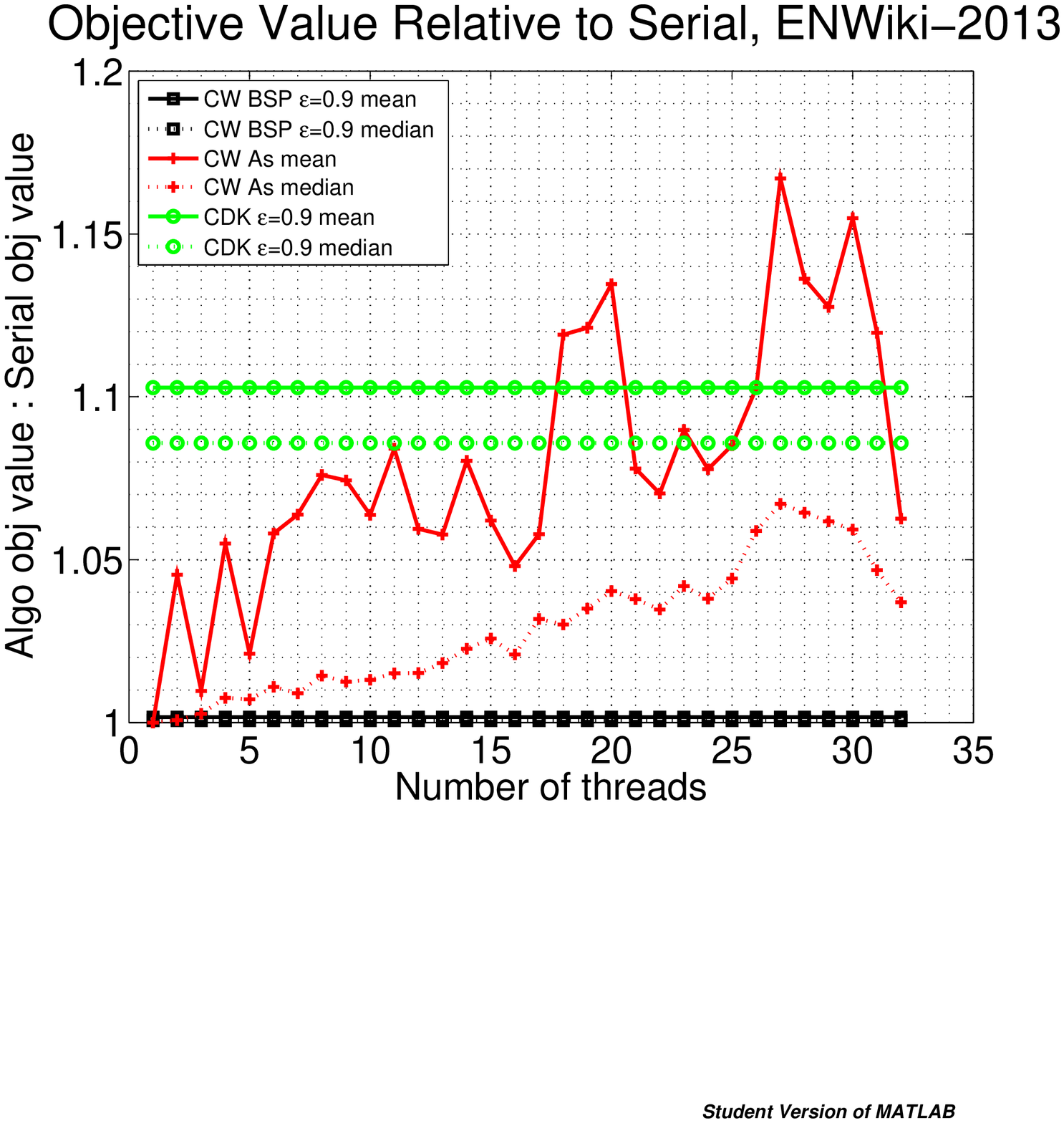}
      \caption{ENWiki-2013, $\epsilon = 0.9$}
      \label{appfig:objvalue_ew13_09}
    \end{subfigure} \\

    \begin{subfigure}[b]{0.31\textwidth}
      \includegraphics[width=120pt]{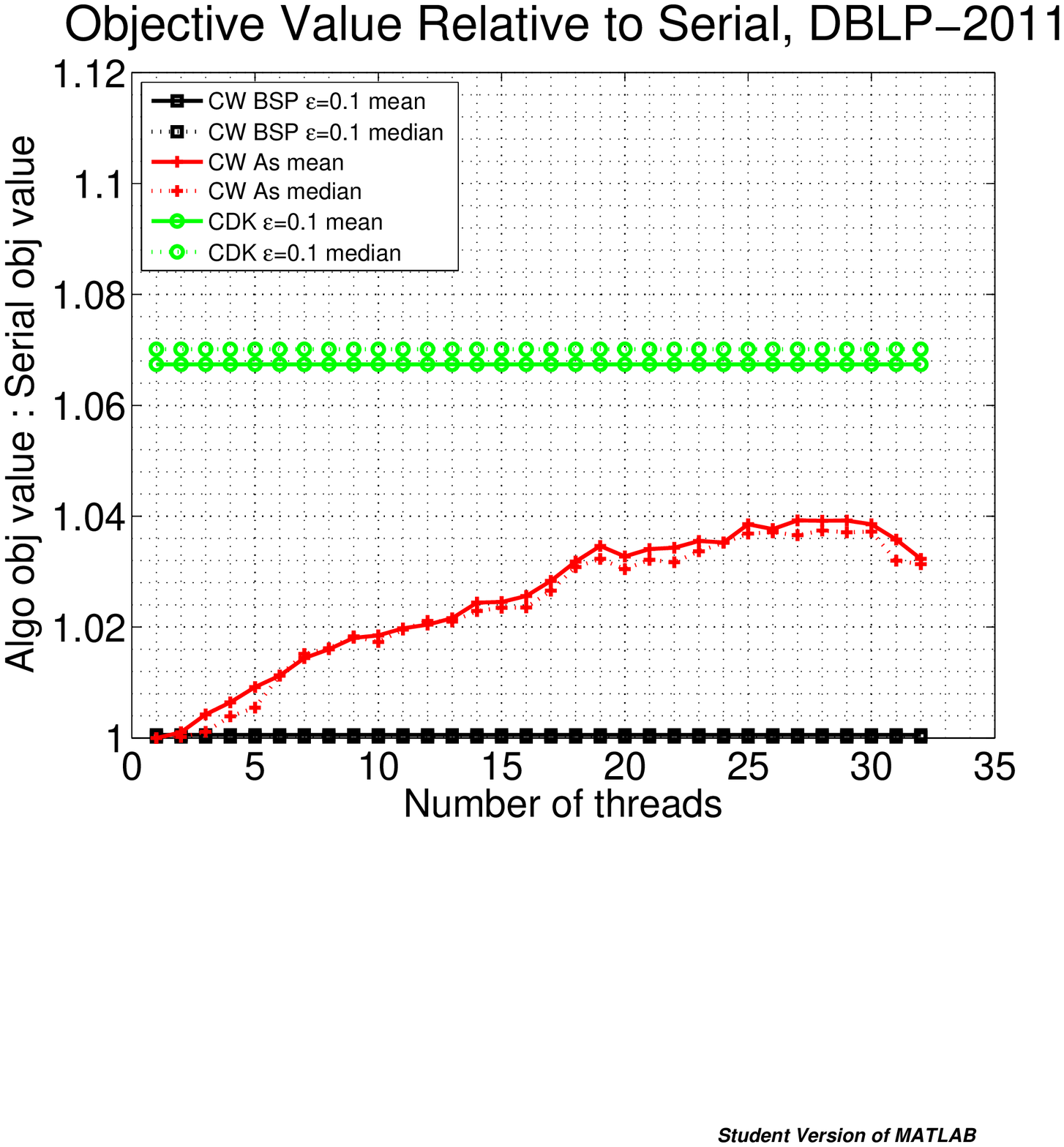}
      \caption{DBLP-2011, $\epsilon = 0.1$}
      \label{appfig:objvalue_db11_01}
    \end{subfigure} &
    \begin{subfigure}[b]{0.31\textwidth}
      \includegraphics[width=120pt]{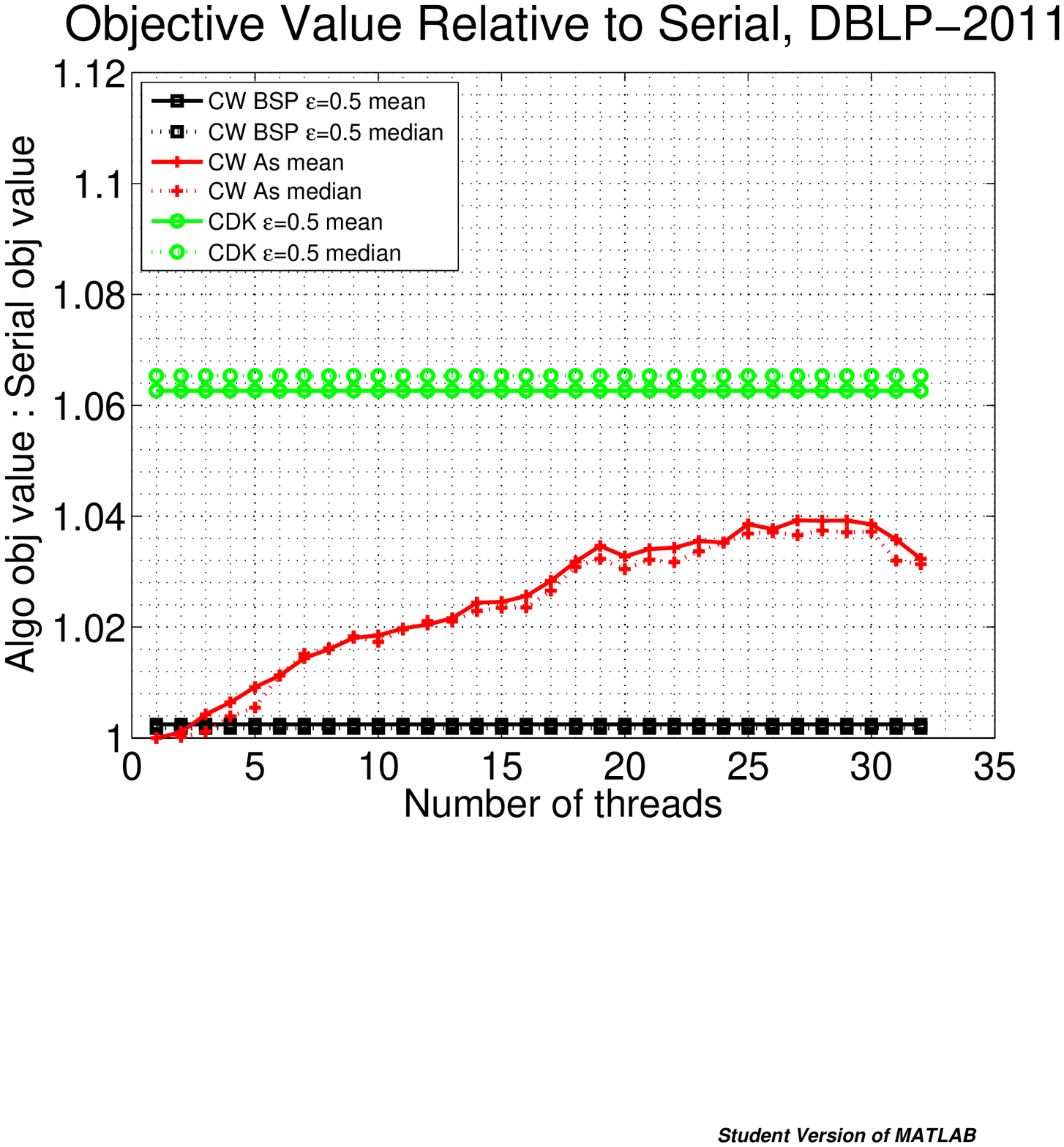}
      \caption{DBLP-2011, $\epsilon = 0.5$}
      \label{appfig:objvalue_db11_05}
    \end{subfigure} &
    \begin{subfigure}[b]{0.31\textwidth}
      \includegraphics[width=120pt]{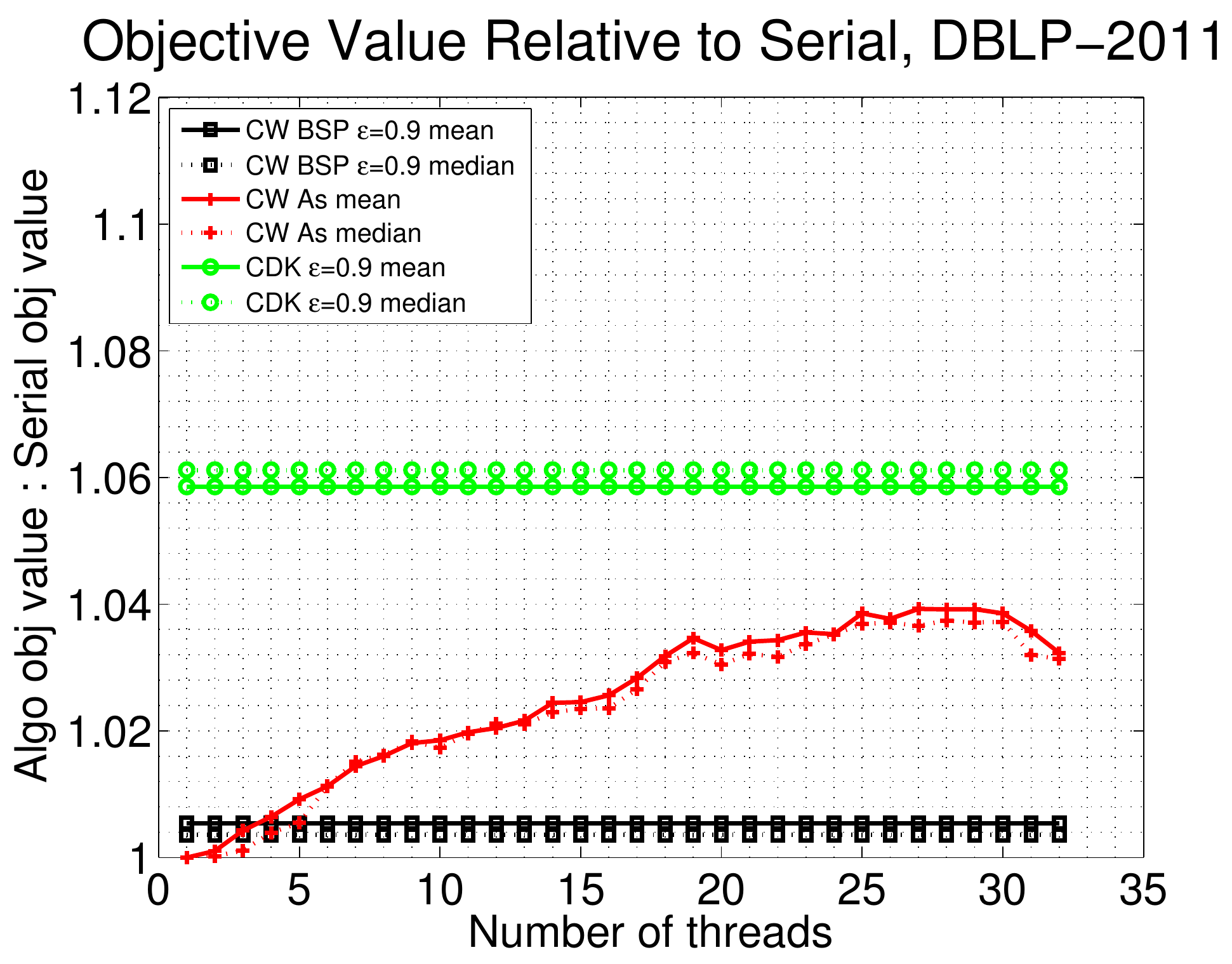}
      \caption{DBLP-2011, $\epsilon = 0.9$}
      \label{appfig:objvalue_db11_09}
    \end{subfigure} \\

  \end{tabular}
  \caption{Empirical objective values relative to mean objective value obtained by serial algorithm.}
\end{figure}

\begin{figure}[ht]
  \centering
  \begin{tabular}{ccc}

    \begin{subfigure}[b]{0.31\textwidth}
      \includegraphics[width=120pt]{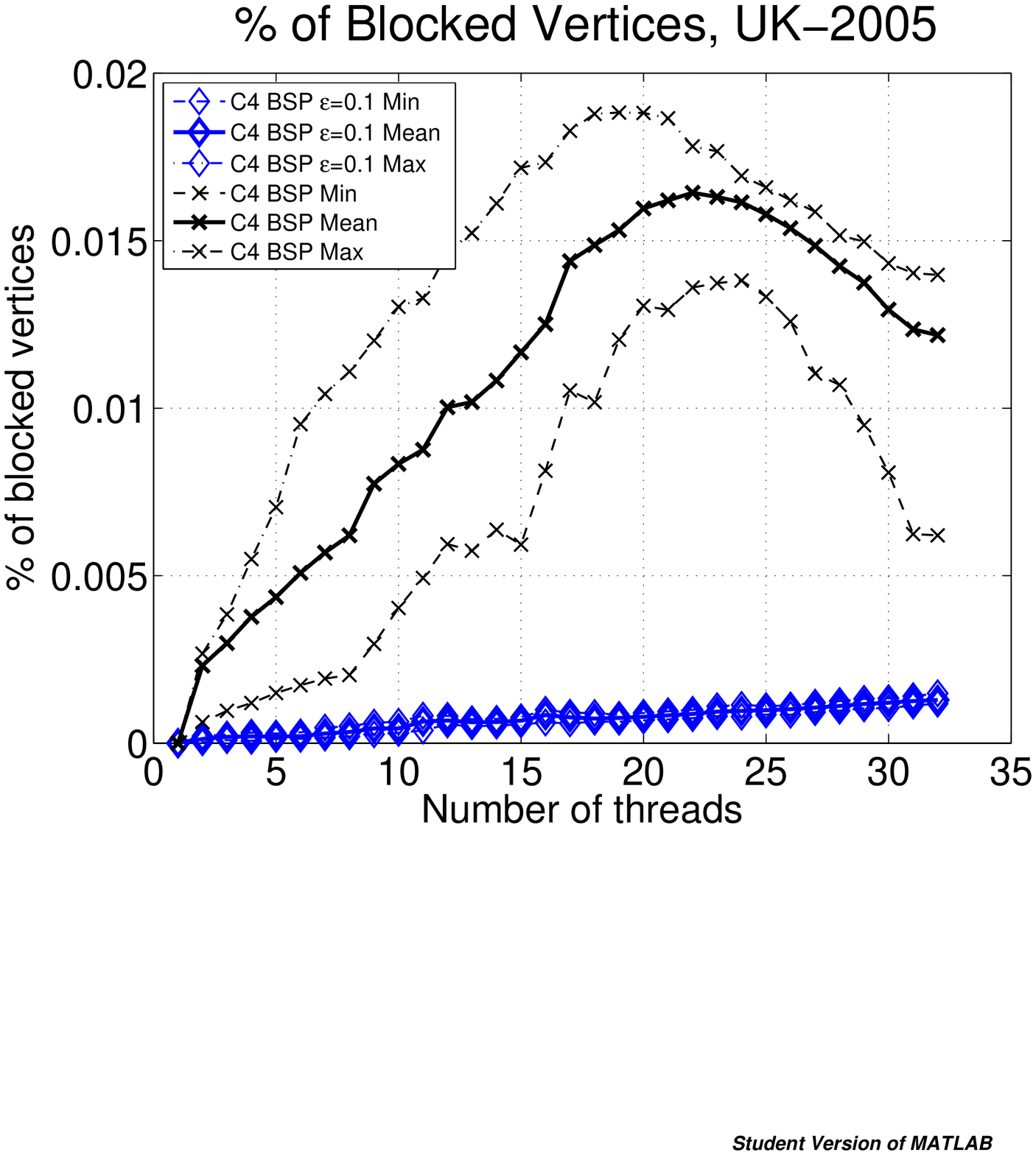}
      \caption{UK-2005, $\epsilon = 0.1$}
      \label{appfig:blktrans_uk05_01}
    \end{subfigure} &
    \begin{subfigure}[b]{0.31\textwidth}
      \includegraphics[width=120pt]{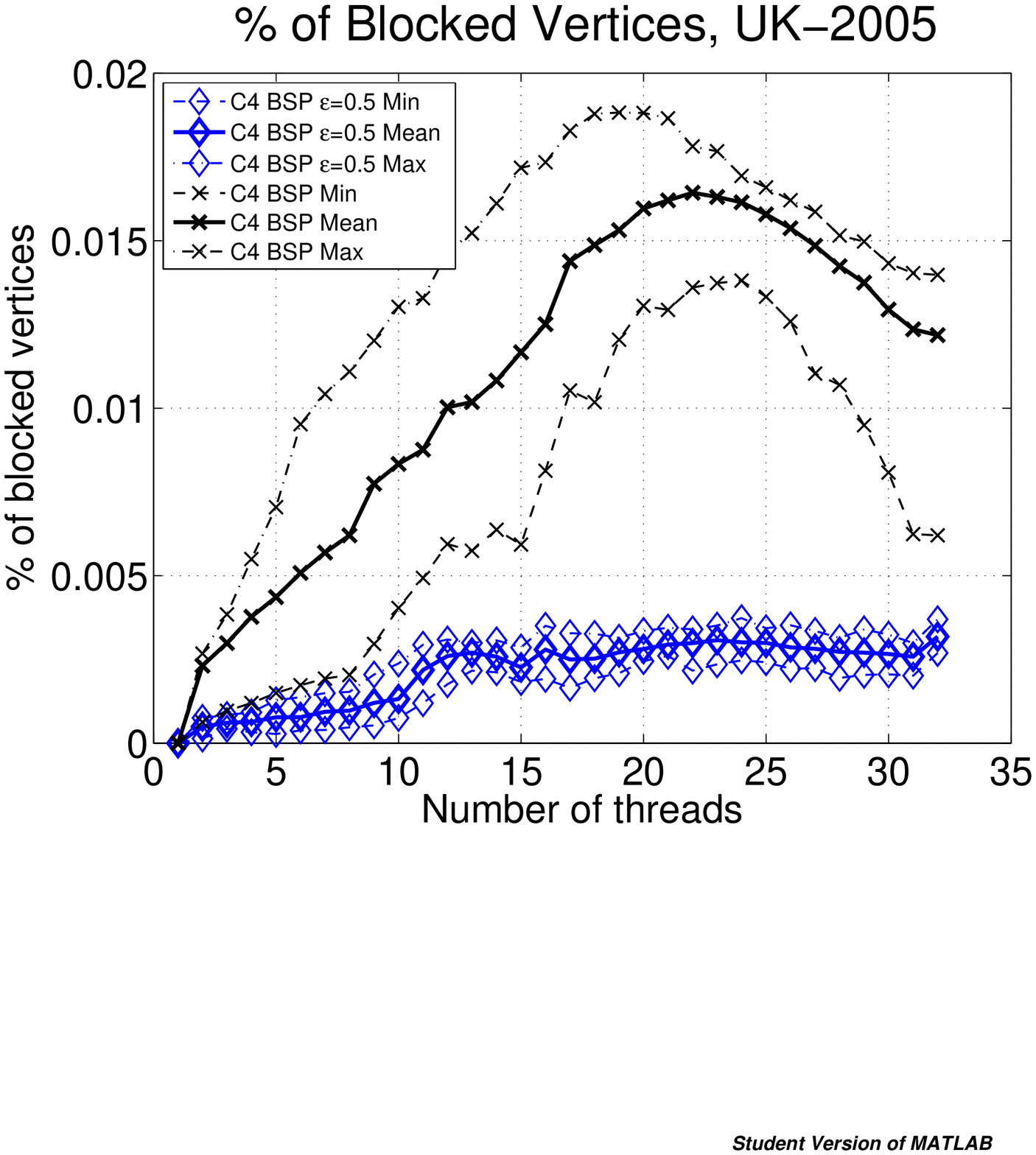}
      \caption{UK-2005, $\epsilon = 0.5$}
      \label{appfig:blktrans_uk05_05}
    \end{subfigure} &
    \begin{subfigure}[b]{0.31\textwidth}
      \includegraphics[width=120pt]{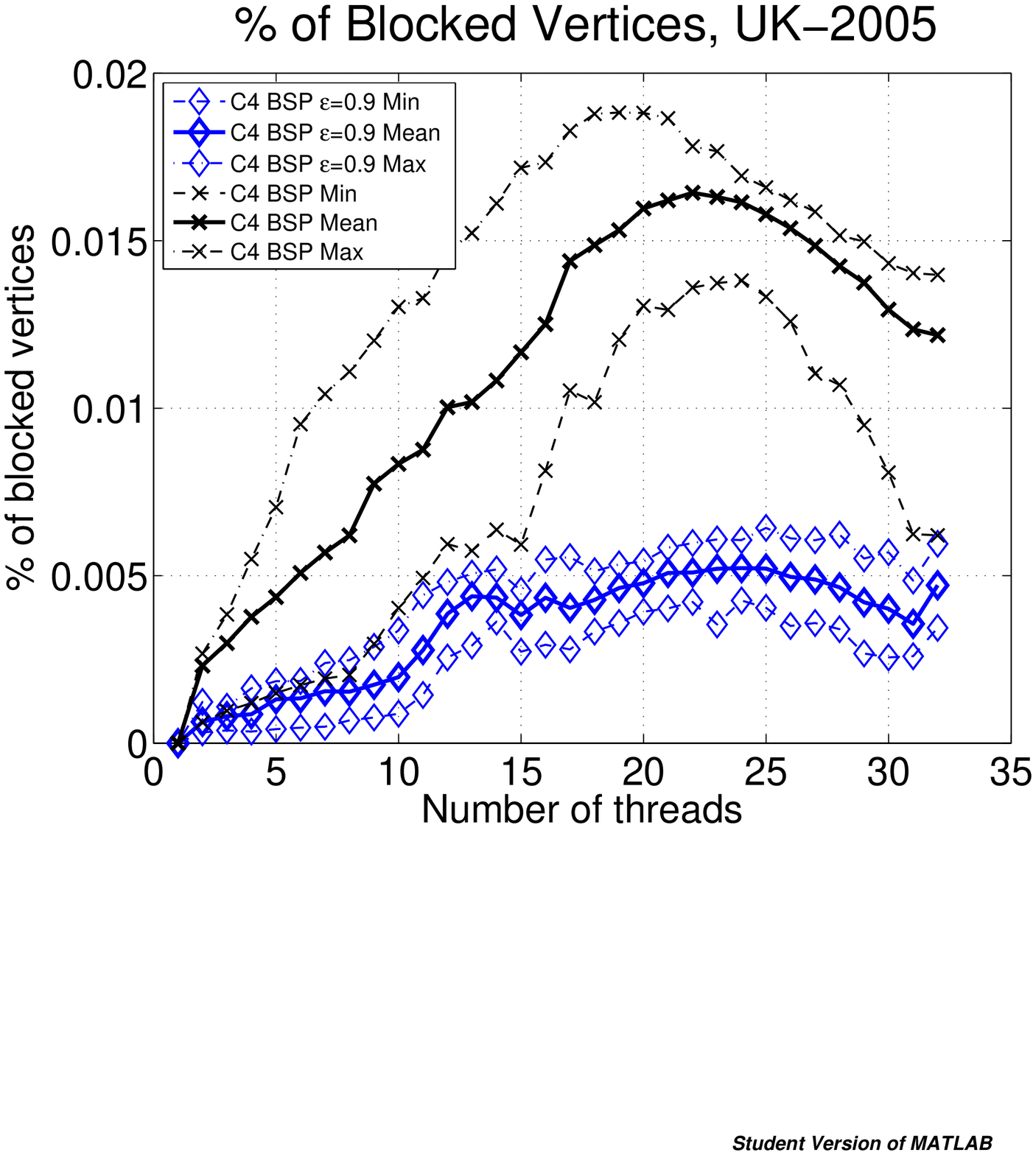}
      \caption{UK-2005, $\epsilon = 0.9$}
      \label{appfig:blktrans_uk05_09}
    \end{subfigure} \\

    \begin{subfigure}[b]{0.31\textwidth}
      \includegraphics[width=120pt]{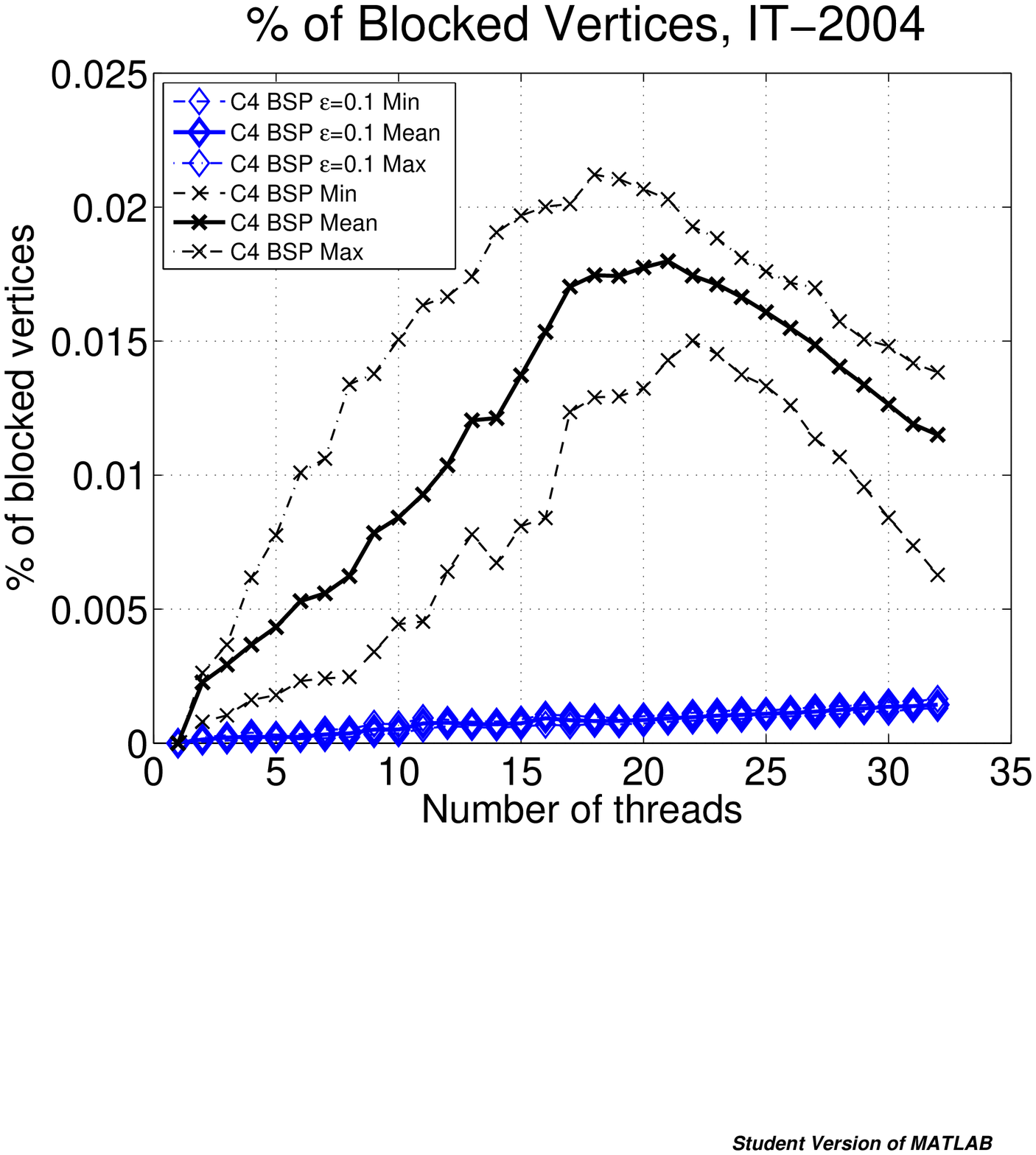}
      \caption{IT-2004, $\epsilon = 0.1$}
      \label{appfig:blktrans_it04_01}
    \end{subfigure} &
    \begin{subfigure}[b]{0.31\textwidth}
      \includegraphics[width=120pt]{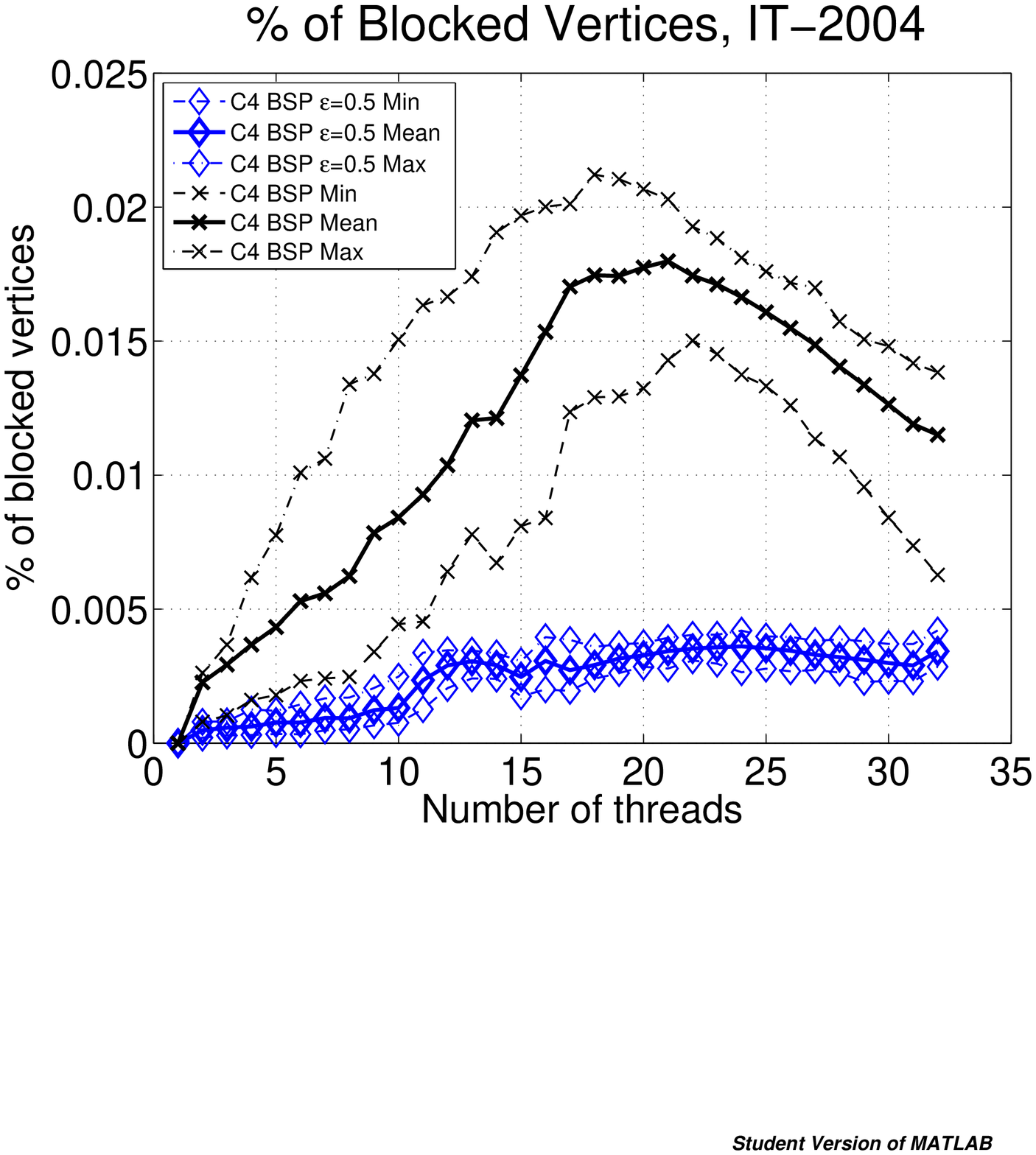}
      \caption{IT-2004, $\epsilon = 0.5$}
      \label{appfig:blktrans_it04_05}
    \end{subfigure} &
    \begin{subfigure}[b]{0.31\textwidth}
      \includegraphics[width=120pt]{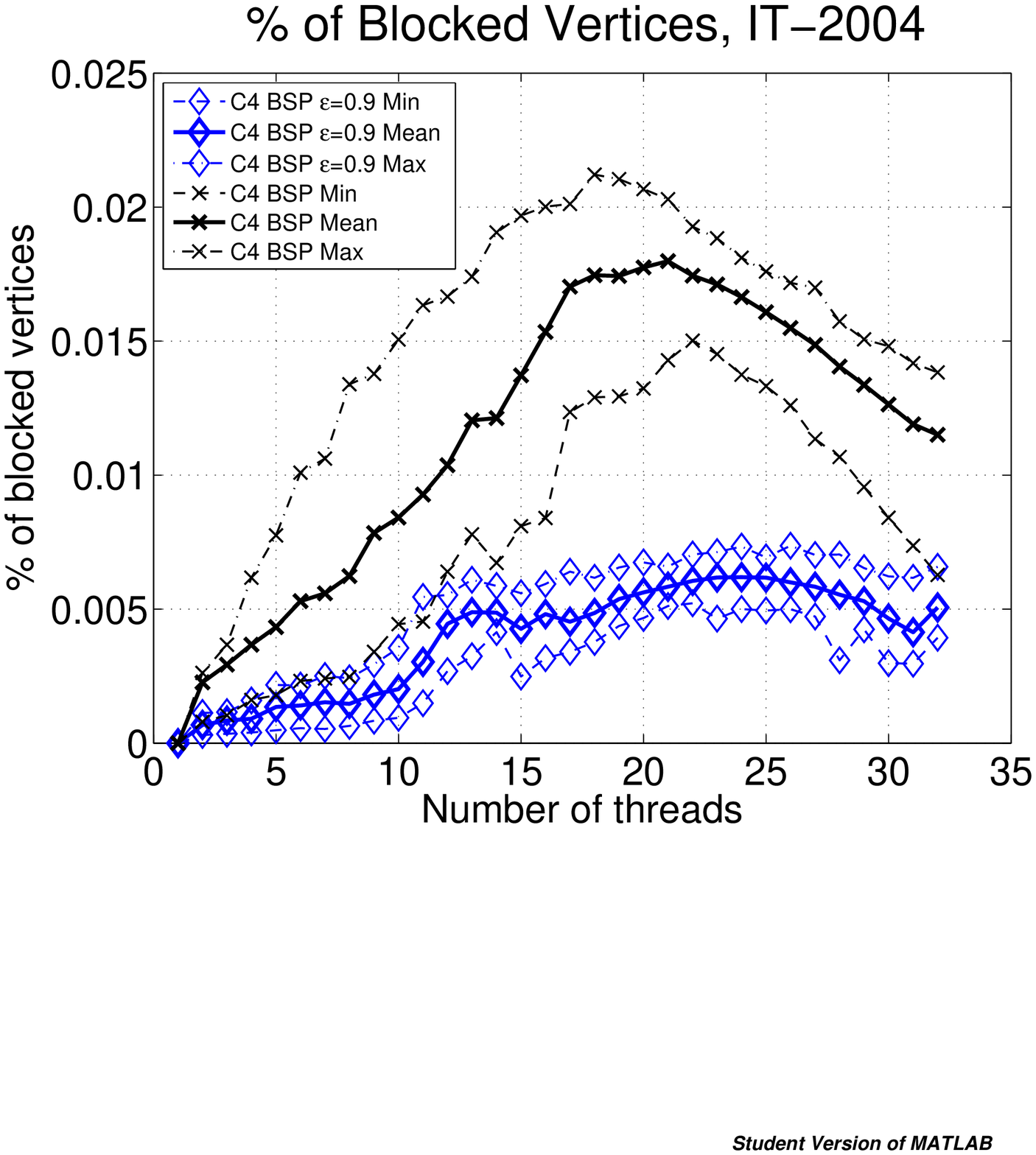}
      \caption{IT-2004, $\epsilon = 0.9$}
      \label{appfig:blktrans_it04_09}
    \end{subfigure} \\

    \begin{subfigure}[b]{0.31\textwidth}
      \includegraphics[width=120pt]{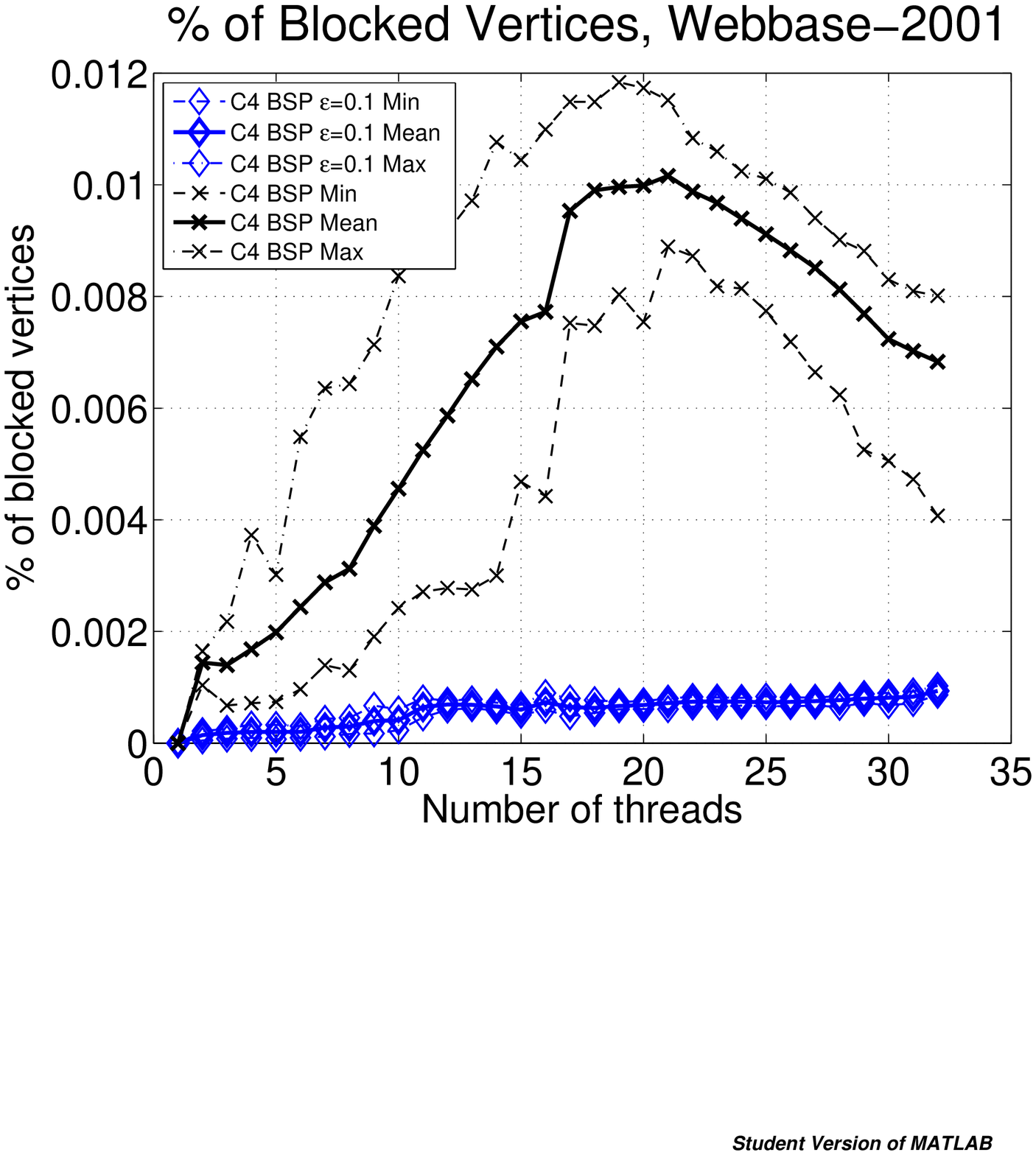}
      \caption{Webbase-2001, $\epsilon = 0.1$}
      \label{appfig:blktrans_wb01_01}
    \end{subfigure} &
    \begin{subfigure}[b]{0.31\textwidth}
      \includegraphics[width=120pt]{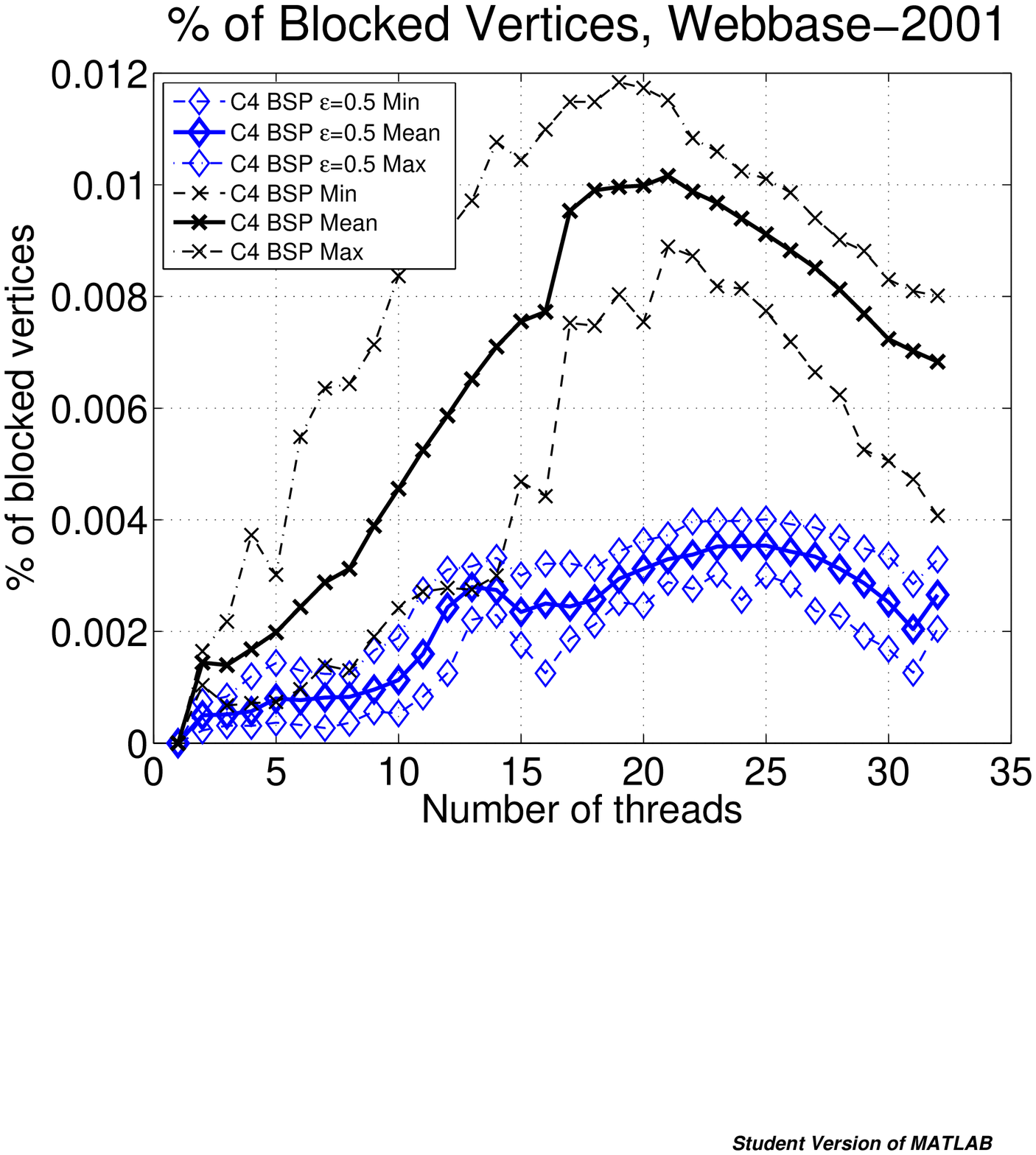}
      \caption{Webbase-2001, $\epsilon = 0.5$}
      \label{appfig:blktrans_wb01_05}
    \end{subfigure} &
    \begin{subfigure}[b]{0.31\textwidth}
      \includegraphics[width=120pt]{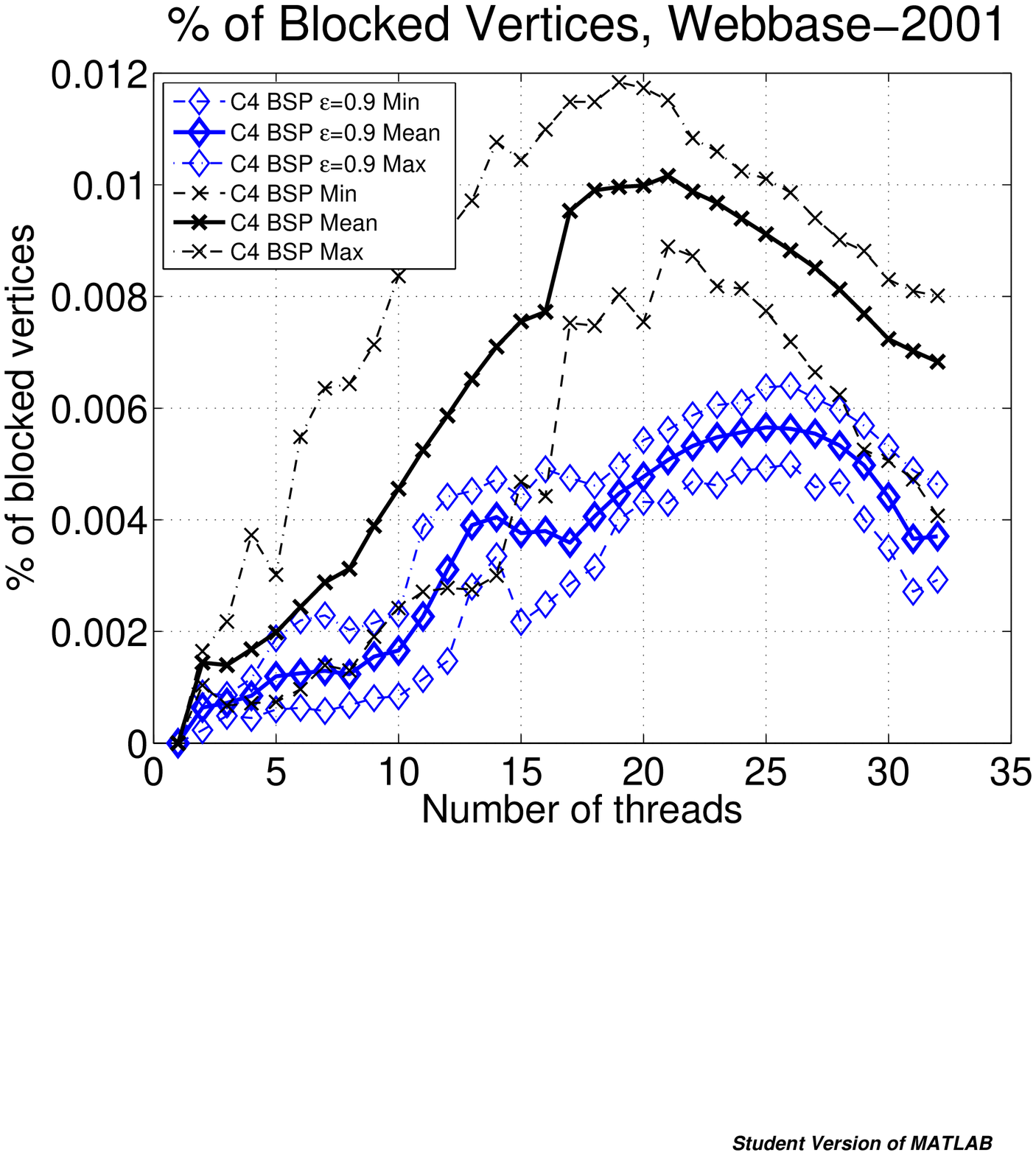}
      \caption{Webbase-2001, $\epsilon = 0.9$}
      \label{appfig:blktrans_wb01_09}
    \end{subfigure} \\

    \begin{subfigure}[b]{0.31\textwidth}
      \includegraphics[width=120pt]{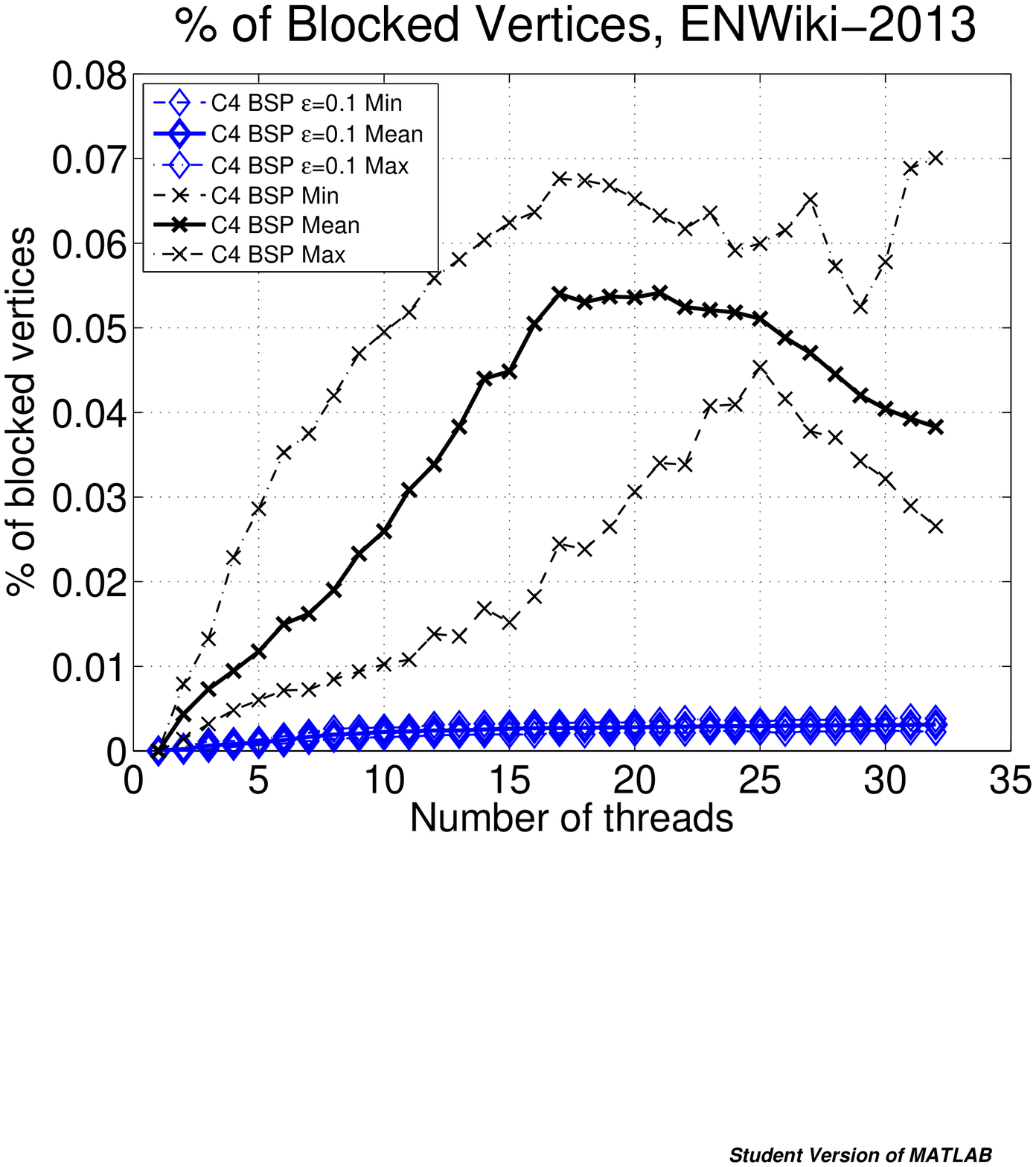}
      \caption{ENWiki-2013, $\epsilon = 0.1$}
      \label{appfig:blktrans_ew13_01}
    \end{subfigure} &
    \begin{subfigure}[b]{0.31\textwidth}
      \includegraphics[width=120pt]{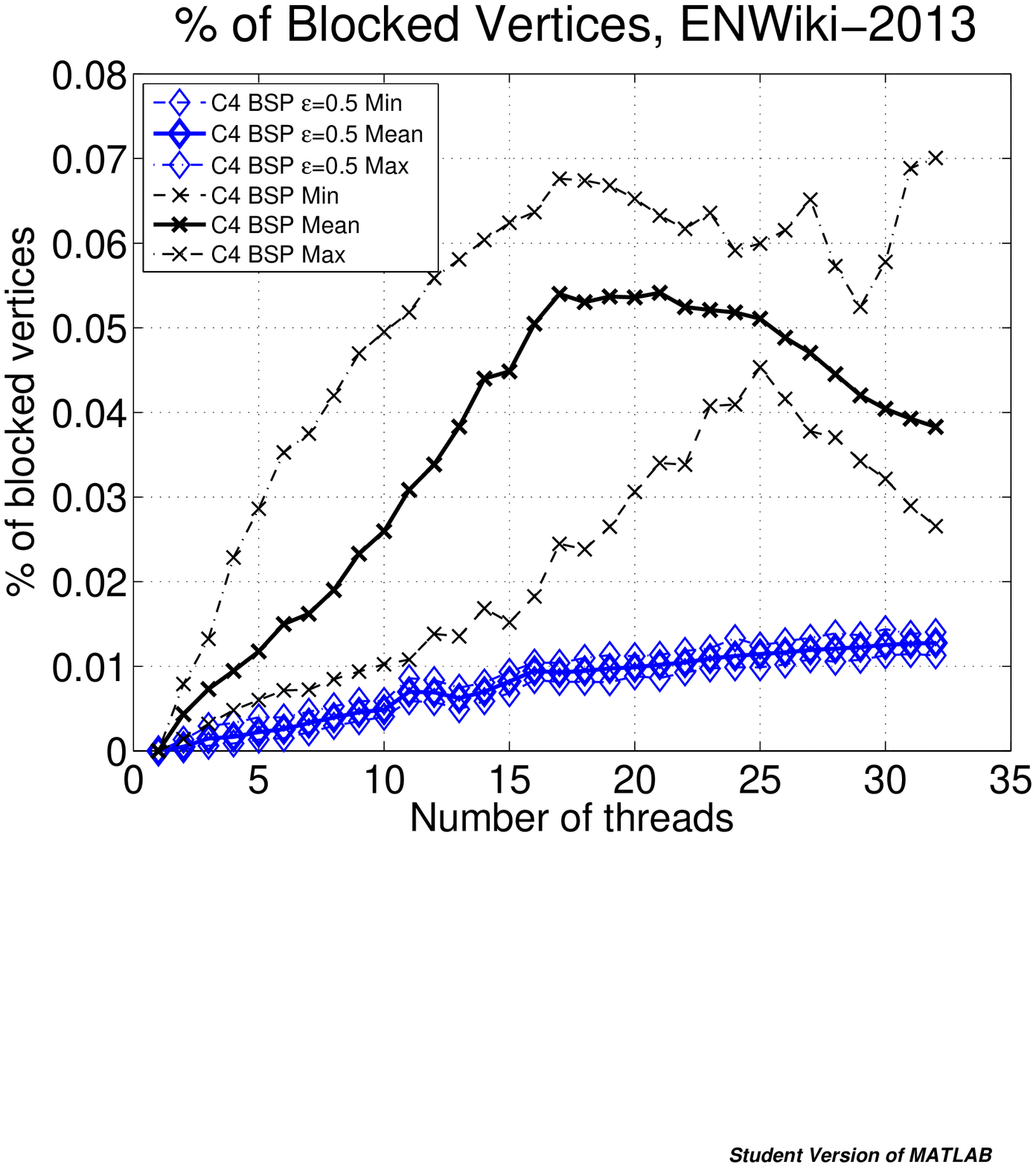}
      \caption{ENWiki-2013, $\epsilon = 0.5$}
      \label{appfig:blktrans_ew13_05}
    \end{subfigure} &
    \begin{subfigure}[b]{0.31\textwidth}
      \includegraphics[width=120pt]{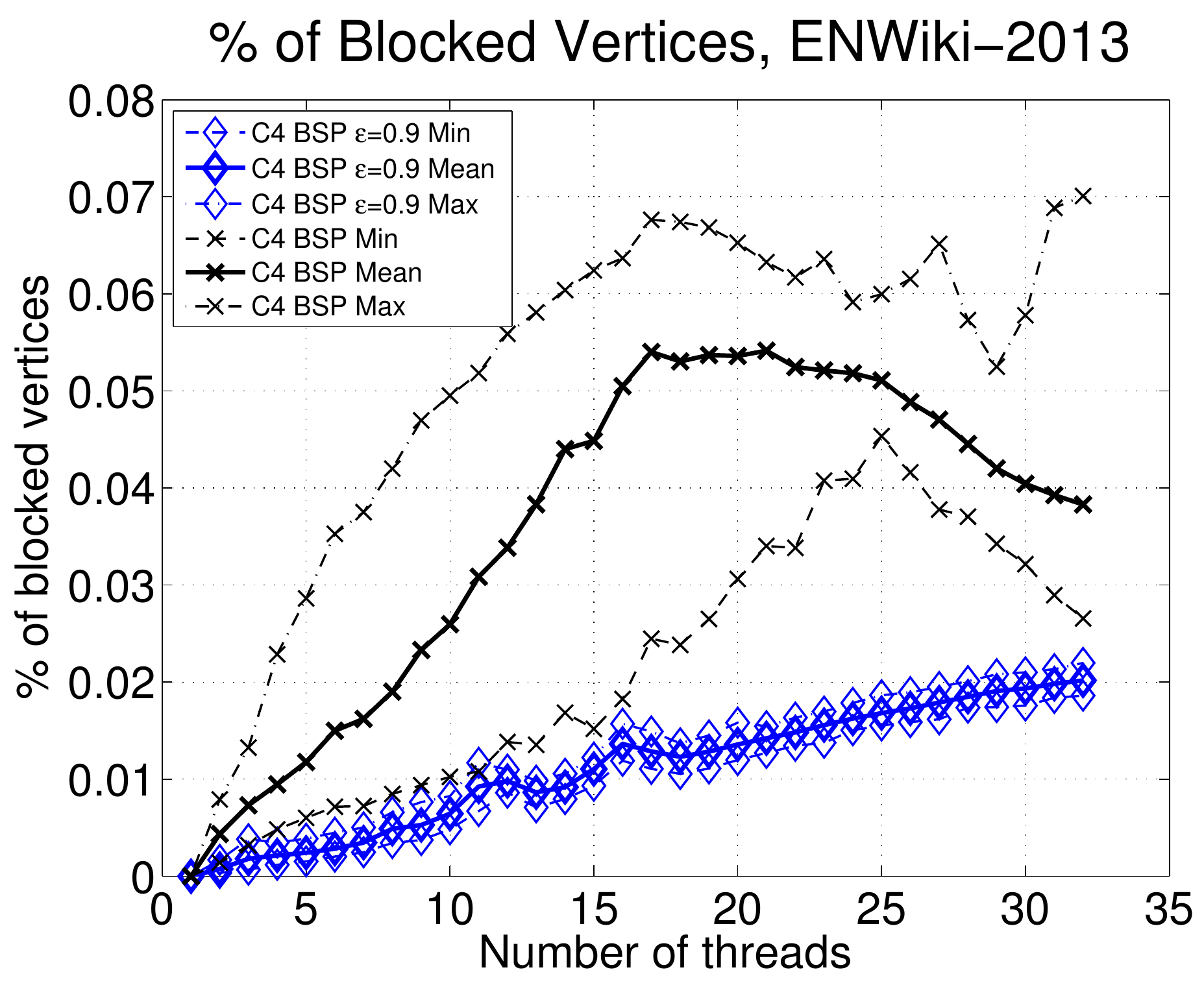}
      \caption{ENWiki-2013, $\epsilon = 0.9$}
      \label{appfig:blktrans_ew13_09}
    \end{subfigure} \\

    \begin{subfigure}[b]{0.31\textwidth}
      \includegraphics[width=120pt]{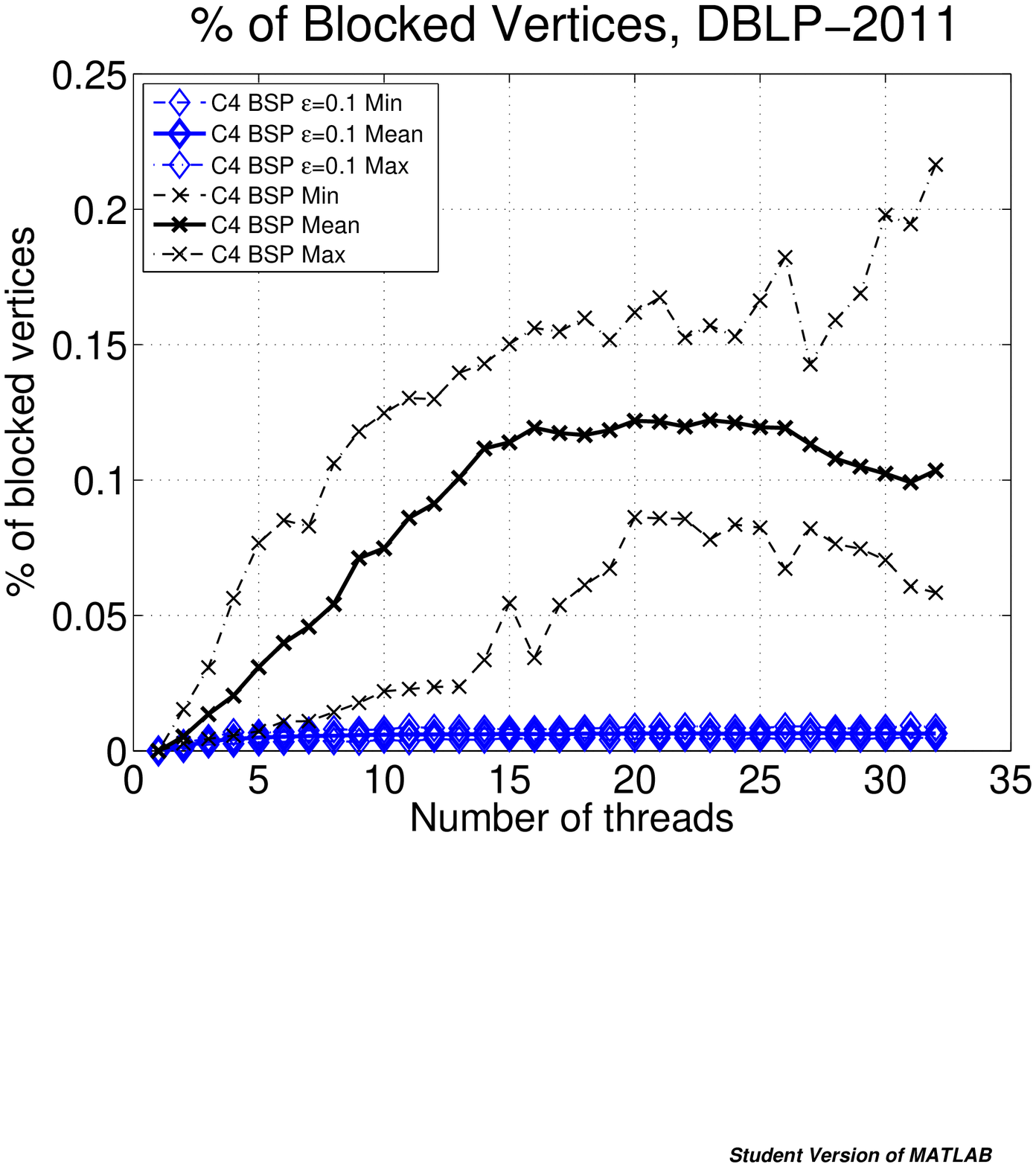}
      \caption{DBLP-2011, $\epsilon = 0.1$}
      \label{appfig:blktrans_db11_01}
    \end{subfigure} &
    \begin{subfigure}[b]{0.31\textwidth}
      \includegraphics[width=120pt]{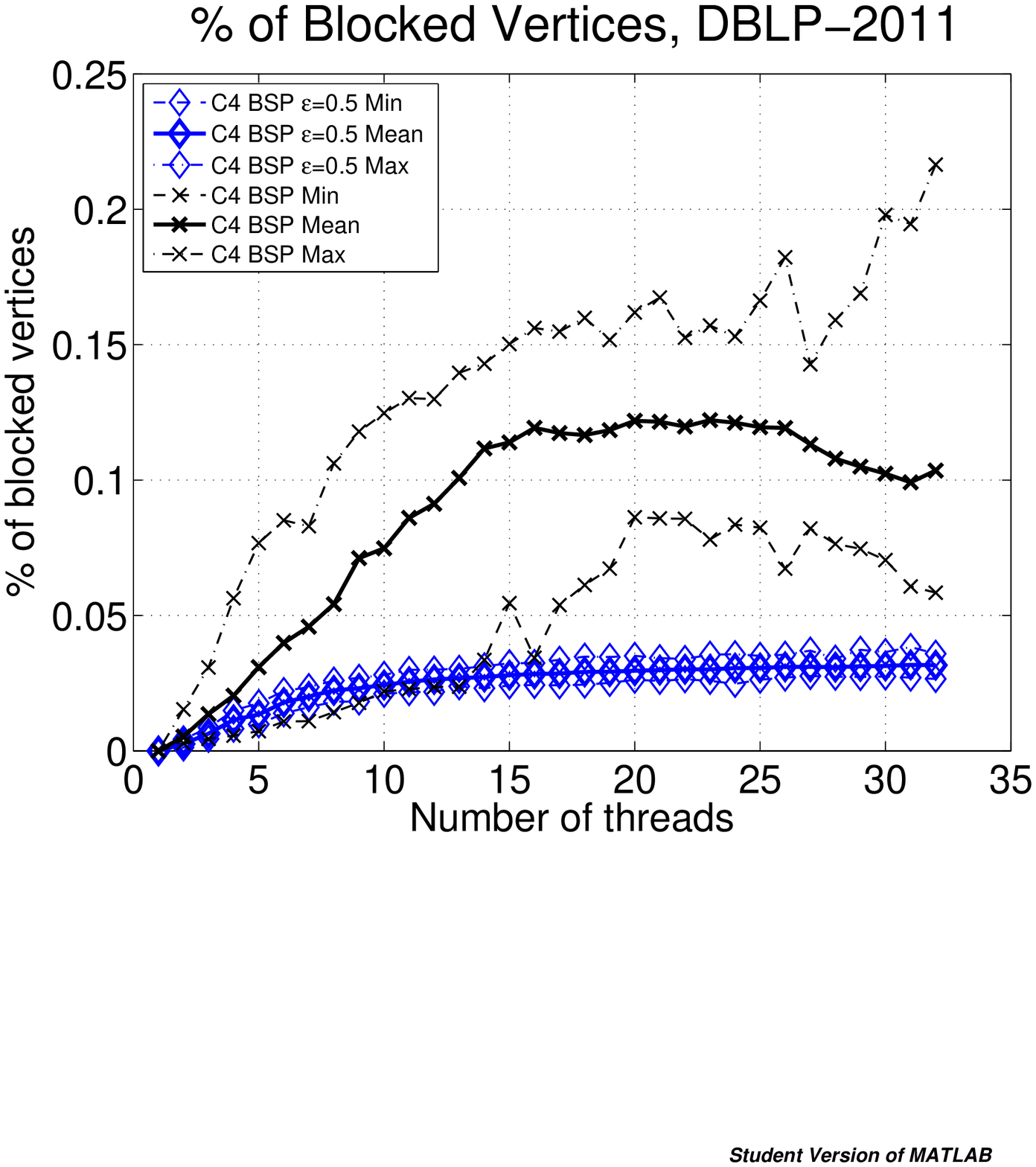}
      \caption{DBLP-2011, $\epsilon = 0.5$}
      \label{appfig:blktrans_db11_05}
    \end{subfigure} &
    \begin{subfigure}[b]{0.31\textwidth}
      \includegraphics[width=120pt]{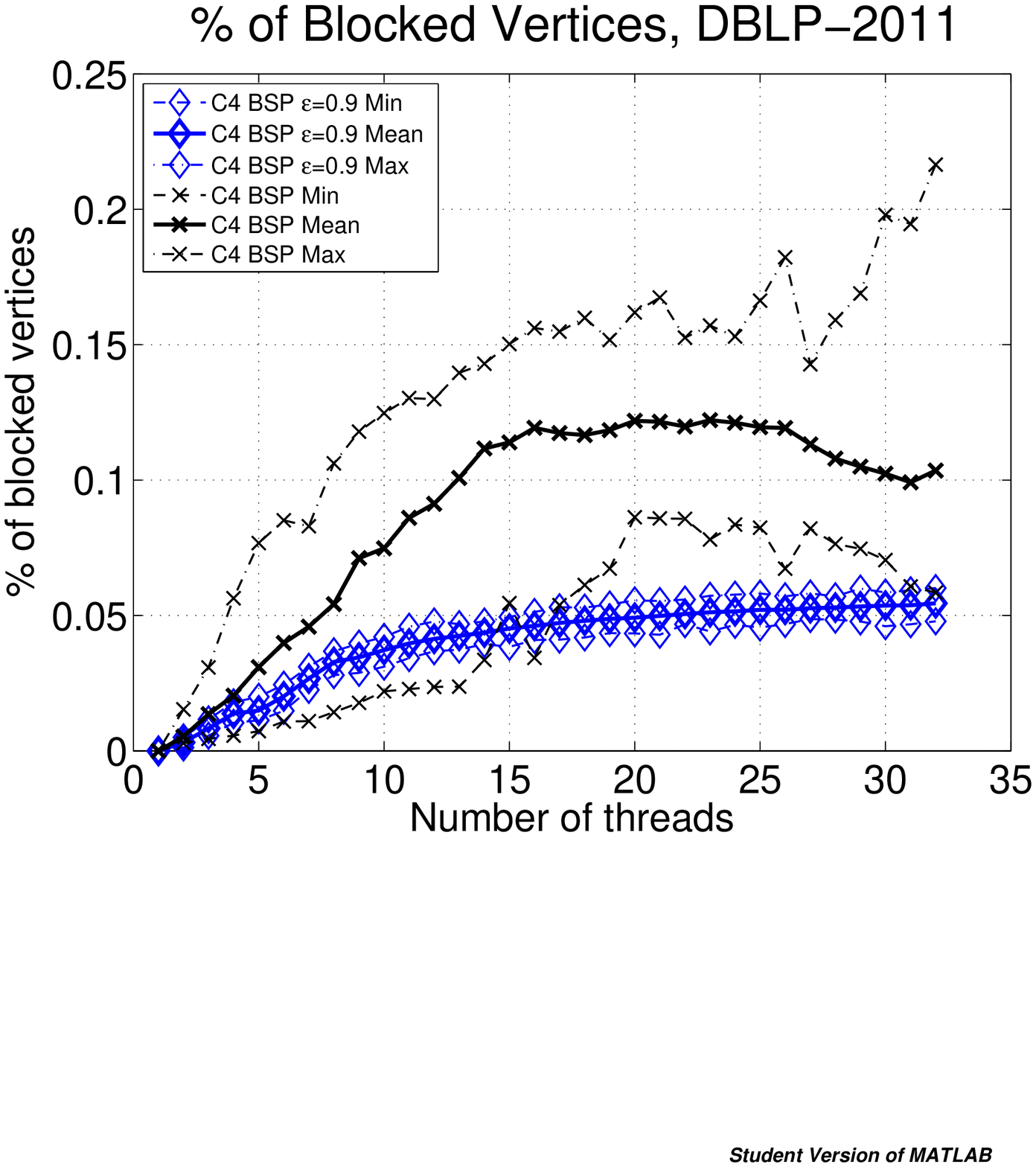}
      \caption{DBLP-2011, $\epsilon = 0.9$}
      \label{appfig:blktrans_db11_09}
    \end{subfigure} \\

  \end{tabular}
  \caption{
  Empirical percentage of blocked vertices.
  Generally the number of blocked vertices increases with the number of threads and larger $\epsilon$ values.
  \CC{} BSP has fewer blocked vertices than asynchronous \CC{}, but at the cost of more synchronization barriers.
  We point out that across all 100 runs of every graphs, the maximum percentage of blocked vertices is less than 0.25\%;
  for large sparse graphs, the maximum percentage is less than 0.025\%, i.e., 1 in 4000.}
\end{figure}